\documentclass[11pt]{article}
\usepackage[margin=1in]{geometry}
\usepackage{amsmath}
\usepackage{graphicx}
\usepackage{enumerate}
\usepackage{natbib}
\usepackage{url} 
\usepackage{mathtools}
\usepackage{empheq}
\usepackage{amsfonts}
\usepackage{amsgen,amsmath,amstext,amsbsy,amsopn,amssymb,amsthm,stmaryrd,bbm,subfigure}
\usepackage{cprotect}
\usepackage{natbib}
\usepackage[dvipsnames]{xcolor}
\usepackage{comment}
\usepackage{enumitem}
\usepackage{multirow}
\usepackage[ruled,vlined,longend]{algorithm2e}

\newcommand{\cN}{\mathcal{N}}

\newcommand{\tr}{\mathrm{tr}}

\newcommand{\RNum}[1]{\uppercase\expandafter{\romannumeral #1\relax}}

\newcommand{\vertiii}[1]{{\left\vert\kern-0.25ex\left\vert\kern-0.25ex\left\vert #1 
		\right\vert\kern-0.25ex\right\vert\kern-0.25ex\right\vert}}

\DeclareMathOperator*{\argmin}{arg\,min}

\newtheorem{theorem}{Theorem}[section]

\newtheorem{lemma}{Lemma}[section]
\newtheorem{proposition}{Proposition}[section]
\newtheorem{assumption}{Assumption}[section]

\newtheorem{remark}{Remark}[section]

\newcommand{\ones}{{{\mathbbm{1}}}}
\newcommand{\ind}[1]{\ones_{\{#1\}}}

\newcommand{\bbR}{\mathbb{R}}
\newcommand{\bbP}{\mathbb{P}}
\newcommand{\bbB}{\mathbb{B}}
\newcommand{\bbE}{\mathbb{E}}
\newcommand{\cF}{\mathcal{F}}

\newcommand{\cL}{\mathcal{L}}
\newcommand{\cG}{\mathcal{G}}
\newcommand{\cR}{\mathcal{R}}
\newcommand{\cC}{\mathcal{C}}
\newcommand{\cK}{\mathcal{K}}
\newcommand{\cS}{\mathcal{S}}
\newcommand{\eg}{\emph{e.g.}}
\newcommand{\MN}{\mathrm{MN}}
\newcommand{\ON}{\mathrm{ON}}
\newcommand{\st}{\mathrm{st}}

\providecommand{\keywords}[1]
{ 
	\small	
	\textbf{\textit{Keywords---}} #1
}

\begin{document}
\title{High-dimensional Multi-class Classification with Presence-only Data}
\author{Lili Zheng\footnotemark[1], Garvesh Raskutti\footnotemark[2]} 
\date{\today}
\maketitle

\begin{abstract}
Classification with positive and unlabeled (PU) data frequently arises in bioinformatics, clinical data, and ecological studies, where collecting negative samples can be prohibitively expensive. While prior works on PU data focus on binary classification, in this paper we consider multiple positive labels, a practically important and common setting. We introduce a {\em multinomial-PU} model and an {\em ordinal-PU} model, suited to unordered and ordered labels respectively. We propose proximal gradient descent-based algorithms to minimize the $\ell_{1,2}$-penalized log-likelihood losses, with convergence guarantees to stationary points of the non-convex objective. Despite the challenging non-convexity induced by the presence-only data and multi-class labels, we prove statistical error bounds for the stationary points within a neighborhood around the true parameters under the high-dimensional regime. This is made possible through a careful characterization of the landscape of the log-likelihood loss in the neighborhood. In addition, simulations and two real data experiments demonstrate the empirical benefits of our algorithms compared to the baseline methods.
\end{abstract}
\footnotetext[1]{Department of Electrical and Computer Engineering, Rice University, Houston, TX, USA}
\footnotetext[2]{Department of Statistics, Department of Computer Science, University of Wisconsin-Madison, Madison, WI, USA}

\section{Introduction}
Positive and unlabeled (PU) data, also referred to as presence-only data, arises in a wide range of applications such as bioinformatics~\citep{elkan2008learning}, ecological modeling of species distribution~\citep{ward2009presence}, text-mining~\citep{liu2003building}, where only a subset of the positive responses are labeled while the rest are unlabeled. For instance, in biological systems engineering~\citep{romero2015dissecting,song2021inferring}, a screening step may output samples of functional protein sequences while all other protein sequences are unlabeled (some being functional and some being nonfunctional); when studying the geographical distribution of species~\citep{ward2009presence}, observed species are certainly present while unobserved ones are not necessarily absent. In these applications, the primary goals are to perform classification and variable selection based on a large number of features/covariates, \eg, learning a model that predicts whether a particular protein sequence would be functional under high temperatures. However, 
directly treating the unlabeled samples as being negative would lead to significant bias and unsatisfactory prediction performance, and hence specific PU-learning methods are in need for the presence-only data.

Prior works~\citep{elkan2008learning,ward2009presence,song2020pulasso} have been proposed to address this problem, with a specific focus on the binary classification setting (one positive category and one negative category). In particular, \cite{ward2009presence,song2020pulasso} propose expectation-maximization (EM) algorithms and majorization-minimization (MM) algorithms that aim to minimize the penalized non-convex log-likelihoods of the observed positive and unlabeled data, accompanied with both statistical and optimization guarantees under the high-dimensional setting \citep{song2020pulasso}. However many PU-learning problems have multiple classes. For instance, in protein engineering,  scientists are interested in whether certain protein sequences would be functional under several levels of temperatures~\citep{fields2001protein}; in disease diagnosis, doctors can give a score representing the severity/stage of the patient's disease development or the type of the disease~\citep{bertens2016nomogram,sivapriya2015ensemble}; in recommender systems, the users may have multiple ways to interact with a recommendation. Considering multi-class classification can largely \emph{increase the model capacity, make more efficient use of the data at hand, and lead to more informative predictions}. On the other hand, existing statistical methods either assume that all observations reflect the true labels and lead to bias, or they consider only binary responses and provide coarse predictions, 
calling for new models and algorithms that take into account both the nature of the presence-only data and the multi-class labels. 

However, dealing with multi-class classification in the PU-setting is a non-trivial task and is not a simple extension of prior work on binary classification. The multiple categories can be unordered when they represent the species, or they can be ordered when they encode the level of temperature each protein sequence can tolerate, calling for two different modeling approaches each of which present unique technical challenges. In this paper, we address this problem by leveraging ideas from the multinomial and ordinal logistic regressions and extending the prior work~\citep{song2020pulasso} on high-dimensional classification for positive and unlabeled data with binary responses. We propose two models (the multinomial-PU model and the ordinal-PU model) and corresponding algorithms for unordered and ordered categorical responses respectively. As we will see in later sections, the non-convex landscape of the log-likelihood losses of PU-data becomes much more complicated in the multi-class setting. Nevertheless, we provide theoretical guarantees for our approaches under both models based on a careful characterization of the landscape of the log-likelihood losses, and demonstrate their empirical merits by simulations and real data examples. 
\subsection{Problem Formulation}\label{sec:problem-form}
\paragraph{Notation:} For any matrix or tensor $A\in \bbR^{p_1\times \dots\times p_k}$ and $1\leq l\leq k$, let $\|A\|_2=(\sum_{i_1,\dots,i_k}A_{i_1,\dots,i_k}^2)^{\frac{1}{2}}.$ 
For any two tensors $A$ and $B$ of the same dimension, let $\langle A, B\rangle$ denote the Euclidean inner product of $A$ and $B$. For any matrix $A\in \bbR^{p_1\times p_2}$, we let $\lambda_{\min}(A)$ denote the smallest eigenvalue of $A$, and let $A_j\in \bbR^{p_2}$ be the $j$th column of $A$.
Let
$\ind{E} = {\footnotesize \begin{cases}1,& \text{if } E \text{ true} \\ 0,& \text{else}\end{cases}}$ be the indicator function.

In this section, we propose our multinomial-PU model and ordinal-PU model which are suited to unordered and ordered categorical responses, respectively. 
For both models, let $x\in\mathbb{R}^{p}$ be the covariate, $y\in\{0,1,\dots,K\}$ be the true categorical response with $K$ positive categories that is not observed. We assume observing only the covariate $x$ and a label $z\in\{0,1,\dots,K\}$, a noisy observation of the true response $y$. If the label $z \neq 0$, then it reflects the true response $y = z$; otherwise, this sample is unlabeled and the unknown response $y$ can take any value from $\{0,1,\dots,K\}$. The assignment of label $z$ is conditionally independent of the covariate $x$ given the response $y$: that is, the true response $y$ is randomly missing in observation with probability only depending on $y$, an assumption commonly seen in the literature on missing data problems~\citep{little2019statistical}.

More specifically, to model the conditional distribution of $z$ given $y$, we consider two common settings: the case-control approach~\citep{lancaster1996case,ward2009presence,song2020pulasso} and the single-training-set scenario~\citep{elkan2008learning}. 
The {\bf case-control} setting is suited to the case when the unlabelled and positive samples are drawn separately, one from the whole population and the other from the positive population. This setting is commonly seen in biotechnology applications~\citep{romero2015dissecting}. More specifically, under the case-control setting, we assume that the $n_u$ unlabeled samples are drawn from the original population while $n_i$ positive samples in each category $1\leq i\leq K$ are drawn from the population with response $i$, and the total sample size $n=n_u+\sum_{j=1}^K n_j$. We introduce another random variable $s$ as an indicator for whether the sample is selected, that is, we only observe a sample $(x,z)$ if the associated indicator variable $s=1$. Also let $\pi_1,\dots,\pi_k\in (0,1)$ be the probability of seeing a sample with response $j$: $\pi_j = P(y=j)$, from the original population. In this case-control setting, one can compute the conditional distribution of $z$ given $y$ and $s=1$ as follows: when $y=k$ and $s=1$, we have $z\in \{0,k\}$ and
\begin{equation}\label{eq:case-control}
    \bbP(z=k|y=k,s=1,x)=\frac{n_k}{n_k+\pi_kn_u}, \quad \bbP(z=0|y=k,s=1,x)=\frac{\pi_k n_u}{n_k+\pi_kn_u};
\end{equation}
When $y=0$ and $s=1$, $\bbP(z=0|y=0,s=1,x)=1$.
More detailed derivation of this conditional distribution can be found in the Appendix \ref{sec:proof-likelihoods}.

While for the {\bf single-training-set} scenario, we first draw all $n$ samples from the whole population randomly, independent of the true response $y$ or covariate $x$. Then each positive sample is unlabelled with some constant probability:
\begin{equation}\label{eq:single-training}
    \bbP(z=k|y=k,s=1,x)=\pi^{\st}_k,\quad \bbP(z=0|y=k,s=1, x)=1-\pi^{\st}_k,
\end{equation}
where $\pi^{\st}\in (0,1)^K$ and we use the superscript ``$\st$'' to denote the single-training-set scenario.

In the following, we will discuss the detailed formulations for the distribution of the true response $y$ given the covariate $x$, under the multinomial-PU model and the ordinal-PU model considered in this paper. 

\paragraph{Multinomial-PU Model:}
This model is suited to the case where the $K$ positive categories are not ordered. Let $\Theta^*\in \bbR^{p\times K}$ and $b^*\in \bbR^K$ be the regression parameter and the offset parameter. We assume 
\begin{equation}\label{eq:MN_model}
    x\sim \mathbb{P}_X, \quad \bbP(y=j|x)=\begin{cases}\frac{e^{x^\top\Theta^{\ast}_{j}+b^*_{j}}}{1+\sum\limits_{k=1}^{K}e^{x^\top\Theta^*_k+b^*_k}},&j>0,\\
    \frac{1}{1+\sum\limits_{k=1}^{K}e^{x^\top\Theta^*_k+b^*_k}},&j=0,
    \end{cases}
\end{equation}
where $\Theta^*_j\in \bbR^p$ is the $j$th column of $\Theta^*$, $b_j^*\in \bbR$ is the $j$th entry of $b^*$. Assuming that $\{(x_i,y_i)\}_{i=1}^n$ are i.i.d samples generated from this model, and each $z_i$ is generated according to either the case-control or the single-training-set scenario, we want to estimate the unknown parameters $\Theta^*, b^*$ based on $\{(x_i,z_i)\}_{i=1}^n$  . 

\paragraph{Ordinal-PU Model:}
When the class labels $0,\,1,\dots, K$ are ordered, it is more appropriate to consider an ordinal modeling approach instead of the multinomial model. This can happen when the class labels are a coarse discretization of some underlying latent variable, such as the maximum temperature level that a protein can stay functional \citep{romero2015dissecting}, or the income level of an American family \citep{mccullagh1980regression}. In this setting, we consider the cumulative logits ordinal regression model \citep{mccullagh1980regression,agresti2010analysis} detailed as follows. Let $\beta^*\in \bbR^p$ be the regression parameter and let $\nu^*\in \bbR^K$ be the offset parameter satisfying $\nu^*_1<\nu^*_2<\dots<\nu^*_K$. We assume 
\begin{equation}\label{eq:ordinal_model}
    x\sim \mathbb{P}_X, \quad \bbP(y<j|x)=\frac{1}{1+e^{-(\nu^*_j-x^\top \beta^*)}}, 1\leq j\leq K.
\end{equation}
The ordinal constraint on $\nu^*$ can impose difficulties in the estimation procedure, and hence we consider the following reparameterization: let $\theta^*\in \bbR^{p+K}$ which satisfies
\begin{equation*}
\theta_j=\begin{cases}\beta_j^*,&1\leq j\leq p,\\
\nu^*_1,&j=p+1,\\
\nu^*_{j-p}-\nu^*_{j-p-1},&j>p+1.\end{cases}
\end{equation*}
By definition, $\theta^*$ can take values in $\bbR^{p+1}\times (0,\infty)^{K-1}$. 
We assume that $\{(x_i,y_i)\}_{i=1}^n$ are i.i.d. samples generated from this model and $z_i$ follows the conditional distribution specified by \eqref{eq:case-control} or \eqref{eq:single-training}, and our goal is to estimate the unknown parameter $\theta^*\in\bbR^{p+1}\times (0,\infty)^{K-1}$ from $\{(x_i,z_i)\}_{i=1}^n$. 

\paragraph{High-dimensional Setting:} In many real applications with presence-only data, the number of covariates can be large compared to the number of samples, and hence we consider the high-dimensional setting. Sparsity or group sparsity of the regression parameters will be assumed, and $\ell_1$ or $\ell_{1,2}$ penalty will be incorporated in our estimators. More details are provided in Section~\ref{sec:algorithm} and Section~\ref{sec:theory}.

\subsection{Related Work}
There have been significant prior works on presence-only data analysis ~\citep{ward2009presence,elkan2008learning,liu2003building,du2015convex,song2020pulasso,song2020convex}, while they all consider binary labels rather than multinomial or ordinal labels. In particular, \cite{ward2009presence} and \cite{song2020pulasso} are most closely related to our works; both consider the logistic regression model which corresponds to the special case of our models with $K=1$, assuming the case-control approach. \cite{ward2009presence} presents an expectation-maximization (EM) algorithm for the low-dimensional setting while \cite{song2020pulasso} proposes a maximization-majorization (MM) algorithm with $\ell_{1,2}$ penalties for the high-dimensional setting. 

Considering the high-dimensional setting, our work is closely related to prior works on estimating sparse generalized linear models~\citep{van2008high,fan2010sure,kakade2010learning,friedman2010note,tibshirani1996regression,li2020graph}. Most of these works focus on convex log-likelihood losses. While in our work, we are concerned with non-convex log-likelihood losses and hence we provide statistical guarantees for stationary points for the penalized losses instead of global minimizers. This approach has also been adopted by~\cite{song2020pulasso,loh2017support,loh2013regularized}.
The main ideas of our proof is similar to \cite{song2020pulasso,loh2017support}, while the key challenge lies in establishing restricted strong convexity for the multinomial and ordinal log-likelihood losses with PU data. 
\cite{xu2017generalized} also considers the estimation of GLMs with sparsity or other constraints, where the MM framework is also used to address the non-convexity of the loss function. However, the non-convexity is induced by their proposed distance-to-set penalties, instead of non-convex log-likelihood losses as in our paper. 

Our work is also closely related to the previous literature on high-dimensional multinomial and ordinal logistic regression~\citep{wurm2017regularized,archer2014ordinalgmifs,krishnapuram2005sparse}. These works consider fully observed data, while for positive and unlabeled data, the only way to apply these methods is to treat the unlabeled samples as being negative. This naive approach is adopted as the baseline method in our numerical experiments on positive and unlabeled data sets, 
demonstrating the non-trivial empirical advantage of our algorithms. 

\paragraph{Organization} The rest of the paper is organized as follows. We propose the algorithms for estimating the multinomial PU model and the ordinal PU model in Section \ref{sec:algorithm}. We then discuss the convergence guarantees of both algorithms in Section \ref{sec:theory}, from both the optimization and statistical perspectives. A series of empirical experiments on synthetic data and real data sets are included in Sections \ref{sec:simulation} and \ref{sec:realdata}. We conclude with some discussion in Section \ref{sec:disc}.
\section{Proposed Algorithms}\label{sec:algorithm}
In this section, we present the corresponding estimation algorithms for our multinomial-PU model and ordinal-PU model under the high-dimensional setting. We will focus on the case-control approach for modeling $\bbP(z_i|y_i)$ and defer our algorithms for the single-traning-set case to Appendix~\ref{sec:alg-st} for simplicity. Specifically, we derive the exact forms of the observed log-likelihood losses for both models and propose to apply the proximal gradient descent (PGD) algorithm to minimize the penalized log-likelihood loss. An alternative approach is to use regularized EM algorithms, whose detailed formulation for both models and convergence properties can also be found in Appendix \ref{append:converg}. We will focus on the PGD algorithm in the main paper due to its computational efficiency. 

\subsection{Multinomial-PU Model}\label{sec:multinom}
Before presenting the detailed algorithms for the multinomial-PU model, we first write out the log-likelihood functions for the observed data $\{(x_i,z_i,s_i=1)\}$ and for the full data $\{(x_i,y_i,z_i,s_i=1)\}$ in the following lemma, where $s_i$ is the selection indicator variable introduced in Section \ref{sec:problem-form}.
\begin{lemma}\label{lem:MN-log-likelihoods}
 The log-likelihood function for the observed presence-only data $\{x_i,z_i,s_i=1\}_{i=1}^n$ is
 \begin{equation*}
 \begin{split}
     &\log L^{\MN}(\Theta,b;\{(x_i,z_i,s_i=1)\}_{i=1}^n)\\
     =&\sum_{i=1}^n \left[\sum_{k=1}^K \ind{z_i=k}[f(\Theta^\top x_i+b)]_k-\log\left(1+\sum_{k=1}^Ke^{[f(\Theta^\top x_i+b)]_k}\right)\right],
 \end{split}
 \end{equation*}
 where $f:\bbR^K\rightarrow \bbR^K$ satisfies $(f(u))_k=u_k+\log \frac{n_k}{\pi_kn_u}-\log(1+\sum_{j=1}^Ke^{u_j})$.
 The log-likelihood function for the full data $\{x_i,y_i,z_i,s_i=1\}_{i=1}^n$ is
 \begin{equation*}
 \begin{split}
     &\log L_f^{\MN}(\Theta,b;\{(x_i,y_i,z_i,s_i=1)\}_{i=1}^n)\\
     =&\sum_{i=1}^n
    \bigg[\sum_{k=1}^K \ind{y_i=k}(x_i^\top \Theta_k+b_k)-\log\left(1+\sum_{k=1}^K(1+\frac{n_k}{\pi_kn_u})e^{x_i^\top\Theta_k+b_k}\right)\\
    &+\sum_{k=1}^K\ind{y_i=z_i=k}\log\frac{n_k}{\pi_kn_u}\bigg],
 \end{split}
 \end{equation*}
 where $\Theta_k$ is the $k$th column of $\Theta$.
\end{lemma}
\noindent The proof of Lemma~\ref{lem:MN-log-likelihoods} is included in Appendix~\ref{sec:proof-likelihoods}. A comparison between the observed PU log-likelihood and the full log-likelihood suggests that the function $f(\cdot)$ reflects the property of presence-only data.
\paragraph{PGD for solving penalized MLE:} One natural idea for estimating the model parameters under the high-dimensional setting is to minimize a penalized log-likelihood loss function that encourages sparsity of the regression parameters. Specifically, let $$\cL_n^{\MN}(\Theta,b)=-\frac{1}{n}\log L^{\MN}(\Theta,b;\{(x_i,z_i,s_i=1)\}_{i=1}^n),$$ then we would like to minimize $$\mathcal{F}_n^{\MN}(\Theta,b)=\mathcal{L}_n^{\MN}(\Theta,b)+P_{\lambda}^{\MN}(\Theta),$$ 
where we set $P_{\lambda}^{\MN}(\Theta)=\lambda\|\Theta\|_{\omega,2,1}$ as a group sparse penalty, which could reflect potential group structures of the covariates that commonly arise in biomedical or genetic applications. More specifically,
\begin{equation}
    \|\Theta\|_{\omega,2,1}=\sum_{j=1}^J \omega_j\|\Theta_{\cG_j}\|_2.
\end{equation}
Here $\cG_1,\dots,\cG_J$ are $J$ non-overlapping groups satisfying $\cup_{j}\cG_j=[p]\times [K]$, $\cG_j=\cR_{j}\times \cC_{j}$ for some $\cR_{j}\subset [p]$, $\cC_{j}\subset [K]$ with sizes $r_j=|\cR_{j}|$, $c_j=|\cC_{j}|$. $\Theta_{\cG_j}$ is a sub-matrix of $\Theta$ with rows indexed by $\cR_{j}$ and columns indexed by $\cC_{j}$. When $J=1$ and $\cG_1=[p]\times [K]$, $P_{\lambda}^{\MN}(\Theta)$ becomes the $\ell_1$ penalty. To minimize $\mathcal{F}_n(\Theta,b)$, we can directly apply the proximal gradient descent~\citep{wright2009sparse} algorithm with an $\ell_{1,2}$ penalty.
More specifically, let $[\Theta^m;b^m]\in \bbR^{(p+1)\times K}$ be concatenated by $\Theta^m$ and $b^m$ in rows, then we can update the parameters at the $m$th iteration as follows: 
\begin{align}\label{eq:pgd-MN}
    &[\Theta^{m+1};b^{m+1}]\notag\\
    =&\argmin_{[\Theta;b]\in \bbR^{(p+1)\times K}}\left\{\eta_mP^{\MN}_{\lambda}(\Theta)+\frac{1}{2}\|[\Theta;b]-([\Theta^m;b^m]-\eta_m\nabla \cL_n^{\MN}(\Theta^m,b^m))\|_2^2\right\},
\end{align}
where $\eta_m$ is the step size. We can choose the initializer $[\Theta^0;b^0]$ such that its corresponding loss function is no greater than any intercept-only model: $\cF_n^{\MN}(\Theta^0,b^0)\leq \min_b\cF_n^{\MN}(0_{p\times K},b)$, which is satisfied by $[0_{p\times K};\argmin_b\cF_n^{\MN}(0_{p\times K},b)]$.
\subsection{Ordinal-PU Model}\label{sec:ordinal}
Similarly from the multinomial-PU model, we first present the log-likelihood functions in the following lemma.
\begin{lemma}\label{lem:ON-log-likelihoods}
 The log-likelihood function for the observed presence-only data $\{x_i,z_i,s_i=1\}_{i=1}^n$ is
 \begin{equation*}
 \begin{split}
     &\log L^{\ON}(\theta;\{(x_i,z_i,s_i=1)\}_{i=1}^n)\\
     =&\sum_{i=1}^n \left[\sum_{k=1}^K \ind{z_i=k}[f(\log r(x_i,\theta))]_k-\log\left(1+\sum_{k=1}^Ke^{[f(\log r(x_i,\theta))]_k}\right)\right],
 \end{split}
 \end{equation*}
 where $f$ is defined in Lemma~\ref{lem:MN-log-likelihoods} and $r:\bbR^p\times \bbR^{p+K}\rightarrow \bbR^{K}$ satisfies
\begin{equation}
\begin{split}
    &r_j(x,\theta)\\
    =&(1+e^{x^\top \theta_{1:p}-\theta_{p+1}})\left[(1+e^{x^\top \theta_{1:p}-\sum_{l=1}^{j+1}\theta_{p+l}})^{-1}-(1+e^{x^\top \theta_{1:p}-\sum_{l=1}^{j}\theta_{p+l}})^{-1}\right],
\end{split}
\end{equation}
for $1\leq j<K$ and $r_K(x,\theta)=(1+e^{x^\top \theta_{1:p}-\theta_{p+1}})\left[1-(1+e^{x^\top \theta_{1:p}-\sum_{l=1}^{j}\theta_{p+l}})^{-1}\right]$.
 The log-likelihood function for the full data $\{x_i,y_i,z_i,s_i=1\}_{i=1}^n$ is
 \begin{equation*}
 \begin{split}
     &\log L_f^{\ON}(\theta;\{(x_i,y_i,z_i,s_i=1)\}_{i=1}^n)\\
     =&\sum_{i=1}^n
    \bigg[\sum_{k=1}^K \ind{y_i=k}\log r_k(x_i,\theta)-\log\left(1+\sum_{k=1}^K(1+\frac{n_k}{\pi_kn_u})r_k(x_i,\theta)\right)\\
    &+\sum_{k=1}^K\ind{y_i=z_i=k}\log\frac{n_k}{\pi_kn_u}\bigg].
 \end{split}
 \end{equation*}
\end{lemma}
\noindent The proof of Lemma~\ref{lem:ON-log-likelihoods} is included in Appendix~\ref{sec:proof-likelihoods}. Compared to the log-likelihoods of the multinomial-PU model in Lemma~\ref{lem:MN-log-likelihoods}, the only difference in Lemma~\ref{lem:ON-log-likelihoods} is that $\Theta^\top x_i+b\in \bbR^K$ is substituted by $\log r(x_i,\theta)$. The reason behind this connection is the following fact: $\Theta_k^\top x_i+b_k$ and $\log r_k(x_i,\theta)$ are the log-ratios between $\bbP(y_i=k|x_i)$ and $\bbP(y_i=0|x_i)$ under the multinomial-PU model and the ordinal-PU model, respectively. 

\paragraph{PGD for solving penalized MLE:} Let $$\mathcal{L}_n^{\ON}(\theta)=-\frac{1}{n}\log L^{\ON}(\theta;\{(x_i,z_i,s_i=1)\}_{i=1}^n).$$ Similarly to the multinomial case, we can estimate $\theta^*$ by applying the proximal gradient descent algorithm on the penalized loss function $\cF^{\ON}_n(\theta)=\cL_n^{\ON}(\theta)+P^{\ON}_{\lambda}(\theta_{1:p})$, where $P_{\lambda}^{\ON}(\theta_{1:p})=\lambda\|\theta_{1:p}\|_{\omega,2,1}=\lambda\sum_{j=1}^J\omega_j\|\theta_{\cG_j}\|_2$, with disjoint groups $\cG_1,\dots,\cG_J\subset[p]$ satisfying $\cup_j\cG_j=[p]$. At each iteration $m$, we update the parameter as follows:
\begin{align}\label{eq:pgd-ON}
    \theta^{m+1}=\argmin_{\theta\in \bbR^{p+K}}\left\{\eta_mP^{\ON}_{\lambda}(\theta_{1:p})+\frac{1}{2}\|\theta-(\theta^m-\eta_m\nabla \cL_n^{\ON}(\theta^m))\|_2^2\right\},
\end{align}
where $\eta_m$ is the step size. We choose the initializer $\theta^0$ with loss function no greater than any intercept-only model: $\cF_n^{\ON}(\theta^0)\leq \min_{\theta_{1:p}=0_{p\times 1}}\cF_n^{\ON}(\theta)$, and a simple choice is $\theta^0=\argmin_{\theta_{1:p}=0_{p\times 1}}\cF_n^{\ON}(\theta)$.
\section{Theoretical Guarantees}\label{sec:theory}
In this section, we provide optimization and statistical guarantees for our algorithms presented in Section~\ref{sec:algorithm}, proposed for the case-control setting. We will briefly describe the theoretical properties for the algorithms under the single-training-set scenario in Appendix~\ref{sec:alg-st}. 
\subsection{Algorithmic Convergence}
We first show that when applying the projected gradient descent algorithm to minimize our penalized log-likelihood losses, it would converges to a stationary point. Proposition~\ref{lem:converg-MN} focuses on the multinomial-PU model, and we also have similar convergence guarantees for the ordinal-PU model, deferred to Appendix \ref{append:converg}. 
\begin{proposition}[Convergence of algorithms for the multinomial-PU model]\label{lem:converg-MN}
 If the parameter iterates $\{\Theta^m,b^m\}_m$ are generated by the PGD update~\eqref{eq:pgd-MN} with proper choices of step sizes, they would satisfy the following:
 \begin{itemize}
    \item[(i)] The sequence $\{(\Theta^m,b^m)\}_m$ has at least one limit point.
     \item[(ii)]There exists $R^{\MN}>0$ such that all limit points of $\{(\Theta^m,b^m)\}_m$ belong to $\Gamma^{\MN}$, the set of first order stationary points of the optimization problem $\min_{\|\Theta\|_2+\|b\|_2\leq R^{\MN}}\mathcal{F}_n^{\MN}(\Theta,b)$.
     \item[(iii)]The sequence of function values $\{\mathcal{F}_n^{\MN}(\Theta^m,b^m)\}_m$ is non-increasing, and $\mathcal{F}_n^{\MN}(\Theta^{m+1},b^{m+1})<\mathcal{F}_n^{\MN}(\Theta^{m},b^{m})$ holds if $(\Theta^m,b^m)\notin \Gamma^{\MN}$. In addition, there exists $(\widetilde{\Theta},\widetilde{b})\in \Gamma^{\MN}$ such that $\{\mathcal{F}_n^{\MN}(\Theta^m,b^m)\}_m$ converges monotonically to $\mathcal{F}_n^{\MN}(\widetilde{\Theta},\widetilde{b})$.
 \end{itemize}
\end{proposition}

\subsection{Statistical Theory for Stationary Points: Multinomial-PU Model}\label{sec:theory-MN}
In this section, we will provide theoretical guarantees for the stationary points of $\mathcal{F}_n^{\MN}(\Theta,b)=\cL_n^{\MN}(\Theta,b)+P^{\MN}_{\lambda}(\Theta)$.  
For simplicity, we assume zero offset parameter $b=0_{K\times 1}$, and we present the statistical properties for any stationary point $\widehat{\Theta}$ of the penalized loss $\mathcal{F}_n(\Theta,0_{K\times 1})$ within some feasible region around $\Theta^*$: $\Theta^*+\bbB_{1,\infty}(R_0)$, where $\bbB_{1,\infty}(R_0)=\{\Delta:\max_{1\leq i\leq K}\|\Delta_{:,i}\|_1\leq R_0\}$ and conditions on $R_0>0$ will be specified shortly in Assumption~\ref{assump:mult_region_bound}. 

\paragraph{Key quantities for theory under the multinomial-PU model:} Before presenting our theoretical results, we first define some key quantities.
The maximum group size parameter is defined as $m=\max_{1\leq j\leq J}c_jr_j$, where $c_j$, $r_j$ are group size parameters associated with each group $\mathcal{G}_j$ in the $\ell_{1,2}$ penalty. Let $R^*=\max_{1\leq i\leq K}\|\Theta_{:,i}^*\|_1$ be a boundedness parameter of $\Theta^*$. In addition, let $S=\{j:\Theta^*_{\cR_j,\cC_j}\neq 0\}$ and $s=|S|$. 

The following assumptions are needed to derive our theoretical guarantees.
\begin{assumption}[Sub-Gaussian covariates]\label{assump:subgauss_row}
$\{x_i\in \bbR^p\}_{i=1}^n$ are independent mean zero sub-Gaussian vectors with sub-Gaussian parameter $\sigma$ and covariance matrix $\Sigma$, satisfying that $\sigma^2\leq C\min\{\lambda_{\min}(\Sigma),1\}$. Meanwhile, $|x_{ij}|\leq C_x$ for some constant $0<C_x\leq C\lambda_{\min}^{\frac{1}{2}}(\Sigma)$. 
\end{assumption}  

The sub-Gaussian condition for $x_i$ in Assumption~\ref{assump:subgauss_row} is commonly seen in the high-dimensional statistics literature, and is satisfied by standard Gaussian vectors. Each entry of $x_i$ is assumed to be bounded in order to ensure bounds for $x_i^\top \Theta$ for any $\Theta$ in the feasible region, so that the loss function can be concentrated appropriately. 
In addition, we have assumed $\sigma^2\leq C\{\lambda_{\min}(\Sigma),1\}$ in order to show the restricted eigenvalue condition with high probability. 
\begin{assumption}[Rate conditions]\label{assump:scaling}
The group weight $\omega$ in the group sparsity penalty satisfies $c\leq \min_j\omega_j\leq \max_j\omega_j\leq C\sqrt{\frac{n}{\log J+m}}$; The sample sizes of labeled and unlabeled data satisfies $c\leq \min_j\frac{n_j}{\pi_jn_u}\leq C$; $\lambda_{\max}(\Sigma)\geq c$, $K\log J\leq Cn$ and $\log J\geq cm\log \log n$. 
\end{assumption}
The rate conditions in Assumption~\ref{assump:scaling} are standard and comparable to the past literature in high-dimensional statistics and PU-learning~\citep[eg.][]{song2020pulasso,loh2013regularized}.

Another key assumption for our statistical theory is concerned with the feasible region of the optimization problem: $\Theta\in \Theta^* + \mathbb{B}_{1,\infty}(R_0)$. As will be explained more clearly in Section \ref{sec:main_proofs}, our non-convex loss function evaluated at $\Theta$ can be decomposed into a restricted convex term and a non-convex mean-zero term which can be concentrated well with high probability. The restricted convex term has lower bounded curvature within a particular neighborhood of $\Theta^*$, motivating us to impose a constraint on the neighborhood radius $R_0$ as follows.
\begin{assumption}\label{assump:mult_region_bound}
$$
R_0\leq \frac{\min_j\frac{n_j}{\pi_j n_u}e^{-C_xR^*}}{4C_x(1+\max_j\frac{n_j}{\pi_j n_u})^2(1+1.1Ke^{C_xR^*})^{3}}.
$$
\end{assumption}
The constraint on the radius $R_0$ of the feasible region in Assumption \ref{assump:mult_region_bound} depends on the case-control study design: a more even $\frac{n_j}{\pi_j n_u}$ over $1\leq j\leq K$ leads to larger feasible set; it also depends on the magnitude of the true model parameter $\Theta^*$: a smaller $R^*$ also means a larger feasible set. The latter relationship has an intuitive explanation: the original multinomial log-likelihood loss with all labels observed has larger curvature around zero, and hence similarly, the non-convexity of the multinomial-PU loss is also milder around zero. Although Assumption \ref{assump:mult_region_bound} might seem a little stringent, as we will show in Sections \ref{sec:simulation} and \ref{sec:realdata}, the initialization with the best intercept-only model usually works well throughout our synthetic and real data experiments. 
 
Before formally stating our theoretical results, we also define the function
\begin{equation}\label{eq:mult_h_def}
    h^{\MN}(R)=\frac{e^{-C_x(R+R^*)}\min_j\frac{n_j}{\pi_j n_u}}{(1+\max_j\frac{n_j}{\pi_j n_u})^2\left(1+Ke^{C_x(R+R^*)}\right)^3}-4C_xR.
\end{equation}
As we will show in our proofs, this function $h^{\MN}(R)$ determines a lower bound of the log-likelihood's restricted strong convexity, and $h^{\MN}(R_0)$ is guaranteed to be positive as long as Assumption \ref{assump:mult_region_bound} holds (see Lemma \ref{lem:MN_region_results}). 
\begin{theorem}\label{thm:mult_upp}
Suppose Assumptions \ref{assump:subgauss_row} to \ref{assump:mult_region_bound} hold, and $\lambda\geq C\lambda^{\frac{1}{2}}_{\max}(\Sigma)(K+C_xR^*\sqrt{K})\sqrt{\frac{\log J+m}{n}}$
for some constant $C>0$. If $\widehat{\Theta}$ is a stationary point of $\min_{\Theta-\Theta^*\in\bbB_{1,\infty}(R_0)}\cF_n^{\MN}(\Theta,0_{K\times 1})$, then
\begin{equation}
    \|\widehat{\Theta}-\Theta^*\|_2\leq \frac{3\|\omega_S\|_2\lambda}{h^{\MN}(R_0)\lambda_{\min}(\Sigma)},\quad \|\widehat{\Theta}-\Theta^*\|_{\omega,2,1}\leq \frac{12\|\omega_S\|^2_2\lambda}{h^{\MN}(R_0)\lambda_{\min}(\Sigma)},
\end{equation}
with probability at least $1-\exp\{-c\frac{K^2(\log J+m)}{\max_jr_j}\}-\exp\{-c(\log J+m)\}.$ 
\end{theorem}
\begin{remark}
When $\omega_j\asymp C$ for all $1\leq j\leq J$, and $\lambda\asymp \lambda_{\max}^{\frac{1}{2}}(\Sigma)(K+C_xR^*\sqrt{K})\sqrt{\frac{\log J+m}{n}}$, Theorem~\ref{thm:mult_upp} suggests that $\|\widehat{\Theta}-\Theta^*\|_2=O(K\sqrt{\frac{s(\log J+m)}{n}})$ and $\|\widehat{\Theta}-\Theta^*\|_{\omega,2,1}=O(Ks\sqrt{\frac{\log J+m}{n}})$ with high probability. In particular, when the number of positive labels $K\leq C$, this rate is the same as the one derived in~\cite{song2020pulasso}.
\end{remark}
\begin{remark}
The major challenge for proving Theorem~\ref{thm:ordinal_upp} is to show that the loss function $\cL_n^{\MN}(\theta)$ is concentrated around a restricted strongly convex function within $\Theta^*+\bbB_{1,\infty}(R_0)$, see Lemma~\ref{lem:RSC}. Our proof is based on the techniques developed in \cite{song2020pulasso}, while considering multiple positive labels requires more careful analysis, e.g., the proof of Lemma~\ref{lem:mult_rsc_term1}. In addition, due to $K>1$, our proof relies on a vector-contraction inequality for Rademacher complexities~\citep{maurer2016vector} where the contraction functions have a vector-valued domain. More detailed discussion on the non-convexity issues and how we address them are presented in Section \ref{sec:proof_sketch_MN}.
\end{remark}
\subsection{Statistical Theory for Stationary Points: Ordinal-PU Model}\label{sec:theory-ON}
Similar to Section~\ref{sec:theory-MN}, in the following we will show an error bound for $\widehat{\theta}-\theta^*$, where $\widehat{\theta}$ is any stationary point of the loss function $\mathcal{F}_n^{\ON}(\theta)$ within some region around $\theta^*$: $\theta-\theta^*\in S(R_0,r_0)$. Here the set $S(R_0,r_0)$ is defined as follows:
\begin{equation}\label{eq:ordinal_region}
    S(R_0,r_0)=\{\Delta:\|\Delta_{1:p}\|_1,\|\Delta_{(p+1):(p+K)}\|_1\leq R_0,\min_{2\leq j\leq K}\Delta_{p+j}\geq -r_0\},
\end{equation}
where $R_0, r_0>0$ depend on the true parameters and will be specified later in Assumption~\ref{assump:ordinal_region_bound}. Different from the multinomial-PU model, here the constraint set $S(R_0,r_0)$ also impose a lower bound on $\Delta_{(p+2):(p+K)}$ instead of only the norms of $\Delta$. This is due to the nature of the ordinal logistic regression model: we need to ensure $\widehat{\theta}_{p+j}>0$ for $2\leq j\leq K$ so that $\bbP_{\widehat{\theta}}(y=k|x)>0$ for all $k>0$. The constraint set $S(R_0,r_0)$ is larger if $R_0$ or $r_0$ increases.

\paragraph{Key quantities for theory under the ordinal-PU model:} The maximum group size parameter is defined as $m=\max_{1\leq j\leq J}g_j$. Let $R^*=\max\{\|\theta_{1:p}^*\|_1,\|\theta_{(p+1):(p+K)}^*\|_1\},$
$r^*=\min_{1\leq j\leq K-1}\theta_{p+j+1}^*>0$ be boundedness parameters of $\theta^*$. In addition, let $S=\{j:\theta^*_{\cG_j}\neq 0\}$ and $s=|S|$. 

For the ordinal-PU model, we also assume sub-gaussian covariates and the same scalings of parameters as the multinomial-PU model (Assumption~\ref{assump:subgauss_row} and Assumption~\ref{assump:scaling}). In addition, similar to Assumption \ref{assump:mult_region_bound}, we need the following condition on the feasible region $S(R_0,r_0)$ so that the log-likelihood loss satisfies the restricted strong convexity within $\theta^*+S(R_0,r_0)$.
\begin{assumption}[Constraint on $R_0,r_0$]\label{assump:ordinal_region_bound}
$0<r_0<r^*$ and 
\begin{equation}\label{eq:ordinal_region_bound}
R_0\leq \min\left\{80,\min_j\frac{n_j}{\pi_jn_u}\right\}\frac{\left(1+e^{(C_x+1)(R^*+0.01)}\right)^{-6}}{512(C_x+1)K^3(1+\frac{2}{r^* - r_0})^3(2+\frac{2}{r^*-r_0})}
\end{equation}
\end{assumption}
The condition $0<r_0<r^*$ ensures that $\widehat{\theta}_{p+j}>0$ for $j\geq 2$. In addition, Assumption \ref{assump:ordinal_region_bound} requires $R_0$ to be not too large, similar to Assumption \ref{assump:mult_region_bound} for the multinomial-PU model. The constraint \eqref{eq:ordinal_region_bound} can be weakened by one of the following changes: a reduction in the number of unlabeled samples, a decrease in the magnitude $R^*$ of the true parameter, or a smaller value of $r_0$ (a weaker lower bound constraint on $\Delta_{p+j}, 2\leq j\leq K$).
We also define the following key function $h^{\ON}:\bbR^2\rightarrow\bbR$ that determines the restricted curvature of the log-likelihood loss in $\theta\in\theta^*+S(R,r)$:
\begin{equation}\label{eq:ordinal_h_def}
    h^{\ON}(R,r)=\max\left\{\frac{(1+e^{(C_x+1)(R^*+R)})^2}{(r^*-r)e^{(C_x+1)(R^*+R)}},1+e^{(C_x+1)(R^*+R)}\right\},
\end{equation}
which is positive as long as $r<r^*$.
\begin{theorem}\label{thm:ordinal_upp}
There exist $L_0, \gamma_0>0$ depending only on $R^*,R_0,r^*,r_0, C_x$, such that if Assumptions \ref{assump:subgauss_row}, \ref{assump:scaling} and \ref{assump:ordinal_region_bound} hold, $\lambda\geq CL_0\lambda_{\max}^{\frac{1}{2}}(\Sigma)\sqrt{\frac{K(\log J+m)}{n}},$ then for any stationary point $\widehat{\theta}$ of 
$\min_{\theta-\theta^*\in S(R_0,r_0)}\cF_n^{\ON}(\theta)$, we have
\begin{equation}
    \|\widehat{\theta}-\theta^*\|_2\leq\frac{\max_j\frac{\pi_jn_u}{n_j}K}{\gamma_0\min\{\lambda_{\min}(\Sigma),1\}}\left( 3\|\omega_S\|_2\lambda+CL_0K\sqrt{\frac{(\log J+m+\log K)}{n}}\right),
\end{equation}
with probability at least $1-\exp\{-\frac{cK\log J}{C_x^2m+1}\}-3\exp\{-(\log J+m)\}$. 
\end{theorem}
\begin{remark}
When $\omega_j\asymp C$ for all $1\leq j\leq J$, $K\leq J^C$ for some constant $C$ and $\lambda\asymp L_0\lambda_{\max}^{\frac{1}{2}}(\Sigma)\sqrt{\frac{K(\log J+m)}{n}}$, Theorem~\ref{thm:ordinal_upp} suggests that $\|\widehat{\theta}-\theta^*\|_2=O(K^{\frac{3}{2}}(\sqrt{s}+\sqrt{K})\sqrt{\frac{\log J+m}{n}})$ with high probability. We recover the same error bound as in ~\cite{song2020pulasso} when $K$ is a constant; otherwise, the rate scales polynomially w.r.t. $K$. 
\end{remark}
\begin{remark}
In the $\ell_2$ error bound, here we have a factor $\sqrt{s}+\sqrt{K}$ instead of $\sqrt{s}$ only as in Theorem~\ref{thm:mult_upp}, since the $K$-dimensional offset parameter $\theta^*_{(p+1):(p+K)}$ is also estimated together with the $s$-sparse regression parameter $\theta_{1:p}$. Due to the same reason,  here we don't provide an error bound for $\|\widehat{\theta}_{1:p}-\theta^*_{1:p}\|_{\omega,2,1}$ which takes a more complicated form, although its proof would be similar to Theorem~\ref{thm:mult_upp}.
\end{remark}

\begin{remark}
As revealed by our proof, the specific definitions of $\gamma_0$ and $L_0$ are as follows: $$\gamma_0=\frac{e^{2(C_x+1)(R^*+R_0)}}{32h^{\ON}(R_0,r_0)(1+e^{(C_x+1)(R^*+R_0)})^4},$$ $$L_0=(\frac{4}{3}(C_x+1)R_0+1)h^{\ON}(R_0,r_0)+\frac{(C_x+1)R_0}{4}(h^{\ON}(R_0,r_0))^{2}.$$
We can see that larger $R^*, R_0, r_0$ and a smaller $r^*,\lambda_{\min}(\Sigma)$ would lead to a larger $L_0$ and a smaller $\gamma_0$, which then further implies larger estimation error bound for $\|\widehat{\theta}-\theta^*\|_2$. This is due to a smaller curvature of $\cL_n^{\ON}$ in $S(R_0,r_0)$ with larger $R^*+R_0$, smaller $r^*-r_0$ and smaller $\lambda_{\min}(\Sigma)$.
\end{remark}
\begin{remark}
Similar to the proof of Theorem~\ref{thm:mult_upp}, the major challenge for proving Theorem~\ref{thm:ordinal_upp} is to show that the loss function $\cL_n^{\ON}(\theta)$ is concentrated around a restricted strongly convex function within $\theta^*+S(R_0,r_0)$, see Lemma~\ref{lem:ordinal_rsc}. In particular, considering the ordinal-PU model with multiple positive labels leads to a much more complicated log-likelihood loss than the multinomial-PU model: recall that the log-ratio term $\Theta^{\top} x_i+b$ in Lemma~\ref{lem:MN-log-likelihoods} changes to $\log r(x_i,\theta)$ in Lemma~\ref{lem:ON-log-likelihoods}. Novel lower and upper bounds for the derivatives of the ordinal log-likelihood losses are needed (see Section~\ref{sec:ordinal_rsc_proof}).
\end{remark}
\section{Proof Sketch}\label{sec:main_proofs}
Following a conventional proof strategy for the statistical error of high-dimensional M-estimators \citep{negahban2009unified}, one main building block of our proofs is to show the restricted strong convexity of our log-likelihood losses within a proper region. As mentioned earlier, the main theoretical challenge is dealing with the non-convex landscape of the observed log-likelihood losses resulting from the presence-only data. The multi-class labels magnify the technical difficulties, especially in the ordinal-PU setting. In this section, we will explain where the non-convexity stems from and then give a brief proof sketch that outlines our main idea for tackling this non-convexity. The full proofs are all deferred to the appendix.
\subsection{Illustration of the Non-convexity of Log-likelihood Losses}\label{sec:proof_sketch_MN}
We start with the log-likelihood loss of the multinomial-PU model for illustration purposes and then discuss the proof ideas for both models subsequently. If assuming the multinomial-PU model with zero offset parameters for simplicity, our log-likelihood loss takes the following form:
\begin{equation}\label{eq:proof_MN_loss}
    \cL_n^{\MN}(\Theta) = \frac{1}{n}\sum_{i=1}^n\left[A(f(\Theta^\top x_i))-\delta_i^\top f(\Theta^\top x_i)\right]
\end{equation}
where $A:\bbR^K\rightarrow \bbR$ is the multinomial log-partition function that satisfies $A(u) = \log(1+\sum_{k=1}^Ke^{u_k})$; $\delta_i\in \mathbb{R}^K$ is the indicator vector for the observed label of sample $i$: $\delta_i=\begin{cases}0, & \text{ if }z_i=0\\e_k, &\text{ if }z_i=k\end{cases}$, while the function $f:\bbR^K\rightarrow \bbR^K$ satisfies
\begin{equation}\label{eq:f_def}
    (f(u))_k=u_k+\underbrace{\log \frac{n_k}{\pi_kn_u}-\log(1+\sum_{j=1}^Ke^{u_j})}_{\text{Effect of PU-modeling}}.
\end{equation} 
Here, the presence-only data generation mechanism is encoded in the last two terms of the nonlinear function $f(\cdot)$ defined above, which introduces non-convexity into the loss function. As a comparison, if not considering the PU-learning setting and assuming all observed labels are true labels, then the log-likelihood loss would also take the form of \eqref{eq:proof_MN_loss} but with function $f(\cdot)$ substituted by the identity link. 
To be more specific about the non-convexity and the intuitive idea on how we address it, we rewrite the loss function as $\cL_n^{\MN}(\Theta)=\frac{1}{n}\sum_{i=1}^ng_i(\Theta^\top x_i),$ where $g_i:\bbR^K\rightarrow \bbR$ satisfies $g_i(u) = A(f(u))-\delta_i^\top f(u)$. In order to show a type of restricted strong convexity for $\cL_n^{\MN}(\Theta)$, one key step is to characterize the landscape/curvature of $g_i$ over the potential range of $\Theta^\top x_i$. 

Specifically, let $u_i = \Theta^\top x_i\in \mathbb{R}^K$, $u_i^* = \Theta^{*\top} x_i\in \mathbb{R}^K$, then  we want to show a lower bound in the the following form:
 \begin{equation*}
     \langle \nabla g_i(u_i) - \nabla g_i(u_i^*), u_i-u_i^*\rangle\geq c\|u_i-u_i^*\|_2^2,
 \end{equation*}
 for some small constant $c>0$. On the other hand, through some calculations and application of the mean value theorem, we know that
 \begin{equation}\label{eq:g_i_rsc}
     \begin{split}
         \langle \nabla g_i(u_i) - \nabla g_i(u_i^*), u_i-u_i^*\rangle&=\underbrace{(u_i-u_i^*)^\top G_i(u_i-u_i^*)}_{\mathrm{\RNum{1}}}+\underbrace{(u_i-u_i^*)^\top h_i}_{\mathrm{\RNum{2}}},
     \end{split}
 \end{equation}
 where $h_i\in \bbR^K$ is a vector
 $$h_i=(\nabla f(u_i)-\nabla f(u_i^*))^\top (\nabla A(f(u_i^*)) - \delta_i),
 $$
 and $G_i\in \bbR^{K\times K}$ is a matrix:
 $$
 G_i = \nabla f(u_i)^\top \nabla^2 A(f(u_i^t))\nabla f(u_i^t),
 $$
 for some $u_i^t = tu_i +(1-t)u_i^*$ lying between $u_i$ and $u_i^*$, with $t\in (0,1)$. The gradient $\nabla f(u)\in \bbR^{K\times K}$, $\nabla A(f(u))\in \bbR^K$ and the hessian $\nabla^2A(f(u))\in \bbR^{K\times K}$ for any vector $u\in \bbR^K$. 
 We observe that the first term $\mathrm{\RNum{1}}$ in \eqref{eq:g_i_rsc} takes a quadratic form, and the second term $\mathrm{\RNum{2}}$ is of mean zero: due to the property of exponential family random variables, one has
 $$
\mathbb{E}[A(f(u_i^*))-\delta_i|x_i]=0.
 $$
 Therefore, this motivates our main proof idea for conquering the non-convexity issue: \emph{we first lower bound $\mathrm{\RNum{1}}^{\MN}$, and then we concentrate $\mathrm{\RNum{2}}^{\MN}$ around zero so that it does not affect the curvature too much.} This proof idea for addressing the non-convexity in PU-learning is not new; in fact, \cite{song2020pulasso} also used similar ideas to prove statistical theory for binary classification with PU data. However, in the binary setting \citep{song2020pulasso} where $K=1$, matrix $G_i$ reduces to a positive \emph{scalar} $f'(u_i)f'(u_i^t)A''(f(u_i^t))>0$ which can be easily lower bounded as a function of model parameters. However, this strategy, especially lower bounding $\mathrm{\RNum{1}}$, is much more challenging for the general $K>1$ case. In the following subsection, we will discuss how we deal with this challenge for both the multinomial and ordinal models.

\subsection{Proof Sketch for Restricted Strong Convexity}\label{sec:ordinal_rsc_proof}
As discussed earlier, one main challenge of our proof lies in lower bounding $\mathrm{\RNum{1}}$, which takes a quadratic form that depends on an asymmetric parameter matrix $G_i\in \mathbb{R}^{K\times K}$. One may consider symmetrizing it by considering $\frac{1}{2}(G_i+G_i^\top)$; but it is not even unclear if $\frac{1}{2}(G_i+G_i^\top)$ is positive definite or not. Therefore, we consider the following strategy instead: first we decompose the matrix $G_i$ as the sum of a positive semi-definite matrix and an error matrix that depends on how $\Theta$ differs from $\Theta^*$:
 \begin{equation}\label{eq:MN_G}
 G_i =  \underbrace{\nabla f(u_i^t)^\top \nabla^2 A(f(u_i^t))\nabla f(u_i^t)}_{G_i^{(1),\MN}} +  \underbrace{[\nabla f(u_i)-\nabla f(u_i^t)]^\top \nabla^2 A(f(u_i^t))\nabla f(u_i^t)}_{G_i^{(2),\MN}},
 \end{equation}
 where $G_i^{(1),\MN}$ is symmetric positive semi-definite, while the magnitude of $G_i^{(2),\MN}$ depends on the difference between $u_i=\Theta^\top x_i$ and $u_i^t=(t\Theta + (1-t)\Theta^{*})^\top x_i$. 
 We then (i) lower bound the minimum eigenvalue of $G_i^{(1),\MN}$ and (ii) show an upper bound of the maximum singular value of $G_i^{(2),\MN}$ that depends on the distance between $\Theta$ and $\Theta^*$, leading to the restricted strong convexity of our log-likelihood loss within an appropriate neighborhood around the true parameter (see Assumptions \ref{assump:mult_region_bound} and \ref{assump:ordinal_region_bound}).
 In addition, it is still non-trivial to show a positive lower bound for the minimum eigenvalue of $G_i^{(1)}$, whose entries all depend on the model parameters in a non-linear manner. 
\paragraph{Proof sketch for the multinomial-PU model.} 
For the multinomial-PU model, we found that both $\nabla^2A(f(u))\in \mathbb{R}^{K\times K}$ and $\nabla f(u)\in \mathbb{R}^{K\times K}$ can be decomposed as sums of a diagonal matrix and a rank-one matrix, which can be easily inverted. The problem of lower bounding the minimum eigenvalue can then be solved by upper bounding the maximum eigenvalue of their inverse. This leads to our following lemma:
\begin{lemma}\label{lem:sketch_MN_I_lrbnd}
    For any $\Theta\in \Theta^*+\mathbb{B}_{1,\infty}(R)$, 
    \begin{equation}\label{eq:MN_I_lrbnd}
        \lambda_{\min}(\nabla f(\Theta^\top x_i)^\top \nabla^2A(f(\Theta^\top x_i))\nabla f(\Theta^\top x_i))\geq h^{\MN}(R)+4C_xR,
    \end{equation}
    where $h^{\MN}(R)$ is defined in \eqref{eq:mult_h_def}, and $C_x>0$ is a constant in Assumption \ref{assump:subgauss_row}.
\end{lemma}
\noindent The detailed proof of Lemma \ref{lem:mult_rsc_term1} can be found in Appendix \ref{sec:supp_lemma_MN}. 
Based on this result and a smoothness guarantee for $\nabla f(u)$ (upper bound for $G_i^{(2),\MN}$), together with some technical tools from the empirical process, we are able to show a restricted strong convexity for $\mathcal{L}_n^{\MN}(\Theta)$ (see Lemma \ref{lem:RSC} for details).

\paragraph{Proof sketch for the ordinal-PU model.} On the other hand, the proof for the ordinal-PU model is much more involved. To see the difference, we can first write out the log-likelihood loss as follows:
\begin{equation*}
    \cL_n^{\ON}(\theta)=\frac{1}{n}\sum_{i=1}^nA(f^{\ON}(u(x_i,\theta)))-\delta_i^\top f^{\ON}(u(x_i,\theta)).
\end{equation*}
where functions $f^{\ON}:\bbR^K\rightarrow \bbR^K$ and $u:\bbR^p\times \bbR^{p+K}\rightarrow \bbR^K$ satisfy the following: 
\begin{equation}\label{eq:ordinal_f_def}
   (f^{\ON}(u))_j=\begin{cases}\log\frac{n_j}{\pi_jn_u}+\log\left[(1+e^{\sum_{l=1}^{j+1}u_{l}})^{-1}-(1+e^{\sum_{l=1}^ju_l})^{-1}\right],&1\leq j<K,\\
    \log\frac{n_K}{\pi_Kn_u}+\log\left[1-(1+e^{\sum_{l=1}^Ku_l})^{-1}\right],&j=K,
\end{cases}
\end{equation}
$$u(x_i,\theta)=(x_i^\top\theta_{1:p}-\theta_{p+1},-\theta_{p+2},\dots,-\theta_{p+K})^\top.$$ 
Compared to the loss function $\mathcal{L}_n^{\MN}(\theta)$ in the multinomial setting (see \eqref{eq:proof_MN_loss}), $\mathcal{L}_n^{\ON}(\theta)$ involves a more complicated non-linear function $f^{\ON}(\cdot)$ than the function $f(\cdot)$ defined in \eqref{eq:f_def}, and calls for new proof techniques. In particular, analogous to the multinomial setting where we deal with a $K\times K$ matrix $G_i$ by lower bounding a PSD matrix and upper bounding an error matrix in \eqref{eq:MN_G}, here we want to lower bound the minimum eigenvalue of 
\begin{equation}\label{eq:ON_G1}
    G_i^{(1),\ON}:=\nabla f^{\ON}(u_i^t)^\top \nabla^2A (f^{\ON}(u_i^t))\nabla f^{\ON}(u_i^t),
\end{equation} 
and upper bound the maximum singular value of 
\begin{equation}\label{eq:ON_G2}
    G_i^{(2),\ON}:=[\nabla f^{\ON}(u_i)-\nabla f^{\ON}(u_i^t)]^\top \nabla^2 A(f^{\ON}(u_i^t))\nabla f^{\ON}(u_i^t).
\end{equation} 
However, the matrix $\nabla f^{\ON}(u)$ is not easily invertible as in the multinomial setting. Instead, we factorize it as a product of constant invertible matrices and a tridiagonal matrix whose entries depend on all model parameters in a highly nonlinear way. 
We then make use of the special sequential subtraction structure in the ordinal model and transform this problem into lower bounding a telescoping sum through careful analysis. Then we show the following result for $f^{\ON}(\cdot)$:
\begin{lemma}\label{lem:ordinal_rsc_term1_1}
Let $R_0, r_0>0$ be any positive constants. For any $\theta\in \theta^* + S(R_0,r_0)$ where $S(R_0,r_0)$ is as defined in \eqref{eq:ordinal_region}, if $u_i=u(x_i,\theta)$, then we have
 $$
 \lambda_{\min}(\nabla f^{\ON}(u_i)^\top \nabla f^{\ON}(u_i))\geq \frac{e^{2(C_x+1)(R^*+R_0)}}{4K(1+e^{(C_x+1)(R^*+R_0)})^4}.
 $$
\end{lemma}
\noindent The detailed proof of Lemma \ref{lem:ordinal_rsc_term1_1} can be found in Appendix \ref{sec:support_lemma_ON}. When combined with the minimum eigenvalue guarantee for $\nabla^2 A(\cdot)$, Lemma \ref{lem:ordinal_rsc_term1_1} can lead to a lower bound for the
minimum eigenvalue of $G_i^{(1),\ON}$ defined in \eqref{eq:ON_G1}. 
As far as we are aware, such a characterization of the curvature of an ordinal log-likelihood loss has not been shown in prior works. Existing works that analyze statistical properties of the ordinal model are only concerned with the uniqueness of the MLE \citep{mccullagh1980regression} in the low-dimensional setting, or they assume that the curvature is lower bounded by a constant \citep{lee2020tensor} without characterizing the lower bound as a function of model parameters. However, as discussed in Section \ref{sec:proof_sketch_MN}, in order to characterize a restricted convex region for the non-convex log-likelihood loss in the PU model, these prior results fall short to achieve our purpose. 
Furthermore, for the second error term $G_i^{(2),\ON}$ defined in \eqref{eq:ON_G2}, we also prove a Lipschitz property for $\nabla f^{\ON}(u)$ in the following lemma:
\begin{lemma}\label{lem:ordinal_rsc_term1_2_brief}
Let $R_0, r_0>0$ be any positive constants. For any $\theta\in \theta^* + S(R_0,r_0)$ where $S(R_0,r_0)$ is as defined in \eqref{eq:ordinal_region}, if $u_i=u(x_i,\theta)$, $u_i^* = u(x_i,\theta)$, then we have
 $$
 \|\nabla f^{\ON}(u_i)-\nabla f^{\ON}(u_i^*)\|\leq \sqrt{5}[(h^{\ON}(R_0,r_0))^2+h^{\ON}(R_0,r_0)](C_x+1)R_0K,
 $$
 where $h^{\ON}(\cdot)$ is as defined in \eqref{eq:ordinal_h_def}.
\end{lemma}
\noindent The detailed proof of Lemma \ref{lem:ordinal_rsc_term1_2_brief} can also be found in Appendix \ref{sec:support_lemma_ON}. 

Based on Lemma~\ref{lem:ordinal_rsc_term1_1}, Lemma~\ref{lem:ordinal_rsc_term1_2_brief}, some other supporting results and probabilistic concentration bounds, we can then prove a restricted strong convexity property for the ordinal-PU model within a region that depends on the model parameters.
\begin{lemma}[Restricted Strong Convexity under the Ordinal-PU Model]\label{lem:ordinal_rsc}
 If the data set $\{(x_i,z_i)\}_{i=1}^n$ is generated from the ordinal-PU model under the case-control setting, Assumptions~\ref{assump:subgauss_row}, \ref{assump:scaling} and \ref{assump:ordinal_region_bound} hold, then with probability at least $1-\exp\{-\frac{cK\log J}{C_x^2m+1}\}$,
 \begin{equation}\label{eq:ordinal_rsc}
 \begin{split}
     &\langle \nabla\cL_n^{\ON}(\theta)-\nabla\cL_n^{\ON}(\theta^*),\theta-\theta^*\rangle \\
     \geq&\alpha\|\theta-\theta^*\|_2^2-C\alpha\frac{\log J+m}{n}\|\theta-\theta^*\|_{\omega,2,1}^2-CL\tau(\theta-\theta^*),
 \end{split}
 \end{equation}
 holds for any $\theta\in\theta^*+S(R_0,r_0)$, where $\alpha,\, L$ are positive constants depending on the model parameters, whose specific forms can be found in Appendix \ref{sec:proof_theorem_ordinal}, and function $\tau(\Delta)=\sqrt{\frac{\lambda_{\max}(\Sigma)(\log J+m)}{n}}$\\$\|\Delta_{1:p}\|_{\omega,2,1}+\sqrt{\frac{K(\log J+\log(2K))}{n}}\|\Delta\|_2$.
\end{lemma}
Under appropriate rate conditions, one can show that the first term at the R.H.S. of \eqref{eq:ordinal_rsc} is the positive dominating term when $\theta-\theta^*$ belongs to a restricted cone, and hence eventually leads to our final statistical error bound in Theorem \ref{thm:ordinal_upp}.

\section{Simulation Study}\label{sec:simulation}
In this section, we present simulation studies to validate the theoretical results and to evaluate our algorithms for both the multinomial-PU and ordinal-PU models.
We focus on the case-control setting in simulations while investigating the the single-training-set scenario in the real data experiments. 
\subsection{Validating Theoretical Guarantees}\label{sec:theory_sim}
First we experimentally validate the theoretical scaling of the mean squared error bounds presented in Section \ref{sec:theory}, w.r.t. the sparsity $s$, dimension $p$, number of categories $K$ and sample size $n$. For the multinomial model, we consider group sparsity where each row $\Theta_{j,:}\in \mathbb{R}^K$ consists of one group; and we focus on entry-wise sparsity for the ordinal model. Given a sparsity level $s$, the support sets are randomly chosen and the non-zero parameters are sampled from $\mathrm{U}([-1,-0.5]\cup[0.5 1])$. The intercepts are chosen such that an intercept-only model assigns equal probabilities to each category. We sample the feature vector $X$ from i.i.d. standard Gaussian distribution. We then generate the positive-unlabeled responses $z_i,\,1\leq i\leq n$ by randomly drawing $n_u=\frac{n}{2}$ unlabeled samples from the population, and draw $n_j = \frac{n}{2K}$ labeled samples from the population with true label $j$, $1\leq j\leq K$. The tuning parameter is set as $\lambda = c\sqrt{\frac{\log p}{n}}$ for appropriately chosen constant $c$ for each $s$ and $K$. The initialization is chosen as the MLE for the intercept-only models which have closed-form solutions.

Figure \ref{fig:sim1} presents the estimation mean squared errors of both models under different $s,\,p,\,n$ with $K=2$. The $x$-axis are the theoretical scalings w.r.t. $s,\,p,\,n$ from Theorem \ref{thm:mult_upp} and \ref{thm:ordinal_upp}: $\sqrt{\frac{s\log p +K}{n}}$ for multinomial parameters with group sparsity of size-$K$ groups, and $\sqrt{\frac{s\log p}{n}}$ for ordinal parameters with sparsity $s$. The straight and close lines validate these scalings. 
\begin{figure}[!ht]
\centering
  \subfigure[Multinomial]{\includegraphics[height = 5cm]{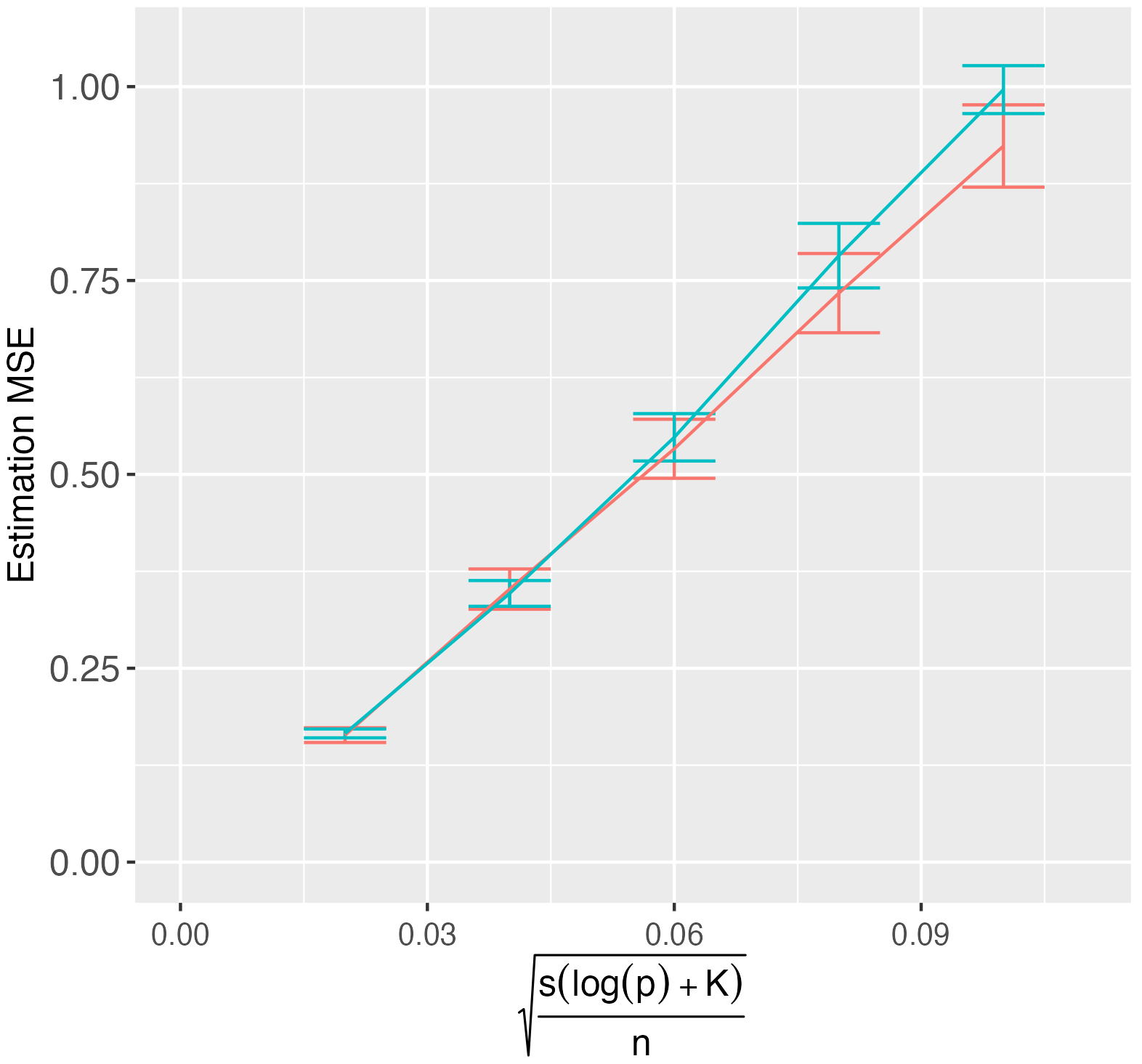}}
  \subfigure[Ordinal]{\includegraphics[height = 5cm]{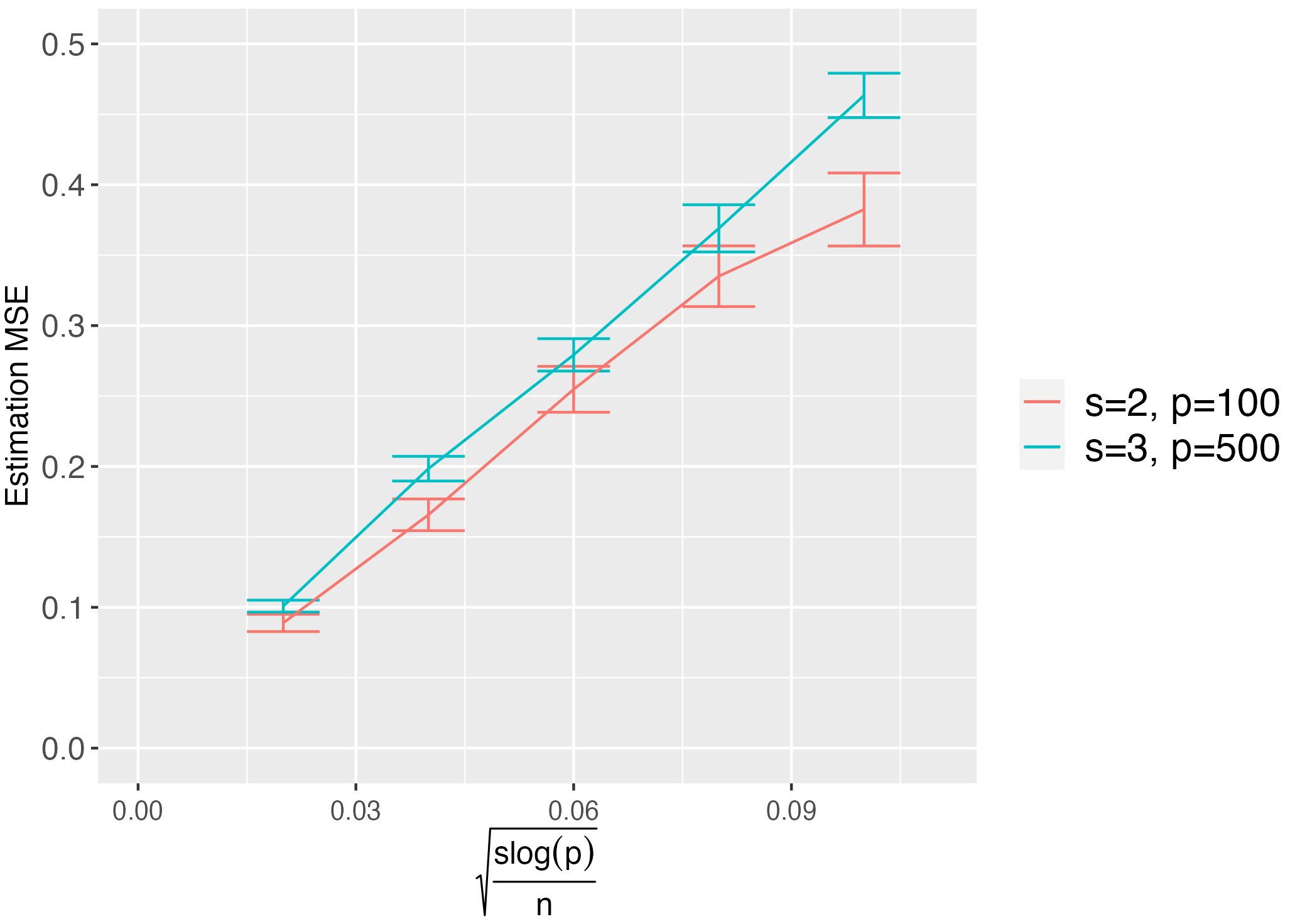}}
  \caption{MSE $\hat{E}\|\hat{\Theta}-\Theta^*\|_F$ for the multinomial model or $\hat{E}\|\hat{\theta}-\theta^*\|_2$ for the ordinal model, plotted against the theoretical rate $\sqrt{s(\log p+K)}{n}$ or $\sqrt{s\log p/n}$. The two lines correspond to different sparsities and dimensions, and they seem to align well to each other, validating our theoretical rate.}
  \label{fig:sim1}
\end{figure}
We also investigate how the mean squared errors depend on the number of positive categories $K$. We focus on $p=100$, $s = 2,3$, and sample size $n$ satisfying $\sqrt{\frac{s\log p}{n}}=0.02$, and the results are presented in Figure \ref{fig:sim2}. It turns out that although our theoretical scaling on $K$ is polynomial for both models, it may only be the case for the multinomial model as the total number of parameters is $Kp$, while the MSE for the ordinal model seems similar across different values of $K$.
\begin{figure}[!ht]
\centering
  \subfigure[Multinomial]{\includegraphics[height = 5cm]{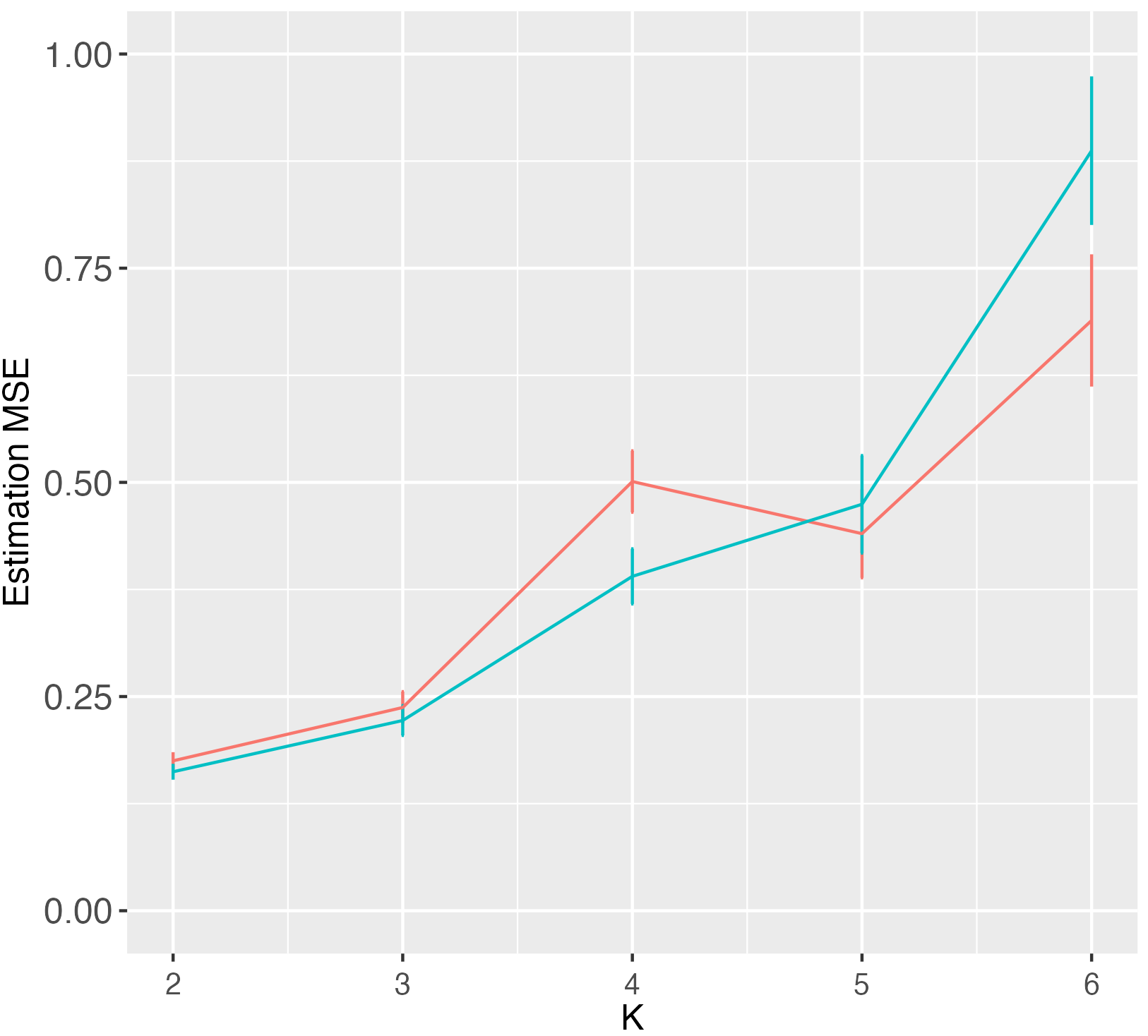}}
  \subfigure[Ordinal]{\includegraphics[height = 5cm]{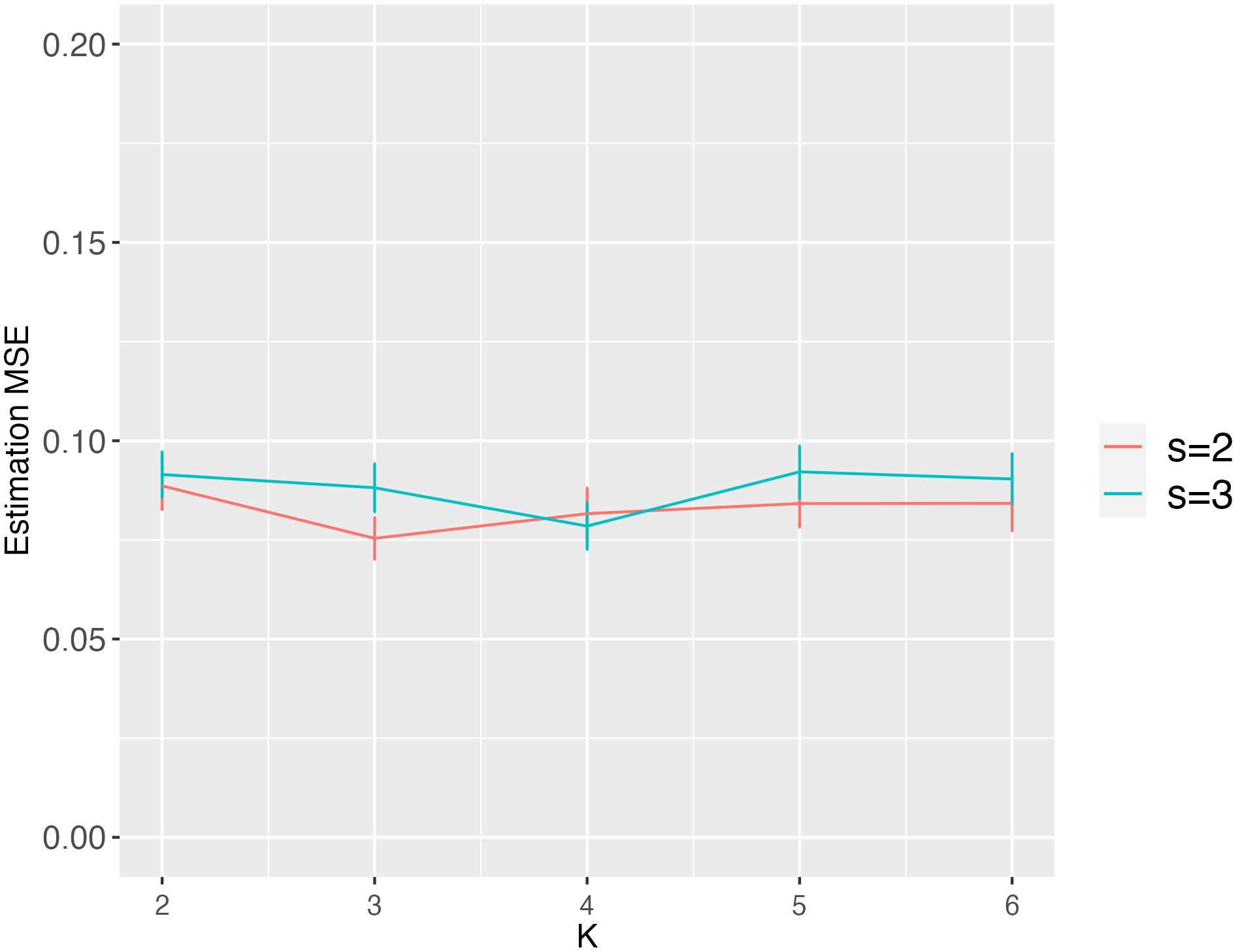}}
  \caption{MSE $\hat{E}\|\hat{\Theta}-\Theta^*\|_F$ for the multinomial model or $\hat{E}\|\hat{\theta}-\theta^*\|_2$ for the ordinal model, plotted against $K$. $p = 100$ and $n$ is chosen to satisfy $\sqrt{\frac{s\log p}{n}} = 0.02$. We can see the estimation error grows steadily as $K$ increases for the multinomial model, but does not change much for the ordinal model.}
  \label{fig:sim2}
\end{figure}
\subsection{Comparative Studies}
Here, we compare the prediction performance using our methods with the corresponding baselines which assume all unlabeled data are truly negative, and which directly minimize the $\ell_1$-penalized multinomial and ordinal loglikelihoods. We refer to our methods as the \emph{MN-PULasso} and \emph{ON-PULasso}, and the baselines as \emph{MN-Lasso} and \emph{ON-Lasso}. We also present the oracle prediction errors for reference, which are based on the true model parameters. Specifically, we compare the prediction errors of the fitted models / true models for predicting the true labels of a test data set of size $100$, and we investigate the effect of the true prevalence $\pi_j=\mathbb{P}(y=j),\,j=0,\dots,K$ of each category in the whole population and the sampling proportions $\frac{n_u}{n}$, $\frac{n_j}{n},\, j=1,\dots,K$. The model parameter settings are mostly the same as described in Section \ref{sec:theory_sim}, except that we set the intercept parameters to achieve different prevalence $\pi_j$, and each feature $X_j$ is now sampled from i.i.d. Gaussian distribution with variance $4$. We focus on $s=2,\,K=2,\,p=400,\,n=200$, and the tuning parameters are all chosen via 5-fold cross-validation.

We first fix $n_u=\frac{n}{2}=200$, $n_j = \frac{n}{2K}=50$ for $j>0$, while varying the true prevalence of different categories. However, directly setting specific values for the true prevalence is difficult, as the prevalence depends on the model parameters in a complex form. Instead, we set the intercept parameters carefully so that the true prevalence lies in an appropriate range. More details can be found in the Supplement. The resulting prevalence is estimated from the simulated data and we present the prevalence of positive samples $\pi_1+\pi_2 = 1-\pi_0$ as the x-axis of Figure \ref{fig:sim4}. 
A larger $\pi_1+\pi_2$ (smaller $\pi_0$) means that more unlabeled samples are in fact positive instead of being negative, and hence the baselines which treat those unlabeled samples as negative would suffer from significant bias. We also find that when the prevalence of different categories become more extreme, the oracle and our prediction errors become smaller since the prediction problem becomes easier with unbalanced test data. 
\begin{figure}[!ht]
\centering
  \subfigure[Multinomial]{\includegraphics[height = 5cm]{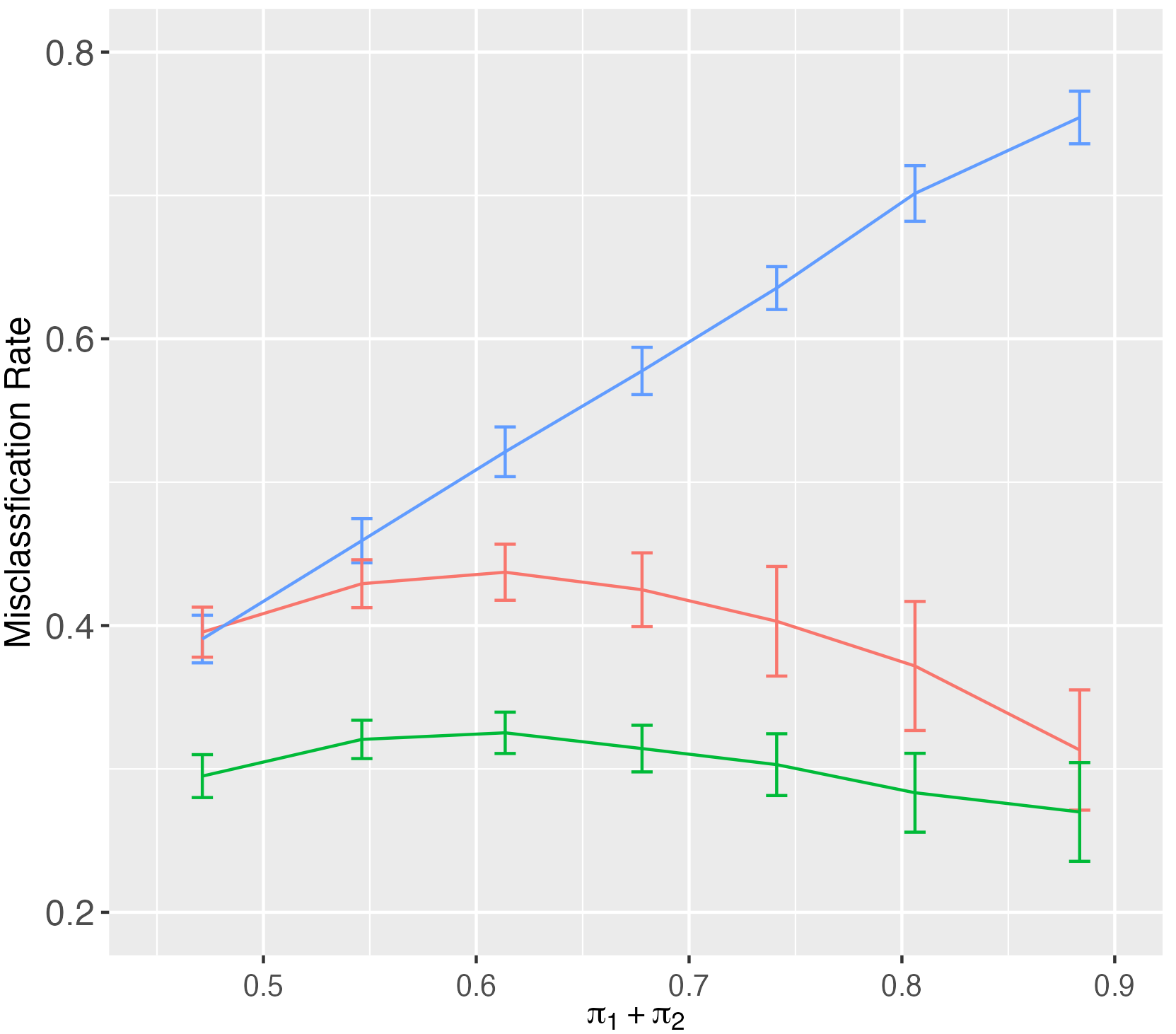}}
  \subfigure[Ordinal]{\includegraphics[height = 5cm]{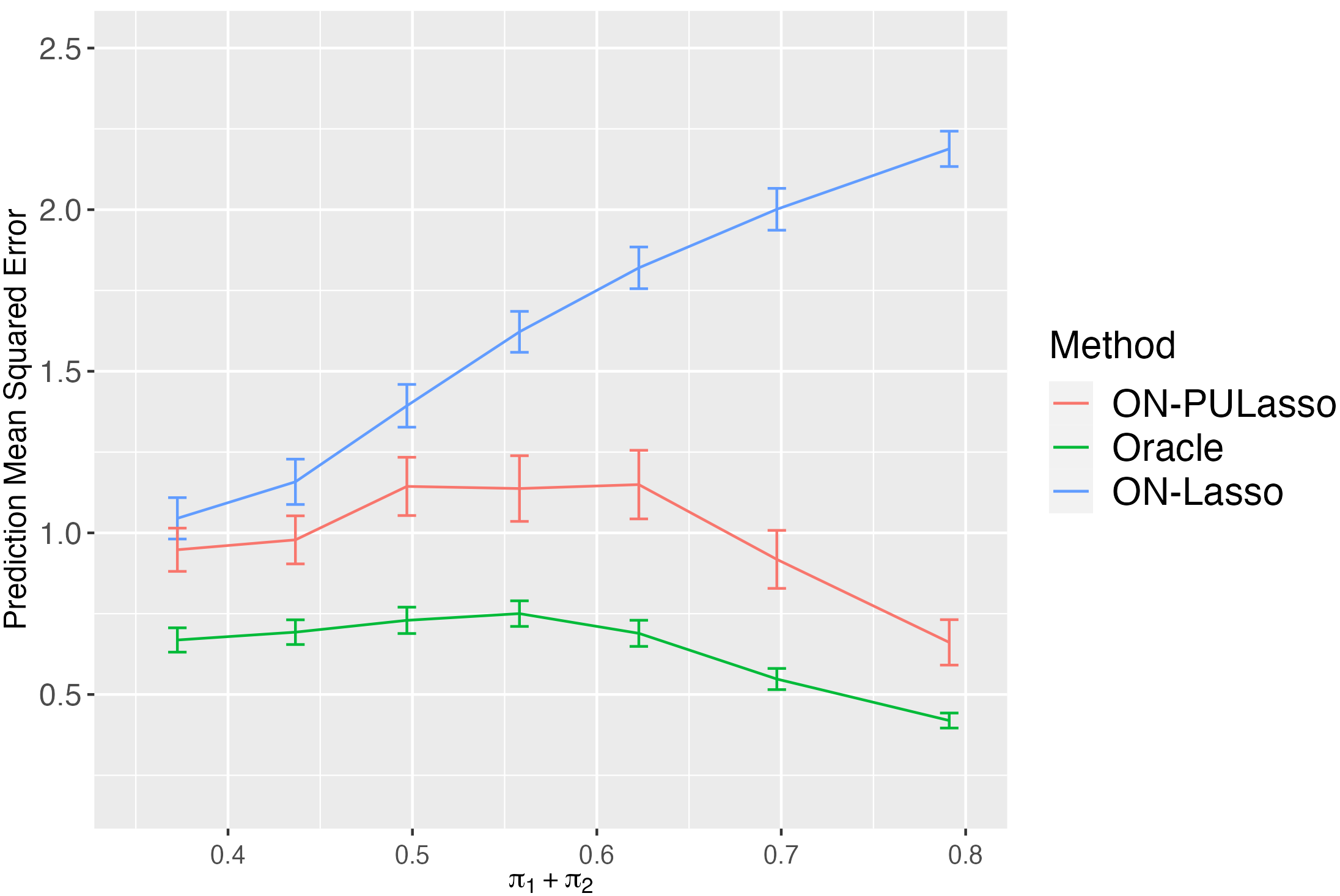}}
  \caption{Prediction performance comparisons between our methods (\emph{MN-PULasso} and \emph{ON-PULasso}), the baseline methods (\emph{MN-Lasso} and \emph{ON-Lasso}), and the oracle prediction based on the true models, plotted against the true prevalence of positive categories. For the multinomial model, the $y$-axis is the misclassification rates $\mathbb{P}_n(\hat{y}\neq y)$; for the ordinal model, the $y$-axis is the prediction mean squared error $\mathbb{E}_n(\hat{y}- y)^2$. 
  We can see our methods outperform the baseline methods that naively treat unlabeled data as zero, especially when the true prevalence of positive data increases.}
  \label{fig:sim4}
\end{figure}

In addition, we also investigate the effect of the sampling proportions. We set the intercepts of both models appropriately so that the true prevalence of all categories are similar (balanced data). More details can be found in the Supplement. We vary $\frac{n_u}{n}$ from $0.3$ to $0.9$, with $n_1=n_2 = \frac{n-n_u}{2}$. A larger $\frac{n_u}{n}$ means that there are more unlabeled samples and hence the estimation problem becomes more challenging. Figure \ref{fig:sim5} suggests that as $\frac{n_u}{n}$ increases, both our approaches and the baselines have larger prediction errors, but the baselines are more adversely affected by this issue.
\begin{figure}[!ht]
\centering
  \subfigure[Multinomial]{\includegraphics[height = 5cm]{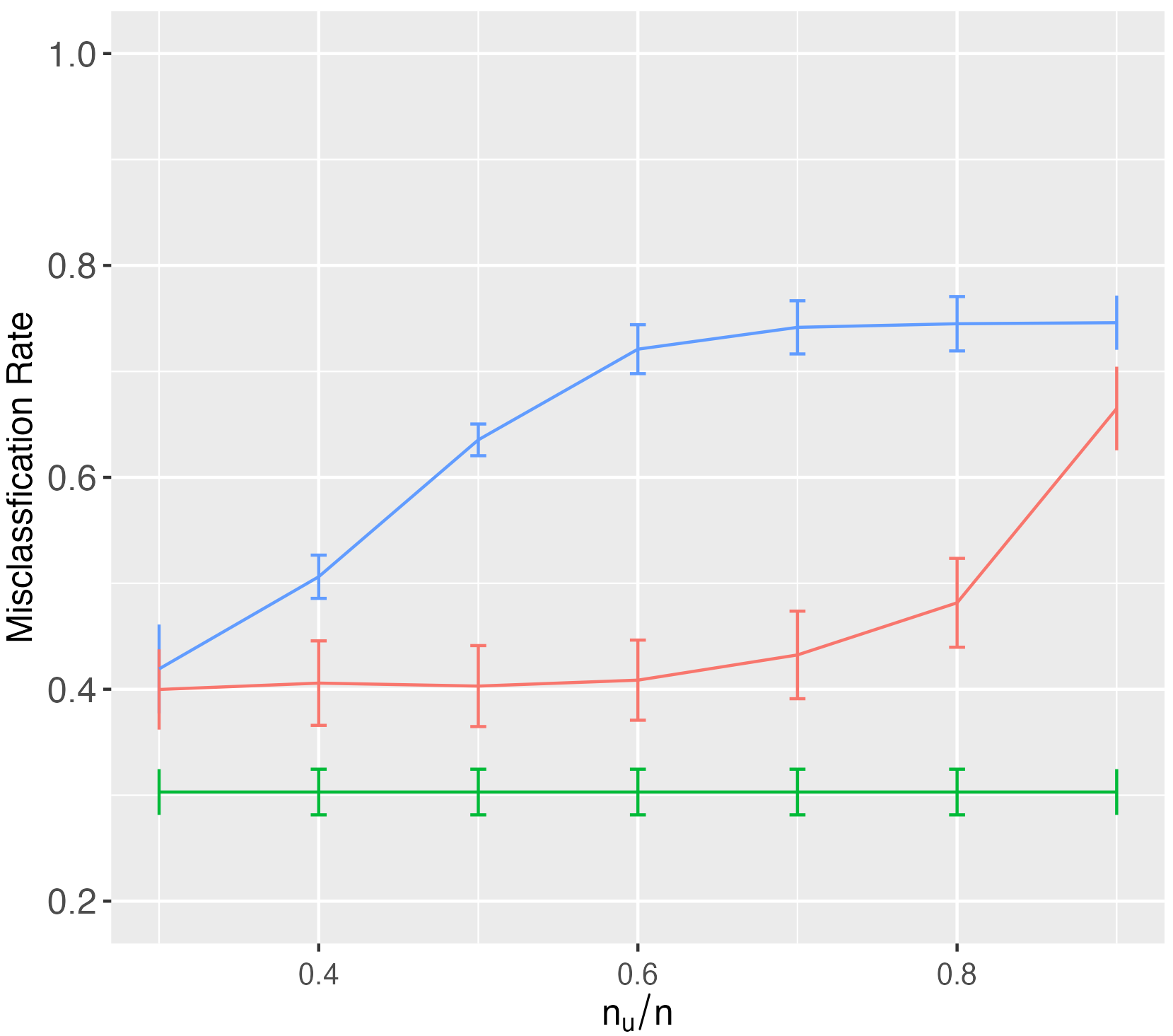}}
  \subfigure[Ordinal]{\includegraphics[height = 5cm]{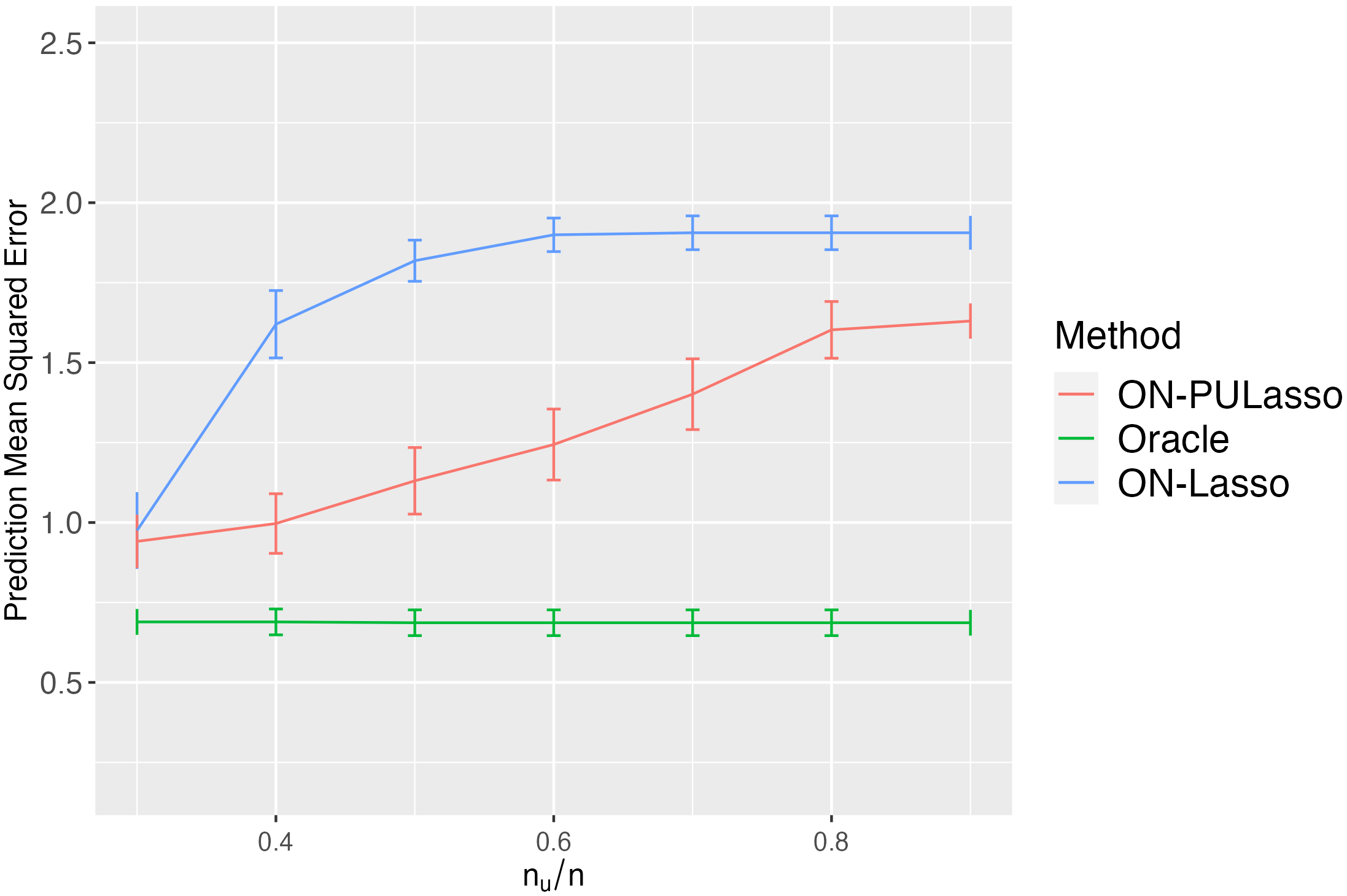}}
  \caption{Prediction performance comparisons between our approaches (\emph{MN-PULasso} and \emph{ON-PULasso}), the baselines (\emph{MN-Lasso} and \emph{ON-Lasso}), and the oracle prediction, plotted against the proportion of unlabeled samples. The problem is more challenging when there are more unlabeled samples, while the prediction errors of our approaches increase much more slowly than the baselines.}
  \label{fig:sim5}
\end{figure}

\section{Real Data Experiments}\label{sec:realdata}
In this section, we validate our approaches with two real data experiments. One data set has unordered labels and is suited to the multinomial PULasso approach; the other one has ordered labels and is used to validate the ordinal PULasso approach.
\subsection{Multinomial Application: Digit Recognition}\label{app:digit}
Digit recognition is an important problem in computer vision, which aims to classify the digits 0 through 9 from their images. However, it is usually expensive and time-consuming to obtain human-labeled samples, calling for a PU-learning approach for multi-class classification. Here, we would like to investigate the potential of our multinomial PULasso approach for solving this problem, since the images of different digits have no ordering pattern.  
We obtained the publicly available data ''Multiple Features Data Set''\footnote{https://archive.ics.uci.edu/ml/datasets/Multiple+Features} from the UCI Machine Learning Repository \citep{Dua:2019}, which consists of $n=2000$ handwritten digits from a Dutch utility map. The data is balanced (i.e. there are $200$ of each of the $10$ digits $0$ through $9$). There are $p=649$ features extracted from the images, including $76$ Fourier coefficients of the character shapes, $216$ profile correlations, $64$ Karhunen-Love coefficients, $240$ pixel averages in $2\times 3$ windows, $47$ Zernike moments, and $6$ morphological features. All samples in this data set are labeled correctly, but in order to validate our approach in PU settings, we manually contaminate the training set to make them positive and unlabeled. In particular, we contaminate the training set (70\% of the full data set) to imitate the PU data one might encounter in real applications, and apply both our multinomial PULasso approach and the baseline multinomial Lasso approach on the contaminated data; we then compare the prediction performance of these two methods on the clean test data. 

More specifically, with probability $1-\pi_i$, each of the digit $i$ in the training set were replaced with a $0$. This is the single-training-set scenario described in Appendix \ref{sec:multinom}, and the Multinomial PUlasso algorithm proposed there is used. To study the effect of $\pi_i$ and its misspecification, here we perform experiments on different values of $\pi_i$ for data generation and $\hat{\pi}_i$ for model fitting. We set all $\pi_i$ ($\hat{\pi}_i$) to be the same across $i\in \{1,\dots,9\}$. The regularization parameters are all chosen via $5$-fold cross validation on the training set. Then we measure the prediction accuracy on the test data set by predicting the class with the highest predicted probability for the each test sample, and report the percentage of correct predictions in Table \ref{tab:digit}. We can see that (i) MN-PULasso outperforms the baseline MN-Lasso when the probability of observing a labeled sample is lower ($\pi_i=0.4,\,0.6$) even when this probability is misspecified in our algorithm; (ii) For different levels of $\pi_i$, MN-PULasso with the correct $\hat{\pi}_i=\pi_i$ always performs the best. 
\begin{table}[!ht]
    \centering
    \scalebox{0.85}{
    \begin{tabular}{|c|c|c|c|c|c|c|c|c|c|}
 \hline
   & \multirow{2}{*}{MN-Lasso} & \multicolumn{5}{c|}{MN-PULasso}\\
   \cline{3-7}
   &&$\hat{\pi}_i=0.5$& $\hat{\pi}_i=0.6$&$\hat{\pi}_i=0.7$& $\hat{\pi}_i=0.8$& $\hat{\pi}_i=0.9$\\
 \hline
    $\pi_i=0.4$ & 26.17\% & 79.83\% & 66.50\% & 52.00\% & 43.50\% & 36.50\% \\
 \hline
 $\pi_i=0.6$ & 66.17\% & 88.50\% & 98.00\% & 91.33\% & 82.00\% & 72.00\%\\
 \hline
 $\pi_i=0.8$ & 92.33\% & 47.00\% & 82.50\% & 89.00\% & 97.33\% & 94.33\%\\
 \hline
\end{tabular}}
    \caption{Digit recognition data experiments: Multinomial PULasso and Multinomial Lasso accuracy on the test data set (30\%), both trained on the training data set (70\%). For any positive sample with label $i$ ($i>0$) in the training set, the probability of it being unlabeled is $1-\pi_i$. The input to our MN-PULasso method is $\hat{\pi}_i$ .We observe that our MN-PULasso with correct observational probability input $\hat{\pi}_i=\pi_i$ always performs the best; it is also reasonably robust against misspecification of $\pi_i$.}
    \label{tab:digit}
\end{table}

\subsection{Ordinal Application: Colposcopy Subjectivity}\label{app:colposcopy}
Colposcopies are used to examine the cervix, vagina, and vulva, typically done if abnormalities are noticed in a pap smear. Instead of having a medical expert looking at the patient in real time, a more efficient strategy is to record a video (digital colposcopy) for the medical expert to examine later, so that this frees time for medical practitioners to see more patients. An important question is whether the digital colposcopies have good enough qualities for the medical experts to make diagnosis.

To investigate this problem, we study the Quality Assessment of Digital Colposcopies Data Set \citep{fernandes2017transfer} from the UCI Machine Learning Repository \citep{Dua:2019}. This data set consists of $n=287$ digital colposcopies used to check for cervical cancer, and $p=62$ features were extracted from each of the digital colposcopies. $6$ medical experts looked at the digital colposcopies and rated the quality of them as either poor or good. 
We then give each colposcopy an ordinal ranking by counting how many of the medical experts rated the colposcopy as good quality, with a range from $0$ to $5$ in this data set, leading to $5$ positive classes and one negative class. The goal of our study here is to develop a model that uses the $p$ features to predict how many experts considered the image to be of good quality.

Just as with the digit recognition application, this particular colposcopy data set is not PU data. However, getting medical experts to label the quality of colposcopies is time-consuming and expensive, which could lead to huge amount of unlabeled samples in real applications. Therefore, developing a method for PU and ordinal data in this application is an important task, and we study the potential of our approach in this scenario by manually masking the data to make the samples positive and unlabeled. In particular, with probability $1-\pi_i$, each of the colposcopies with $i$ experts rating as good quality in the training set (70\% of the data set) were altered to have $0$ as its contaminated label, making $0$ the unlabeled class and all other classes labeled. As before, this was only done on the training data, so we can still use the testing set (the remaining 30\% of the data set) to compute accuracy. This is the single-training-set scenario described in Appendix \ref{sec:ordinal}, and the Ordinal PUlasso algorithm proposed there is used. In addition, we consider different values of $\pi_i$ for generating the PU data and $\hat{\pi}_i$ used in the traininng algorithm to study their effects. For simplicity, all $\pi_i$'s and $\hat{\pi}_i$'s are set as the same for all $0\leq i\leq 6$. 

We trained using $5$-fold cross validation on the training set. 
Because the labels are ordinal, predicting a sample with true label $6$ as zero is worse than predicting it as $5$. Hence we report the mean squared error to indicate the prediction performance instead of the percentage of correct classifications. As shown in Table~\ref{tab:colposcopy}, Ordinal PUlasso always outperform Ordinal Lasso which treats all unlabeled samples as negative ones, when the probability $\pi_i$ of correctly observing positively labeled data ranges from $0.6$ to $0.8$, even when $\pi_i$ is misspecified moderately.

\begin{table}[!ht]
    \centering
    \scalebox{0.85}{
    \begin{tabular}{|c|c|c|c|c|c|c|c|c|c|}
 \hline
   & \multirow{2}{*}{ON-Lasso} & \multicolumn{5}{c|}{ON-PULasso}\\
   \cline{3-7}
   &&$\hat{\pi}_i=0.5$& $\hat{\pi}_i=0.6$&$\hat{\pi}_i=0.7$& $\hat{\pi}_i=0.8$& $\hat{\pi}_i=0.9$\\
 \hline
    $\pi_i=0.6$ & 10.70 & 7.27 & 3.79 & 4.72 & 6.92 & 8.60 \\
 \hline
 $\pi_i=0.7$ & 9.36 & 7.27 & 7.27 & 5.53 & 6.28 & 9.01 \\
 \hline
 $\pi_i=0.8$ & 7.97 & 6.41 & 4.51 & 7.27 & 4.99 & 5.93\\
 \hline
\end{tabular}}
    \caption{Colposcopy quality assessment data experiments: Ordinal PULasso and Ordinal Lasso prediction MSE on the test data set (30\%), both trained on the training data set (70\%). For any positive sample with label $i$ ($i>0$) in the training set, the probability of it being unlabeled is $1-\pi_i$. The input to our ON-PULasso method is $\hat{\pi}_i$. We observe that for reasonably chosen $\hat{\pi}_i$, our ON-PULasso algorithm always have more accurate predictions than the baseline method ON-Lasso.}
    \label{tab:colposcopy}
\end{table}
\section{Discussion}\label{sec:disc}
In this paper, we focus on the high-dimensional classification problem with the positive and unlabeled (PU) data, where only some positive samples are labeled, a common situation arising in many applications. Such data sets posit unique optimization and statistical challenges due to the non-convex landscape of the log-likelihood loss. Going beyond prior works that focus on binary classification, here we are interested in the setting with multiple positive categories, which is more general but also magnifies the non-convexity issues. In particular, we propose a multinomial PU model and an ordinal PU model for unordered and ordered labels, respectively, with accompanying algorithms to estimate the models. Despite the challenging non-convexity of the problems, especially for the ordinal model, we manage to show the algorithmic convergence and characterize the statistical error bound for a reasonable initialization. A series of simulation and real data studies suggest the practical usefulness and application potential of our proposed models and methods.

There are also a number of open problems that might be worth investigating in the future. Although our theory and empirical studies have validated the efficacy of our non-convex approaches, it may still be of interest to develop convex methods and compare their performance with our approaches, probably leveraging the idea of moment methods in \cite{song2020convex}. Our current models are all parametric, while it is also possible to consider extensions to semi-parametric models, or to incorporate more flexible non-parametric machine learning algorithms such as deep neural networks. Furthermore, there are a number of other applications in biomedical engineering and others that we can adapt our framework to. 

\clearpage
\appendix
\section{Proof of Theorem~\ref{thm:mult_upp}}\label{sec:proof_theorem_mult}
For simplicity, we will omit $\cL_n^{\MN}(\Theta,0_{K\times 1})$ to $\cL_n^{\MN}(\Theta)$ in the following. We first present two major supporting lemmas for proving Theorem~\ref{thm:mult_upp}.
\begin{lemma}[Deviation Bound under the Multinomial-PU Model]\label{lem:multinomial_dev_bnd}
If the data set $\{(x_i,z_i)\}_{i=1}^n$ is generated by the multinomial-PU model under the case-control setting, and Assumptions \ref{assump:subgauss_row}-\ref{assump:mult_region_bound} hold, then
\begin{equation}
     \|\nabla\cL_n^{\MN}(\Theta^*)\|_{\omega^{-1},2,\infty}\leq C\lambda^{\frac{1}{2}}_{\max}(\Sigma)\max_{1\leq j\leq J}\sqrt{\frac{c_j(m+\log J)}{n}},
\end{equation}
with probability at least $1-\exp\{-c(\log J+m)\}$.
\end{lemma}
\begin{lemma}[Restricted Strong Convexity under the Multinomial-PU Model]\label{lem:RSC}
If the data set $\{(x_i,z_i)\}_{i=1}^n$ is generated by the multinomial-PU model under the case-control setting, and Assumptions \ref{assump:subgauss_row}-\ref{assump:mult_region_bound} hold, then with probability at least $$1-\exp\{-cR^{*2}K\lambda_{\max}(\Sigma)\min_jr_j^{-1}(\log J+m)\},$$
\begin{equation}
\begin{split}
    &\left\langle\nabla \cL_n^{\MN}(\Theta)-\nabla \cL(\Theta^*),\Theta-\Theta^*\right\rangle\\
    \geq &\alpha\|\Theta-\Theta^*\|_2^2-\tau_1\frac{\log J+m}{n}\|\Theta-\Theta^*\|_{\omega,2,1}^2-\tau_2\sqrt{\frac{\log J+m}{n}}\|\Theta-\Theta^*\|_{\omega,2,1}
\end{split}
\end{equation}
holds for any $\Theta\in \{\Theta':\max_{1\leq i\leq K}\|\Theta^*_{:,i}-\Theta'_{:,i}\|_1\leq R_0\}$, where $\alpha=\frac{h^{\MN}(R_0)\lambda_{\min}(\Sigma)}{2}$, $\tau_1=C\alpha$, $\tau_2=C(K+C_xR^*\sqrt{K})\lambda^{\frac{1}{2}}_{\max}(\Sigma)$.
\end{lemma}
To further make use of Lemma \ref{lem:RSC}, here we present another lemma that guarantees the curvature term $\alpha>0$ and shows that the dominating slack term is $\tau_2\sqrt{\frac{\log J+m}{n}}\|\Theta-\Theta^*\|_{\omega,2,1}$.
\begin{lemma}\label{lem:MN_region_results}
    Under Assumptions \ref{assump:subgauss_row} and \ref{assump:mult_region_bound}, for any $\Theta\in \Theta^* + \mathbb{B}_{1,\infty}(R_0)$, we have
    \begin{equation*}
        h^{\MN}(R_0)>0,\quad \tau_1\frac{\log J+m}{n}\|\Theta-\Theta^*\|_{\omega,2,1}^2\leq \tau_2\sqrt{\frac{\log J+m}{n}}\|\Theta-\Theta^*\|_{\omega,2,1}.
    \end{equation*}
\end{lemma}
The proofs of Lemmas~\ref{lem:multinomial_dev_bnd}-\ref{lem:MN_region_results} will be presented in Section~\ref{sec:proof-MN-lem}. Now we are ready to prove Theorem \ref{thm:mult_upp}.
\begin{proof}[Proof of Theorem \ref{thm:mult_upp}]
    Since $\widehat{\Theta}$ is a stationary point of $$\min_{\Theta-\Theta^*\in\bbB_{1,\infty}(R_0)}\cL_n^{\MN}(\Theta)+\lambda\|\Theta\|_{\omega,2,1},$$ for any $\Theta\in \bbR^{p\times K}$ such that $\max_{1\leq i\leq K}\|\Theta_{:,i}-\Theta^*_{:,i}\|_1\leq R_0$, we have
    \begin{equation}
        \begin{split}
            \langle\nabla\cL_n^{\MN}(\widehat{\Theta})+\nabla P_{\lambda}^{\MN}(\widehat{\Theta}),\Theta-\widehat{\Theta}\rangle\geq 0,
        \end{split}
    \end{equation}
    where $\nabla \lambda\|\Theta\|_{\omega,2,1}$ is any sub-gradient of $\lambda\|\Theta\|_{\omega,2,1}$. Since $\Theta^*$ is also in the feasible set, 
    \begin{equation}
        \begin{split}
            \langle\nabla\cL_n^{\MN}(\widehat{\Theta})+\nabla P_{\lambda}^{\MN}(\widehat{\Theta}),\Theta^*-\widehat{\Theta}\rangle\geq& 0,\\
            \langle\nabla\cL_n^{\MN}(\widehat{\Theta})-\nabla\cL_n^{\MN}(\Theta^*),\widehat{\Theta}-\Theta^*\rangle\leq&\langle -\nabla \cL_n^{\MN}(\Theta^*)-\nabla P_{\lambda}^{\MN}(\widehat{\Theta}),\widehat{\Theta}-\Theta^*\rangle.
        \end{split}
    \end{equation}
    Let $\widehat{\Delta}=\widehat{\Theta}-\Theta^*$, $S=\{j:\|\Theta_{\cG_j}\|>0\}$, $\cG_S=\cup_{j\in S}\cG_j$. By Lemma \ref{lem:multinomial_dev_bnd}, Lemma \ref{lem:RSC},
    \begin{equation}
        \begin{split}
            &\alpha \|\widehat{\Delta}\|_2^2-\tau_1\frac{\log J+m}{n}\|\widehat{\Delta}\|_{\omega,2,1}^2-\tau_2\sqrt{\frac{\log J+m}{n}}\|\widehat{\Delta}\|_{\omega,2,1}\\
            \leq& C\tau_2\sqrt{\frac{\log J+m}{n}}\|\widehat{\Delta}\|_{\omega,2,1}+\langle \nabla P_{\lambda}^{\MN}(\widehat{\Theta}),\Theta^*-\widehat{\Theta}\rangle.
        \end{split}
    \end{equation}
    Invoking Lemma \ref{lem:MN_region_results}, one has $\tau_1\frac{\log J+m}{n}\|\widehat{\Delta}\|_{\omega,2,1}^2\leq \tau_2\sqrt{\frac{\log J+m}{n}}\|\widehat{\Delta}\|_{\omega,2,1}$. Also note that 
    \begin{equation}
        \begin{split}
            \langle \nabla P_{\lambda}^{\MN}(\widehat{\Theta}),\Theta^*-\widehat{\Theta}\rangle\leq& P_{\lambda}^{\MN}(\Theta^*)-P_{\lambda}^{\MN}(\widehat{\Theta})\\
            \leq&\lambda \|\widehat{\Delta}_{\cG_S}\|_{\omega_S,2,1}-\lambda\|\widehat{\Delta}_{\cG_{S^c}}\|_{\omega_{S^c},2,1},
        \end{split}
    \end{equation}
    which implies
    \begin{equation}
    \begin{split}
        \alpha \|\widehat{\Delta}\|_2^2\leq &\frac{3\lambda}{2}\|\widehat{\Delta}_{\cG_S}\|_{\omega_S,2,1}-\frac{\lambda}{2}\|\widehat{\Delta}_{\cG_S^c}\|_{\omega_{S^c},2,1},\\
        \|\widehat{\Delta}\|_2\leq&\frac{3\|\omega_S\|_2\lambda}{h^{\MN}(R_0)\lambda_{\min}(\Sigma)},\\
        \|\widehat{\Delta}\|_{\omega,2,1}\leq& 4\|\widehat{\Delta}_{\cG_S}\|_{\omega_S,2,1}\leq\frac{12\|\omega_S\|^2_2\lambda}{h^{\MN}(R_0)\lambda_{\min}(\Sigma)}.
    \end{split}
    \end{equation}
as long as $\lambda\geq C\tau_2\sqrt{\frac{\log J+m}{n}}$, with probability at least \begin{equation}
    1-\exp\{-cR^{*2}K\lambda_{\max}(\Sigma)\min_jr_j^{-1}(\log J+m)\}-\exp\{-c(\log J+m)\}.
\end{equation}
\end{proof}
\subsection{Proofs of Lemmas~\ref{lem:multinomial_dev_bnd}, \ref{lem:RSC} and \ref{lem:MN_region_results}}\label{sec:proof-MN-lem}
\begin{proof}[Proof of Lemma \ref{lem:multinomial_dev_bnd}]
First note that for any $1\leq j\leq p$, $1\leq k\leq K$,
\begin{equation}
\begin{split}
    (\nabla \cL_n^{\MN}(\Theta^*))_{jk}=&-\frac{1}{n}\sum_{i=1}^n \left[\sum_{k'=1}^K (\delta_{ik'}-\bbE(\delta_{ik'}|x_i))\frac{\partial \eta_{ik'}}{\partial \Theta^*_{jk}}\right]\\
    =&-\frac{1}{n}\sum_{i=1}^n X_{ij}\left[\sum_{k'=1}^K \epsilon_{ik'}\left(\ind{k'=k}-\frac{e^{x_i^\top\Theta_k}}{1+\sum_{m=1}^Ke^{x_i^\top\Theta_m}}\right)\right],
\end{split}
\end{equation}
where $\delta_{ik}=\ind{z_i=k}$, $\epsilon_{ik'}=\delta_{ik'}-\bbE(\delta_{ik'}|x_i)$. Let $$\widetilde{\epsilon}_{ik}=\sum_{k'=1}^K \epsilon_{ik'}\left(\ind{k'=k}-\frac{e^{x_i^\top\Theta_k}}{1+\sum_{m=1}^Ke^{x_i^\top\theta_m}}\right),$$ then $\{\widetilde{\epsilon}_{ik}\}_{i=1}^n$ are independent mean 0 random variables lying on $(-2,2)$. Here we have $|\widetilde{\epsilon}_{ik}|\leq 2$ since $$\left|\sum_{k'=1}^K \delta_{ik'}\left(\ind{k'=k}-\frac{e^{x_i^\top\Theta_k}}{1+\sum_{m=1}^Ke^{x_i^\top\theta_m}}\right)\right|\leq 1.$$ Thus we can write
\begin{equation}
    \begin{split}
        \left\|(\nabla \cL_n^{\MN}(\Theta^*))_{\cR_j,\cC_j}\right\|_F=\left\|\frac{1}{n}X^\top_{\cR_j}\widetilde{\epsilon}_{:,\cC_j}\right\|_F\leq\frac{\sqrt{c_j}}{n}\max_{k\in \cC_j}\left\|X^\top_{\cR_j}\widetilde{\epsilon}_{:,k}\right\|_2.
    \end{split}
\end{equation}
Now we bound $\|X^\top_{\cR_j}\widetilde{\epsilon}_{:,k}\|_2$ for each particular $k\in \cC_j$ w.h.p., conditioning on $X$. First note that 
\begin{equation}
    \begin{split}
        \bbE(\|X^\top_{\cR_j}\widetilde{\epsilon}_{:,k}\|_2|X)\leq&\left(\bbE(\|X^\top_{\cR_j}\widetilde{\epsilon}_{:,k}\|_2^2|X)\right)^{\frac{1}{2}}\\
        =&\left(\sum_{i=1}^n\bbE(\widetilde{\epsilon}^2_{i,k}|X)(X_{\cR_j}X^\top_{\cR_j})_{ii}\right)^{\frac{1}{2}}\\
        \leq &2\sqrt{\tr(X_{\cR_j}X^\top_{\cR_j})}\\
        \leq&2\sqrt{r_j}\|X_{\cR_j}\|.
    \end{split}
\end{equation}
Meanwhile, let $g(u)=\|X^\top_{\cR_j}u\|_2$ be a function from $\bbR^n$ to $\bbR$, then $g$ is convex and $\|X_{\cR_j}\|$-Lipschitz: 
\begin{equation}
    \left|g(u)-g(v)\right|\leq \|X^{\top}_{\cR_j}(u-v)\|_2\leq \|X_{\cR_j}\|\|u-v\|_2.
\end{equation}
Applying Talagrand's contraction inequality (Theorem 5.2.16 in \cite{vershynin2018high}) leads to 
\begin{equation}\label{eq:dev-bnd-term-expt}
    \left\|\|X^\top_{\cR_j}\widetilde{\epsilon}_{:,k}\|_2-\bbE(\|X^\top_{\cR_j}\widetilde{\epsilon}_{:,k}\|_2|X)\right\|_{\psi_2}\leq C\|X_{\cR_j}\|,
\end{equation}
which further implies that
\begin{equation}\label{eq:dev-bnd-term}
    \bbP(\|X^\top_{\cR_j}\widetilde{\epsilon}_{:,k}\|_2-\bbE(\|X^\top_{\cR_j}\widetilde{\epsilon}_{:,k}\|_2)>C\|X_{\cR_j}\|\sqrt{n}t_j)\leq \exp\{-nt_j^2\},
\end{equation}
for any $t_j>0$. Combine \eqref{eq:dev-bnd-term-expt}, \eqref{eq:dev-bnd-term}, and take a union bound over $k\in \cC_j$ and $j\in [J]$, one can show that with probability at least $1-\sum_{j=1}^J c_j\exp\{-nt_j^2\}$,
\begin{equation}\label{eq:dev_bnd_X_norm}
\begin{split}
    \|\nabla\cL_n^{\MN}(\Theta)\|_{\omega^{-1},2,\infty}=&\max_{1\leq j\leq J}\frac{1}{\omega_j}\|(\nabla \cL_n^{\MN}(\Theta^*))_{\cR_j,\cC_j}\|_F\\
    \leq&\max_{1\leq j\leq J}\frac{\sqrt{c_j}\|X_{\cR_j}\|}{\omega_j \sqrt{n}}\left(2\sqrt{\frac{r_j}{n}}+Ct_j\right)
\end{split}
\end{equation}
Now we provide a probabilistic bound for $\|X_{R_j}\|$. Due to Assumption \ref{assump:subgauss_row}, $X_{\cR_j}\Sigma_{\cR_j,\cR_j}^{-\frac{1}{2}}\in \bbR^{n\times r_j}$ has independent sub-Gaussian rows with sub-Gaussian parameter $\sigma\lambda^{-\frac{1}{2}}_{\min}(\Sigma_{R_j,R_j})\leq C$ and covariance $I_{r_j}$. One can apply bounds on spectral norm of matrices with independent isotropic sub-Gaussian rows (Theorem 5.39 in \cite{vershynin2010introduction}) on $X_{\cR_j}\Sigma_{\cR_j,\cR_j}^{-\frac{1}{2}},1\leq j\leq J$ and get the following:
\begin{equation}\label{eq:X_norm_bnd}
    \|X_{\cR_j}\|\leq \lambda^{\frac{1}{2}}_{\max}(\Sigma)\|X_{\cR_j}\Sigma_{\cR_j,\cR_j}^{-\frac{1}{2}}\|\leq \lambda^{\frac{1}{2}}_{\max}(\Sigma)C(\sqrt{n}+\sqrt{r_j})\leq C\lambda^{\frac{1}{2}}_{\max}(\Sigma)\sqrt{n},
\end{equation}
holds for $1\leq j\leq J$ with probability at least $1-2J\exp\{-cn\}\geq 1-\exp\{-cn\}$, if $C>0$ is chosen appropriately in \eqref{eq:X_norm_bnd}. 
Let $t_j=\sqrt{\frac{m+\log (2c_jJ)}{n}}$ in \eqref{eq:dev_bnd_X_norm}, then we have
\begin{equation}
\begin{split}
    \|\nabla\cL_n^{\MN}(\Theta)\|_{\omega^{-1},2,\infty}\leq& C\lambda^{\frac{1}{2}}_{\max}(\Sigma)\max_{1\leq j\leq J}\frac{\sqrt{c_j}}{\omega_j}\sqrt{\frac{m+\log (2c_jJ)}{n}}\\
     \leq&C\lambda^{\frac{1}{2}}_{\max}(\Sigma)\max_{1\leq j\leq J}\sqrt{\frac{c_j(m+\log J)}{n}},
\end{split}
\end{equation}
with probability at least $$1-\exp\{-(\log J+m)\}-\exp\{-cn\}\geq 1-\exp\{-c(\log J+m)\},$$
where we have applied Assumption~\ref{assump:scaling}: $\log J,m\leq Cn$.
\end{proof}

\begin{proof}[Proof of Lemma~\ref{lem:RSC}]
Let $\Delta=\Theta-\Theta^*$. First note that
\begin{equation}
    \nabla \cL_n^{\MN}(\Theta)=\frac{1}{n}\sum_{i=1}^n x_i(\nabla A(\eta_i)-\delta_i)^\top \nabla f(\Theta^\top x_i),
\end{equation}
where $\delta_i\in \bbR^K$ with $\delta_{ik}=\ind{z_i=k}$.
Hence,
\begin{equation}
    \begin{split}
        &\left\langle\nabla \cL_n^{\MN}(\Theta)-\nabla \cL(\Theta^*),\Theta-\Theta^*\right\rangle\\
        =&\frac{1}{n}\sum_{i=1}^n x_i^\top (\Theta-\Theta^*)\nabla f(\Theta^\top x_i)^\top \left[\nabla A(f(\Theta^\top x_i))-\nabla A(f(\Theta^{*\top}x_i))\right]\\
        &-\frac{1}{n}\sum_{i=1}^n x_i^\top (\Theta-\Theta^*)\left(\nabla f(\Theta^\top x_i)-\nabla f(\Theta^{*\top}x_i)\right)^\top \epsilon_i\\
        :=&\mathrm{\RNum{1}}(\Delta)-\mathrm{\RNum{2}}(\Delta),
    \end{split}
\end{equation}
where $\epsilon_i=\delta_i-\nabla A(f(\Theta^{*\top}x_i))$.
We will provide a lower bound for $\mathrm{\RNum{1}}(\Delta)$ and concentrate $\mathrm{\RNum{2}}(\Delta)$ around $0$.
\begin{enumerate}
    \item Lower bounding $\mathrm{\RNum{1}}(\Delta)$.\\
    The following lemma provides a lower bound for $\mathrm{\RNum{1}}(\Delta)$ in terms of $\frac{1}{n}\sum_{i=1}^n \|\Delta^\top x_i\|_2^2$.
    \begin{lemma}\label{lem:mult_rsc_term1}
   As long as $\max_{1\leq i\leq K}\|\Theta_{:,i}-\Theta^*_{:,i}\|_1\leq R_0$ and $h^{\MN}(R_0)>0$ where $h$ is defined in \eqref{eq:mult_h_def}, it is guaranteed that
\begin{equation}
   \mathrm{\RNum{1}}(\Delta) \geq \frac{h^{\MN}(R_0)}{n}\sum_{i=1}^n \|\Delta^\top x_i\|_2^2.
\end{equation}
    \end{lemma}
    In the following we prove that
    \begin{equation}\label{eq:rec1}
        \frac{h^{\MN}(R_0)}{n}\sum_{i=1}^n \|\Delta^\top x_i\|_2^2\geq \alpha\|\Delta\|_2^2-\tau_1\frac{\log J+m}{n}\|\Delta\|_{\omega,2,1}^2.
    \end{equation}
    for $\alpha=\frac{h^{\MN}(R_0)\lambda_{\min}(\Sigma)}{2},\tau_1=C\alpha$ where $C>0$ is a constant. If $\|\Delta\|_{\omega,2,1}\geq \sqrt{\frac{\alpha n}{\tau_1(\log J+m)}}\|\Delta\|_2$, then the R.H.S. of \eqref{eq:rec1} is non-positive and \eqref{eq:rec1} holds trivially. Thus it suffices to prove 
    \begin{equation}\label{eq:rec2}
        \inf_{\|\Delta\|_{\omega,2,1}\leq \sqrt{\rho},\|\Delta\|_2=1}\frac{1}{n}\sum_{i=1}^n \|\Delta^\top x_i\|_2^2\geq \frac{\lambda_{\min}(\Sigma)}{2},
    \end{equation}
     where $\rho=\frac{\alpha n}{\tau_1(\log J+m)}$. Since $\bbE(\frac{1}{n}\sum_{i=1}^n \|\Delta^\top x_i\|_2^2)\geq \lambda_{\min}(\Sigma)\|\Delta\|_2^2$, \eqref{eq:rec2} can be implied by 
      \begin{equation}\label{eq:rec_uniform_dev}
        \sup_{\Delta\in \bbB_{\omega,2,1}(\sqrt{\rho})\cap \bbB_{2}(1)}\left|\frac{1}{n}\sum_{i=1}^n\left(\|\Delta^\top x_i\|_2^2-\bbE(\|\Delta^\top x_i\|_2^2)\right)\right|\leq  \frac{\lambda_{\min}(\Sigma)}{2}.
    \end{equation}
     The following Lemma shows the connection between $\bbB_{\omega,2,1}(\sqrt{\rho})\cap \bbB_{2}(1)$ and sparse set:
    \begin{lemma}\label{lem:l1_l0_balls}
    For any $\rho>0$, 
    \begin{equation}
        \bbB_{\omega,2,1}(\sqrt{\rho})\cap \bbB_{2}(1)\subset \left(1+\frac{2}{\min_j \omega_j}\right)\mathrm{cl}\{\mathrm{conv}\{\bbB_{\cG,0}(\lfloor \rho\rfloor)\cap \bbB_2(1)\}\},
    \end{equation}
    where $\bbB_{\cG,0}(\lfloor \rho\rfloor)=\{U\in \bbR^{p\times K}: |\{j:\|U_{\cR_j,\cC_j}\|_2>0\}|\leq \lfloor \rho\rfloor\}$.
    \end{lemma}
    By Lemma \ref{lem:l1_l0_balls}, one can show that
    \begin{equation}
        \begin{split}
            &\sup_{\Delta\in \bbB_{\omega,2,1}(\sqrt{\rho})\cap \bbB_{2}(1)}\left|\frac{1}{n}\sum_{i=1}^n\left(\|\Delta^\top x_i\|_2^2-\bbE(\|\Delta^\top x_i\|_2^2)\right)\right|\\
            \leq&\left(1+\frac{2}{\min_j \omega_j}\right)\sup_{\Delta\in \mathrm{conv}\{\bbB_{\cG,0}(\lfloor \rho\rfloor)\cap \bbB_2(1)\}}\left|\tr\left[\Delta^\top \left(\frac{1}{n}X^\top X-\Sigma\right)\Delta\right]\right|.
        \end{split}
    \end{equation}
    For any $\Delta\in \mathrm{conv}\{\bbB_{\cG,0}(\lfloor \rho\rfloor)\cap \bbB_2(1)\}$, there exists $\Delta_1,\Delta_2,\dots,\Delta_k\in \bbB_{\cG,0}(\lfloor \rho\rfloor)\cap \bbB_2(1)$, $\beta_1,\dots,\beta_k\geq 0$ and $\sum_{i=1}^k\beta_i=1$ such that $\Delta=\sum_{i=1}^k\beta_i\Delta_i$. Define function $g(U,U)=\tr\left[U^\top \left(\frac{1}{n}X^\top X-\Sigma\right)U\right]$. We can prove that
    \begin{equation}
        \begin{split}
            &\left|g(\Delta,\Delta)\right|\\
            =&\left|g(\sum_{i=1}^k\beta_i\Delta_i,\sum_{i=1}^k\beta_i\Delta_i)\right|\\
            =&\left|\sum_{i,j=1}^k\beta_i\beta_j\tr\left[\Delta_i^\top\left(\frac{1}{n}X^\top X-\Sigma\right)\Delta_j\right]\right|\\
            \leq&\left|\frac{1}{2}\sum_{i,j=1}^k\beta_i\beta_j\left(g(\Delta_i+\Delta_j,\Delta_i+\Delta_j)-g(\Delta_i,\Delta_i)-g(\Delta_j,\Delta_j)\right)\right|\\
            \leq&3\sup_{U\in \bbB_{\cG,0}(2\lfloor \rho\rfloor)\cap \bbB_2(1)}\left|g(U,U)\right|\sum_{i,j=1}^k\beta_i\beta_j\\
            =&3\sup_{U\in \bbB_{\cG,0}(2\lfloor \rho\rfloor)\cap \bbB_2(1)}\left|g(U,U)\right|.
        \end{split}
    \end{equation}
    Thus now we only need to upper bound $\sup_{\Delta\in \bbB_{\cG,0}(2\lfloor \rho\rfloor)\cap \bbB_2(1)}\left|g(\Delta,\Delta)\right|$. For any $S\subset [J]$ with $|S|=2\lfloor \rho\rfloor$, let $\cK(\cG_S)=\{U\in \bbR^{p\times K}:\|U\|_2= 1, \|U_{\cG_j}\|=0\text{ for }j\notin S\}$. Then one can show that
    \begin{equation}
        \sup_{\Delta\in \bbB_{\cG,0}(2\lfloor \rho\rfloor)\cap \bbB_2(1)}\left|g(\Delta,\Delta)\right|= \sup_{S\subset[J], |S|=2\lfloor \rho\rfloor}\sup_{\Delta\in \cK(\cG_S)}\left|g(\Delta,\Delta)\right|,
    \end{equation}
    For any such $S$, by the upper bound of covering numbers of the sphere (Lemma 5.2 in \cite{vershynin2010introduction}), there exists a $\frac{1}{10}$-net $\cN_S$ of $\cK(\cG_S)$, w.r.t. $\ell_2$ norm, with size $|\cN_S|\leq 21^{\sum_{j\in S}c_jr_j}\leq 21^{2\rho m}$. We can bound $\sup_{\Delta\in \cK(\cG_S)}\left|g(\Delta,\Delta)\right|$ in terms of $\sup_{\Delta\in \cN_S}\left|g(\Delta,\Delta)\right|$ as follows: $\forall U\in \cK(\cG_S)$, $\exists U'\in \cN_S, \varepsilon=U-U'$ such that $\|\varepsilon\|_2\leq \frac{1}{10}$, thus
    \begin{equation}
        \begin{split}
            \left|g(U,U)\right|=&\left|g(U'+\epsilon,U+\epsilon)\right|\\
            \leq&\left|g(U',U')\right|+\left|g(\varepsilon,\varepsilon)\right|+2\left|\tr\left[U'^\top (\frac{1}{n}X^\top X-\Sigma)\varepsilon\right]\right|\\
            \leq&\frac{5}{4}\left|g(U',U')\right|+5\left|g(\varepsilon,\varepsilon)\right|\\
            &+\left|\tr\left[(\frac{1}{2}U'+2\varepsilon)^\top (\frac{1}{n}X^\top X-\Sigma)(\frac{1}{2}U'+2\varepsilon)\right]\right|\\
            \leq &\frac{5}{4}\sup_{\Delta\in \cN_S}\left|g(\Delta,\Delta)\right|+\frac{27}{50}\sup_{\Delta\in \cK(\cG_S)}\left|g(\Delta,\Delta)\right|,
        \end{split}
    \end{equation}
    which implies 
    \begin{equation}
    \begin{split}
        (1-\frac{27}{50})\sup_{\Delta\in \cK(\cG_S)}\left|g(\Delta,\Delta)\right|\leq \frac{5}{4}\sup_{\Delta\in \cN_S}\left|g(\Delta,\Delta)\right|\\
        \sup_{\Delta\in \cK(\cG_S)}\left|g(\Delta,\Delta)\right|\leq 3\sup_{\Delta\in \cN_S}\left|g(\Delta,\Delta)\right|.
    \end{split}    
    \end{equation} 
    Therefore, \eqref{eq:rec2} can be implied by  
    \begin{equation}\label{eq:rec_reduced}
        \sup_{S\subset[J],|S|=2\lfloor\rho\rfloor}\sup_{\Delta\in \cN_S}\left|g(\Delta,\Delta)\right|\leq \frac{\lambda_{\min}(\Sigma)}{18(1+2\max_j\omega_j^{-1})}.
    \end{equation}
    The following lemma bounds $\left|g(\Delta,\Delta)\right|$ with high probability.
    \begin{lemma}\label{lem:rec_dev}
    For any $\Delta\in \bbR^{p\times K}$ such that $\|\Delta\|_2=1$,
    \begin{equation}
            \bbP\left(\left|g(\Delta,\Delta)\right|\geq t\right)\leq 2K\exp\left\{-cn\min\left\{\frac{t^2}{\sigma^4},\frac{t}{\sigma^2}\right\}\right\}.
        \end{equation}
    for any $t>0$.
    \end{lemma}
    Let $t=\frac{\lambda_{\min}(\Sigma)}{18(1+2\max_j\omega_j^{-1})}$ in Lemma \ref{lem:rec_dev} and take a union bound over $\Delta\in \cN_S$, $S\subset[J], |S|=2\lfloor\rho\rfloor$, one can see that \eqref{eq:rec_reduced} holds with probability at least
    \begin{equation}
    \begin{split}
        &1-2K(21)^{2\rho m}\binom{J}{2\lfloor \rho\rfloor}\exp\{-cn\}\\
        \geq &1-2K\exp\{\rho[2\log(21)m+3\log(eJ/2\lfloor \rho\rfloor)]-cn\}\\
        \geq&1-2K\exp\{\rho[2\log(21)m+3\log J]-cn\}\\
        \geq&1-2K\exp\{-\frac{cn}{2}\}
        \end{split}
    \end{equation}
    where the third line holds if $\rho=\frac{\alpha n}{\tau_1(\log J+m)}=c\frac{n}{\log J+m}\geq 2$, and the last line holds since $\rho(\log J+m)= \frac{\alpha n}{\tau_1}\leq Cn$.
    \item Concentrating $\mathrm{\RNum{2}}(\Delta)$ around $0$.\\
    First recall that 
    \begin{equation}
        \mathrm{\RNum{2}}(\Delta)=\frac{1}{n}\sum_{i=1}^n x_i^\top \Delta\left(\nabla f(\Theta^{*\top} x_i+\Delta^\top x_i)-\nabla f(\Theta^{*\top}x_i)\right)^\top \epsilon_i.
    \end{equation}
    Let $\phi(\Delta,\{x_i\}_{i=1}^n,\{\epsilon_i\}_{i=1}^n)=\frac{\mathrm{\RNum{2}}(\Delta)}{\|\Delta\|_{\omega,2,1}}$, then our goal is to provide an upper bound for 
    \begin{equation*}
        \sup_{\Delta\in \bbB_{1,\infty}(R_0)}|\phi(\Delta,\{x_i\}_{i=1}^n,\{\epsilon_i\}_{i=1}^n)|.
    \end{equation*}
    We start by upper bounding 
    \begin{equation*}
        g^{(1)}_t((x_1,\epsilon_1),\dots,(x_n,\epsilon_n)):=\sup_{\frac{t}{2}\leq \|\Delta\|_{\omega,2,1}\leq t}\phi(\Delta,\{x_i\}_{i=1}^n,\{\epsilon_i\}_{i=1}^n)
    \end{equation*} 
    for any $t>0$, and following similar arguments we obtain the upper bound for
    \begin{equation*}
        g^{(2)}_t((x_1,\epsilon_1),\dots,(x_n,\epsilon_n)):\sup_{\frac{t}{2}\leq \|\Delta\|_{\omega,2,1}\leq t}-\phi(\Delta,\{x_i\}_{i=1}^n,\{\epsilon_i\}_{i=1}^n).
    \end{equation*}
    Then we apply a peeling argument to extend the bound to $\bbB_{1,\infty}(R_0)$.
    \begin{lemma}[Symmetrization theorem]\label{lem:symmetrization}
    Let $U_1,\dots, U_n$ be independent random variables with values in $\mathcal{U}$ and $\{\varepsilon_i\}$ be an i.i.d. sequence of Rademacher variables, which take values $\pm 1$ each with probability $\frac{1}{2}$. Let $\Gamma$ be a class of real-valued functions on $\mathcal{U}$, then
    \begin{equation*}
        \bbE\left(\sup_{\gamma\in \Gamma}\sum_{i=1}^n(\gamma(U_i)-\bbE(\gamma(U_i)))\right)\leq2\bbE\left(\sup_{\gamma\in \Gamma}\sum_{i=1}^n\varepsilon_i\gamma(U_i)\right).
    \end{equation*}
    \end{lemma}
    The symmetrization theorem commonly seen in literature \citep{vaart1997weak, wainwright2019high} considers the absolute value $\left|\gamma(U_i)-\bbE(\gamma(U_i))\right|$ instead of $\gamma(U_i)-\bbE(\gamma(U_i))$. For completeness, we will also provide a proof for Lemma~\ref{lem:symmetrization}, although this proof is basically the same as the one with absolute values, and probably has already been shown in past literatures.
    
    We apply Lemma \ref{lem:symmetrization} by letting $U_i=(\epsilon_i,x_i)$, \begin{equation*}
        \gamma(x_i,\epsilon_i)=\frac{1}{n\|\Delta\|_{\omega,2,1}}x_i^\top \Delta(\nabla f((\Theta^*+\Delta)^\top x_i)-\nabla f(\Theta^{*\top}x_i))^\top \epsilon_i.
    \end{equation*} Then we have 
    \begin{equation}
    \begin{split}
        &\bbE g^{(1)}_t((x_1,\epsilon_1),\dots,(x_n,\epsilon_n))\\
        \leq &\frac{2}{n}\bbE\sup_{\frac{t}{2}\leq \|\Delta\|_{\omega,2,1}\leq t}\frac{\sum_{i=1}^n h_i(\Delta^\top x_i)\varepsilon_i}{\|\Delta\|_{\omega,2,1}}\\
        \leq&\frac{4}{nt}\bbE\sup_{\frac{t}{2}\leq \|\Delta\|_{\omega,2,1}\leq t}\sum_{i=1}^n h_i(\Delta^\top x_i)\varepsilon_i,
    \end{split}
    \end{equation}
    where the multivariate function $h_i:\bbR^K\rightarrow \bbR$ is defined as 
    \begin{equation}
    h_i(u)=u^\top\left(\nabla f(\Theta^{*\top} x_i+u)-\nabla f(\Theta^{*\top}x_i)\right)^\top \epsilon_i.
    \end{equation}
    Meanwhile, apply Lemma~\ref{lem:symmetrization} on $-\gamma(x_i,\epsilon_i)$ also leads to 
    \begin{equation*}
        \bbE g^{(2)}_t((x_1,\epsilon_1),\dots,(x_n,\epsilon_n))\leq \frac{4}{nt}\bbE\sup_{\frac{t}{2}\leq\|\Delta\|_{\omega,2,1}\leq t}\sum_{i=1}^n h_i(\Delta^\top x_i)\varepsilon_i.
    \end{equation*}
    The following lemma shows that $h_i(u)$ is $L$-Liptchitz, where $L=K+\sqrt{K}R^*C_x+2$.
    \begin{lemma}\label{lem:liptchitz-h}
    For any $u\in \bbR^K$, $1\leq i\leq n$, $\|\nabla h_i(u)\|_2\leq L=K+\sqrt{K}R^*C_x+2$. 
    \end{lemma}

    The following lemma proved in \cite{maurer2016vector} presents a contraction inequality for Rademacher average when the contraction function has vector-valued domain:
    \begin{lemma}[\cite{maurer2016vector}]\label{lem:contraction_vector_domain}
   For any countable set $\mathcal{S}$ and functions $\psi_i:\mathcal{S}\rightarrow \bbR$, $\phi_i: \mathcal{S}\rightarrow \ell_2$, $1\leq i\leq n$ satisfying
   \begin{equation*}
       \forall s,s'\in \mathcal{S}, \psi_i(s)-\psi_i(s')\leq \|\phi_i(s)-\phi_i(s')\|,
   \end{equation*}
   we have
   \begin{equation*}
       \bbE\sup_{s\in\mathcal{S}}\sum_{i=1}^n\varepsilon_i\psi_i(s)\leq \sqrt{2}\bbE\sup_{s\in\mathcal{S}}\sum_{i,k}\widetilde{\epsilon}_{ik}\phi_i(s)_k,
   \end{equation*}
   where $\ell_2$ is the Hilbert space of square summable sequences of real numbers, $\varepsilon_i$, $\widetilde{\varepsilon}_{ik}$ are independent Rademacher random variables for $1\leq i\leq n$ and $1\leq k< \infty$, and $\phi_i(s)_k$ is the $k$-th coordinate of $\phi_i(s)_k$.
    \end{lemma}
    We apply Lemma \ref{lem:contraction_vector_domain} by letting $\mathcal{S}=\{\Delta\in \mathbb{Q}^{p\times K}:\frac{t}{2}\leq \|\Delta\|_{\omega,2,1}\leq t\}$, $\psi_i(s)=h_i(s^\top x_i)$, $(\phi_i(s))_{1:K}=Ls^\top x_i$ and $(\phi_i(s))_{(K+1):\infty}=0$. Then one can show that
    \begin{equation*}
        \begin{split}
            \bbE \sup_{\Delta\in\mathcal{S}}\sum_{i=1}^n h_i(\Delta^\top x_i)\varepsilon_i
            \leq&\sqrt{2}L\bbE\sup_{\Delta\in \mathcal{S}}\sum_{i=1}^n\sum_{k=1}^K\Delta_{:,k}^\top x_i\widetilde{\varepsilon}_{ik}\\
            \leq&\sqrt{2}Lt\bbE\left\|\sum_{i=1}^n x_i\widetilde{\varepsilon}_i^\top\right\|_{\omega^{-1},2,\infty}.
        \end{split}
    \end{equation*}
    Note that $\mathcal{S}$ is dense in $\{\Delta\in \bbR^{p\times K}:\frac{t}{2}\leq \|\Delta\|_{\omega,2,1}\leq t\}$ and $\sum_{i=1}^nh_i(\Delta^\top x_i)\varepsilon_i$ is a continuous function of $\Delta$, thus we have for $j=1,2$,
    \begin{equation}\label{eq:g_t_expt}
        \begin{split}
            &\bbE g_t^{(j)}((x_1,\epsilon_1),\dots,(x_n,\epsilon_n))\\
            =&\frac{4}{nt}\bbE\sup_{\frac{t}{2}\leq \|\Delta\|_{\omega,2,1}\leq t}\sum_{i=1}^n h_i(\Delta^\top x_i)\varepsilon_i\\
            =&\frac{4}{nt}\bbE\sup_{\Delta\in \mathcal{S}}\sum_{i=1}^n h_i(\Delta^\top x_i)\varepsilon_i\\
            \leq&\frac{4\sqrt{2}L}{n}\bbE\left\|\sum_{i=1}^n x_i\widetilde{\varepsilon}_i^\top\right\|_{\omega^{-1},2,\infty}.
        \end{split}
    \end{equation}
    The following lemma provides an upper bound for $\bbE\left\|\sum_{i=1}^nx_i\widetilde{\varepsilon}_i^\top\right\|_{\omega^{-1},2,\infty}$. 
    \begin{lemma}\label{lem:mult_dev_expt}
     \begin{equation*}
        \bbE\left\|\sum_{i=1}^nx_i\widetilde{\varepsilon}_i^\top\right\|_{\omega^{-1},2,\infty}\leq C\lambda_{\max}^{\frac{1}{2}}(\Sigma)\sqrt{(m+\log J)n}.
     \end{equation*}
    \end{lemma}
    Getting back to \eqref{eq:g_t_expt}, we know that for $j=1,2$,
    \begin{equation}
        \bbE g^{(j)}_t((x_1,\epsilon_1),\dots,(x_n,\epsilon_n))\leq CL\lambda^{\frac{1}{2}}_{\max}(\Sigma)\sqrt{\frac{m+\log J}{n}}.
    \end{equation}
    Note that
    \begin{equation}
    \begin{split}
         &\frac{\left|x_i^\top \Delta(\nabla f(\Theta^{*\top}x_i+\Delta^\top x_i)-\nabla f(\Theta^{*\top}x_i))\epsilon_i\right|}{n\|\Delta\|_{\omega,2,1}}\\
        \leq&\frac{\left|x_i^\top \Delta(\nabla A(\Theta^{*\top}x_i+\Delta^\top x_i)-\nabla A(\Theta^{*\top}x_i))\right|}{n\|\Delta\|_{\omega,2,1}}\\
        \leq&\frac{1}{n}\|x_i(\nabla A(\Theta^{*\top}x_i+\Delta^\top x_i)-\nabla A(\Theta^{*\top}x_i))^\top\|_{\omega^{-1},2,\infty}\\
        \leq&\frac{C}{n}\max_j\|x_{i,\cR_j}\|_2\|\nabla A(\Theta^{*\top}x_i+\Delta^\top x_i)_{\cC_j}-\nabla A(\Theta^{*\top}x_i)_{\cC_j}\|_2\\
        \leq &\frac{CC_x\max_j\sqrt{r_j}}{n},
    \end{split}
    \end{equation}
    which implies that for any $\{(x_i,\epsilon_i)\}_{i=1}^n$, $(x_i',\epsilon_i')$, $j=1,2$,
    \begin{equation*}
    \begin{split}
        &\bigg|g_t^{(j)}((x_1,\epsilon_1),\dots,(x_i,\epsilon_i),\dots,(x_n,\epsilon_n))\\
        &-g^{(j)}_t((x_1,\epsilon_1),\dots,(x_i',\epsilon_i'),\dots,(x_n,\epsilon_n))\bigg|\\
        \leq&\frac{CC_x\max_j\sqrt{r_j}}{n}.
    \end{split}
    \end{equation*}
    Thus we can apply the bounded difference inequality \citep{mcdiarmid1989method} and obtain the following result:
    \begin{equation}\label{eq:rsc_term2_slice}
        \sup_{\frac{t}{2}\leq \|\Delta\|_{\omega,2,1}\leq t}\left|\phi(\Delta,\{x_i\}_{i=1}^n,\{\epsilon_i\}_{i=1}^n)\right|\leq CL\lambda^{\frac{1}{2}}_{\max}(\Sigma)\sqrt{\frac{m+\log J}{n}},
    \end{equation}
    with probability at least $1-2\exp\{-cL^2C_x^{-2}\lambda_{\max}(\Sigma)\min_jr_j^{-1}(m+\log J)\}$.
    Now we apply a peeling argument to extend probabilistic bound \eqref{eq:rsc_term2_slice} to all $\Delta\in \bbB_{1,\infty}(R_0)$. Since $\|\Delta\|_{\omega,2,1}\leq \max_j\omega_jK\|\Delta\|_{1,\infty}$, $\bbB_{1,\infty}(R_0)\subset \bbB_{\omega,2,1}(\max_j\omega_jKR_0)$. Let $C_n=\frac{C}{C_x^2}L\lambda^{\frac{1}{2}}_{\max}(\Sigma)\min_j\omega_j^2\sqrt{\frac{\log J+m}{m^2n}}$, $N=\log_2(\frac{\max_j\omega_jKR_0}{C_n})$, then one can show that
    \begin{equation}
    \begin{split}
        \bbB_{1,\infty}(R_0)\subset&\cup_{k=1}^{N}\{\Delta:2^{k-1}C_n\leq \|\Delta\|_{\omega,2,1}\leq 2^kC_n\}\\
        &\cup \{\Delta:0\leq \|\Delta\|_{\omega,2,1}\leq C_n\}.
    \end{split}
    \end{equation}
    First we consider how to establish the bound uniformly for $\{\Delta:0\leq \|\Delta\|_{\omega,2,1}\leq C_n\}$. For any $r>0, d\in \bbR^{p\times K}$ satisfying $\|d\|_{\omega,2,1}=1$, let $\widetilde{\phi}(r,d,\{x_i\}_{i=1}^n,\{\epsilon_i\}_{i=1}^n)=\phi(rd,\{x_i\}_{i=1}^n,\{\epsilon_i\}_{i=1}^n)$. First note that
    \begin{equation}
        \begin{split}
            &\left|\frac{\partial \widetilde{\phi}(r,d,\{x_i\}_{i=1}^n,\{\epsilon_i\}_{i=1}^n)}{\partial r}\right|\\
            =&\left|\frac{1}{n}\sum_{i=1}^nx_i^\top d\nabla^2A(\Theta^{*\top}x_i+rd^\top x_i)d^\top x_i1_K^\top \epsilon_i\right|\\
            \leq&\frac{1}{n}\sum_{i=1}^n\|x_i^\top d\|_2^2\\
            \leq&C_x^2\|d\|_1^2\\
            \leq&C_x^2\max_j\omega_j^{-2}m,
        \end{split}
    \end{equation}
    which implies that for any $\|\Delta\|_{\omega,2,1}\leq C_n$,
    \begin{equation*}
    \begin{split}
        &\left|\phi(\Delta,\{x_i\}_{i=1}^n,\{\epsilon_i\}_{i=1}^n)\right|\\
        =&\left|\widetilde{\phi}(\|\Delta\|_{\omega,2,1},\frac{\Delta}{\|\Delta\|_{\omega,2,1}},\{x_i\}_{i=1}^n,\{\epsilon_i\}_{i=1}^n)\right|\\
        \leq&\left|\widetilde{\phi}(C_n,\frac{\Delta}{\|\Delta\|_{\omega,2,1}},\{x_i\}_{i=1}^n,\{\epsilon_i\}_{i=1}^n)\right|+C_x^2\max_j\omega_j^{-2}mC_n\\
        \leq&\sup_{\|\Delta\|_{\omega,2,1}=C_n}\left|\phi(\Delta,\{x_i\}_{i=1}^n,\{\epsilon_i\}_{i=1}^n)\right|+CL\lambda^{\frac{1}{2}}_{\max}(\Sigma)\sqrt{\frac{\log J+m}{n}}.
    \end{split}
    \end{equation*}
    Therefore, 
    \begin{equation}
        \begin{split}
            &\sup_{\Delta\in \bbB_{1,\infty}(R_0)}\left|\phi(\Delta,\{x_i\}_{i=1}^n,\{\epsilon_i\}_{i=1}^n)\right|\\
            \leq &\max_{0\leq k\leq N}\sup_{2^{k-1}C_n\leq \|\Delta\|_{\omega,2,1}\leq 2^kC_n}\left|\phi(\Delta,\{x_i\}_{i=1}^n,\{\epsilon_i\}_{i=1}^n)\right|\\
            &+CL\lambda^{\frac{1}{2}}_{\max}(\Sigma)\sqrt{\frac{\log J+m}{n}}.
        \end{split}
    \end{equation}
    Now note that
    \begin{equation*}
    \begin{split}
        N=&\log_2\frac{C_x^2R_0Km\max_j\omega_j}{L\lambda_{\max}^{\frac{1}{2}}(\Sigma)\min_j\omega_j^2}\sqrt{\frac{n}{\log J+m}}\\
        \leq&\frac{1}{2}\log_2 n+\log_2 m+\log_2K+\log_2\frac{C_x\max_j\omega_j}{\lambda_{\max}^{\frac{1}{2}}(\Sigma)}+C\\
        \leq&C\log n,
    \end{split}
    \end{equation*}
    where we applied the fact that $$R_0\leq\frac{\min_j\frac{n_j}{\pi_jn_u}}{4C_x(1+\max_j\frac{n_j}{\pi_jn_u})^2}\leq\frac{1}{16C_x},$$
    $L\geq 1$, $\min_j\omega_j\geq c$ on the second line, and applied $K,m\leq Cn$, $$C_x\max_j\omega_j\lambda_{\max}^{-\frac{1}{2}}(\Sigma)\leq n^C$$ on the third line. Take a union bound over $0\leq k\leq N$, which implies that
    \begin{equation}
        \sup_{\Delta\in \bbB_{1,\infty}(R_0)}\left|\phi(\Delta,\{x_i\}_{i=1}^n,\{\epsilon_i\}_{i=1}^n)\right|\leq CL\lambda^{\frac{1}{2}}_{\max}(\Sigma)\sqrt{\frac{\log J+m}{n}},
    \end{equation}
    with probability at least 
    \begin{equation}
        \begin{split}
            &1-C\log n\exp\{-cL^2C_x^{-2}\lambda_{\max}(\Sigma)\min_jr_j^{-1}(\log J+m)\}\\
            \geq&1-\exp\{C\log\log n-cL^2C_x^{-2}\lambda_{\max}(\Sigma)\min_jr_j^{-1}(\log J+m)\}\\
            \geq&1-\exp\{-cK^2\min_jr_j^{-1}(\log J+m)\},
        \end{split}
    \end{equation}
    where the last line is due to that 
    \begin{equation*}
    \begin{split}
        &L^2C_x^{-2}\lambda_{\max}(\Sigma)\min_jr_j^{-1}(\log J+m)\\
        \geq &cK^2\min_jr_j^{-1}(\log J+m)\\
        \geq&C\log\log n.
    \end{split}
    \end{equation*}
\end{enumerate}
\end{proof}
\begin{proof}[Proof of Lemma \ref{lem:MN_region_results}]
Recall that Assumption \ref{assump:mult_region_bound} requires $R_0\leq \frac{\min_j\frac{n_j}{\pi_jn_u}e^{-C_xR^*}}{4C_x(1+\max_j\frac{n_j}{\pi_jn_u})^2(1+1.1Ke^{C_xR^*})^3}$.
    We first prove that $h^{\MN}(R_0)>0$ under this constraint for $R_0$. Noting that $\frac{\min_j\frac{n_j}{\pi_jn_u}}{(1+\max_j\frac{n_j}{\pi_jn_u})^2}<1$, $e^{-C_xR^*}\leq 1$, $1+1.1Ke^{C_xR^*}\geq 3.2$, we know that $e^{C_xR_0}<e^{\frac{1}{4\times 3.2^3}}<1.01$. Therefore, 
\begin{equation*}
\begin{split}
    h^{\MN}(R_0) &= \frac{e^{-C_x(R_0+R^*)}\min_j\frac{n_j}{\pi_j n_u}}{(1+\max_j\frac{n_j}{\pi_j n_u})^2\left(1+Ke^{C_x(R_0+R^*)}\right)^3}-4C_xR_0\\
    &>\frac{e^{-C_xR^*}\min_j\frac{n_j}{\pi_j n_u}}{1.01(1+\max_j\frac{n_j}{\pi_j n_u})^2\left(1+1.01Ke^{C_xR^*}\right)^3}-4C_xR_0\\
    &\geq \frac{\min_j\frac{n_j}{\pi_j n_u}e^{-C_xR^*}}{(1+\max_j\frac{n_j}{\pi_j n_u})^2\left(1+1.1Ke^{C_xR^*}\right)^3}-4C_xR_0\\
    &\geq 0,
\end{split}
\end{equation*}
where the last line is due to Assumption \ref{assump:mult_region_bound}, the penultimate line is due to the fact that
\begin{equation*}
    \frac{1+1.1Ke^{C_xR^*}}{1+1.01Ke^{C_xR^*}}\geq 1+ \frac{0.09Ke^{C_xR^*}}{1+1.01Ke^{C_xR^*}}\geq 1+ \frac{0.09Ke^{C_xR^*}}{1.51Ke^{C_xR^*}}=1.059> (1.01)^{1/3},
\end{equation*}
We have also applied the fact that $K\geq 2$ in the inequality above. 

While for the second claim in Lemma \ref{lem:MN_region_results}, we note that our previous argument suggests $C_xR_0 < c$ for a constant $c>0$. In addition, by the definition of $h^{\MN}(R_0)$, it is straightforward to see that $0<h^{\MN}(R_0)<1$. Furthermore, Assumptions \ref{assump:subgauss_row} and \ref{assump:scaling} suggest that $C\lambda_{\min}^{\frac{1}{2}}(\Sigma)\leq C_x$ and $\max_j\omega_j\leq C\sqrt{\frac{n}{\log J+m}}$, which then implies 
    $$R_0\leq \frac{c\lambda_{\max}^{\frac{1}{2}}(\Sigma)(K+C_xR^*\sqrt{K})}{h^{\MN}(R_0)\lambda_{\min}(\Sigma)\max_j\omega_jK}\sqrt{\frac{n}{\log J+m}}.$$ Therefore,
    \begin{equation}
        \|\widehat{\Delta}\|_{\omega,2,1}\leq K\max_j\omega_j R_0\leq\frac{\tau_2}{\tau_1}\sqrt{\frac{n}{\log J+m}},
    \end{equation}
    and thus $\tau_1\frac{\log J+m}{n}\|\widehat{\Delta}\|_{\omega,2,1}^2\leq \tau_2\sqrt{\frac{\log J+m}{n}}\|\widehat{\Delta}\|_{\omega,2,1}$. 
\end{proof}
\section{Proof of Theorem~\ref{thm:ordinal_upp}}\label{sec:proof_theorem_ordinal}
Lemma~\ref{lem:ordinal_rsc} shown in Section~\ref{sec:main_proofs} and the following lemma are the key building blocks for the proof of Theorem~\ref{thm:ordinal_upp}.
\begin{lemma}[Deviation Bound under the Ordinal-PU Model]\label{lem:ordinal_dev_bnd}
If the data set $\{(x_i,z_i)\}_{i=1}^n$ is generated from the ordinal-PU model under the case-control setting, then
\begin{equation}
\begin{split}
    \|(\nabla\cL_n^{\ON}(\Theta^*))_{1:p}\|_{\omega^{-1},2,\infty}\leq& Ch^{\ON}(0,0)\lambda_{\max}^{\frac{1}{2}}(\Sigma)\sqrt{\frac{\log J+m}{n}},\\
    \|(\nabla\cL_n^{\ON}(\theta^*))_{(p+1):(p+K)}\|_{\infty}\leq&Ch^{\ON}(0,0)\sqrt{\frac{\log J+m+\log K}{n}},
\end{split}
\end{equation}
with probability at least $1-3\exp\{-(\log J+m)\}$. Here $h^{\ON}(\cdot,\cdot)$ is as defined in~\eqref{eq:ordinal_h_def}.
\end{lemma} 
While for Lemma \ref{lem:ordinal_rsc}, 
the specific forms of constants $\alpha$ and $L$ are as follows:
 $\alpha=\frac{1}{2}\gamma\min\{\lambda_{\min}(\Sigma),1\}$,  where $$\gamma=\min_j\frac{n_j}{\pi_jn_u}\frac{e^{2(C_x+1)(R^*+R_0)}}{32h^{\ON}(R_0,r_0)(1+e^{(C_x+1)(R^*+R_0)})^4K},$$
 and $L=(\frac{4}{3}(C_x+1)R_0+1)\sqrt{K}h^{\ON}(R_0,r_0)+\frac{(C_x+1)R_0}{4}\sqrt{K}(h^{\ON}(R_0,r_0))^2$.

Now we are ready to prove Theorem~\ref{thm:ordinal_upp} based on Lemma~\ref{lem:ordinal_dev_bnd} and Lemma~\ref{lem:ordinal_rsc}.
\begin{proof}[Proof of Theorem \ref{thm:ordinal_upp}]
    Similarly from the proof of Theorem~\ref{thm:mult_upp}, one can show that 
    \begin{equation}
        \begin{split}
            \langle\nabla\cL_n^{\ON}(\widehat{\theta})+\lambda\nabla \|\widehat{\theta}_{1:p}\|_{\omega,2,1},\theta^*-\widehat{\theta}\rangle\geq& 0,\\
            \langle\nabla\cL_n^{\ON}(\widehat{\theta})-\nabla\cL_n^{\ON}(\theta^*),\widehat{\theta}-\theta^*\rangle\leq&\langle -\nabla \cL_n^{\ON}(\theta^*)-\lambda\nabla \|\widehat{\theta}_{1:p}\|_{\omega,2,1},\widehat{\theta}-\theta^*\rangle.
        \end{split}
    \end{equation}
    where $\nabla \|\theta_{1:p}\|_{\omega,2,1}$ is any sub-gradient of $\|\theta_{1:p}\|_{\omega,2,1}$ as a function of $\theta$. 
    Let $\widehat{\Delta}=\widehat{\theta}-\theta^*$, and $\cG_S=\cup_{j\in S}\cG_j$ where $S$ is defined earlier as $S=\{j:\theta^*_{\cG_j}\neq 0\}$. By Lemma \ref{lem:ordinal_dev_bnd}, Lemma \ref{lem:ordinal_rsc},
    \begin{equation}
        \begin{split}
            &\alpha \|\widehat{\Delta}\|_2^2-C\alpha\frac{\log J+m}{n}\|\widehat{\Delta}_{1:p}\|_{\omega,2,1}^2\\
            &-CL\left(\lambda_{\max}^{\frac{1}{2}}(\Sigma)\sqrt{\frac{\log J+m}{n}}\|\widehat{\Delta}_{1:p}\|_{\omega,2,1}+\sqrt{\frac{K(\log J+\log(2K))}{n}}\|\widehat{\Delta}\|_2\right)\\
            \leq& Ch^{\ON}(0,0)\lambda_{\max}^{\frac{1}{2}}(\Sigma)\sqrt{\frac{\log J+m}{n}}\|\widehat{\Delta}_{1:p}\|_{\omega,2,1}\\
            &+Ch^{\ON}(0,0)\sqrt{\frac{K(\log J+m+\log K)}{n}}\|\widehat{\Delta}\|_2+\lambda\langle \nabla \|\widehat{\theta}_{1:p}\|_{\omega,2,1},\theta^*-\widehat{\theta}\rangle.
        \end{split}
    \end{equation}
    By the definition of $L$, $L\geq h^{\ON}(R_0,r_0)\sqrt{K}\geq h^{\ON}(0,0)$ and thus the inequality above can be transformed to 
    \begin{equation*}
    \begin{split}
        \alpha \|\widehat{\Delta}\|_2^2\leq& C\alpha\frac{\log J+m}{n}\|\widehat{\Delta}_{1:p}\|_{\omega,2,1}^2+CL\lambda_{\max}^{\frac{1}{2}}(\Sigma)\sqrt{\frac{\log J+m}{n}}\|\widehat{\Delta}_{1:p}\|_{\omega,2,1}\\
            &+CL\sqrt{\frac{K(\log J+m+\log K)}{n}}\|\widehat{\Delta}\|_2+\lambda\langle \nabla \|\widehat{\theta}_{1:p}\|_{\omega,2,1},\theta^*-\widehat{\theta}\rangle.
    \end{split}
    \end{equation*}
    Meanwhile, by the definition of $\alpha, L$, Assumption~\ref{assump:scaling} and the fact that $h^{\ON}(R_0,r_0)\geq 1$, one can show that $\alpha\leq C$, $L\geq 1$, $\max_j\omega_j\leq \sqrt{\frac{n}{\log J +m}}$. In addition, by Assumptions~\ref{assump:subgauss_row} and \ref{assump:scaling}, one has $\lambda_{\max}(\Sigma)\leq C\sigma^2\leq C\min\{\lambda_{\min}(\Sigma),1\}$, which implies $0<C_1\leq \lambda_{\min}(\Sigma)\leq \lambda_{\max}(\Sigma)\leq C_2$. Assumption \ref{assump:ordinal_region_bound} also implies $R_0\leq C$ for a constant $C>0$. Hence we know that
    \begin{equation*}
    \begin{split}
        \|\widehat{\Delta}_{1:p}\|_{\omega,2,1}\leq &\max_j\omega_j\|\widehat{\Delta}_{\cG_j}\|_1\\
        \leq &\max_j\omega_j R_0\\
        \leq &\frac{C}{\alpha}L\lambda_{\max}^{\frac{1}{2}}(\Sigma)\sqrt{\frac{n}{\log J+m}},
    \end{split}
    \end{equation*}
    which further implies $\alpha\frac{\log J+m}{n}\|\widehat{\Delta}_{1:p}\|_{\omega,2,1}^2\leq CL\lambda_{\max}^{\frac{1}{2}}(\Sigma)\sqrt{\frac{\log J+m}{n}}\|\widehat{\Delta}_{1:p}\|_{\omega,2,1}$. Moreover, \begin{equation*}
    \begin{split}
        \langle \nabla \|\widehat{\theta}_{1:p}\|_{\omega,2,1},\theta^*-\widehat{\theta}\rangle\leq& \|\theta^*_{1:p}\|_{\omega,2,1}-\|\widehat{\theta}_{1:p}\|_{\omega,2,1}\\
        \leq&\|\widehat{\Delta}_{\cG_S}\|_{\omega_S,2,1}-\|\widehat{\Delta}_{\cG_S^c}\|_{\omega_{S^c},2,1}.
    \end{split}
    \end{equation*}
    Therefore, one can show that
    \begin{equation*}
        \begin{split}
            \alpha\|\widehat{\Delta}\|_2^2\leq& \frac{3\lambda}{2}\|\widehat{\Delta}_{\cG_S}\|_{\omega_S,2,1}+CL\sqrt{\frac{K(\log J+m+\log K)}{n}}\|\widehat{\Delta}\|_2\\
            \leq&\left(\frac{3\lambda\|\omega_S\|_2}{2}+CL\sqrt{\frac{K(\log J+m+\log K)}{n}}\right)\|\widehat{\Delta}\|_2,\\
           \|\widehat{\Delta}\|_2 \leq&\frac{3\|\omega_S\|_2\lambda}{\gamma\min\{\lambda_{\min}(\Sigma),1\}}+\frac{CL}{\gamma\min\{\lambda_{\min}(\Sigma),1\}}\sqrt{\frac{K(\log J+m+\log K)}{n}},
        \end{split}
    \end{equation*}
    with probability at least $1-\exp\{-\frac{cK\log J}{C_x^2m+1}\}-3\exp\{-(\log J+m)\}$. Let $\gamma_0=\max_j\frac{\pi_jn_u}{n_j}K\gamma$ and $L_0=\frac{L}{\sqrt{K}}$, we obtain the final results.
\end{proof}
\subsection{Proofs of Lemma~\ref{lem:ordinal_dev_bnd} and Lemma~\ref{lem:ordinal_rsc}}\label{sec:proof_ordinal_lem}
Recall the definitions of functions $f^{\ON}:\bbR^K\rightarrow \bbR^K$ and $u:\bbR^p\times \bbR^{p+K}\rightarrow \bbR^K$, first defined in Section~\ref{sec:ordinal_rsc_proof}:
\begin{equation*}
   (f^{\ON}(u))_j=\begin{cases}\log\frac{n_j}{\pi_jn_u}+\log\left[(1+e^{\sum_{l=1}^{j+1}u_{l}})^{-1}-(1+e^{\sum_{l=1}^ju_l})^{-1}\right],&1\leq j<K,\\
    \log\frac{n_K}{\pi_Kn_u}+\log\left[1-(1+e^{\sum_{l=1}^Ku_l})^{-1}\right],&j=K,
\end{cases}
\end{equation*}
$$u(x_i,\theta)=(x_i^\top\theta_{1:p}-\theta_{p+1},-\theta_{p+2},\dots,-\theta_{p+K})^\top.$$
The log-likelihood loss $\cL_n^{\ON}(\theta)$ can be written as follows:
\begin{equation*}
    \cL_n^{\ON}(\theta)=\frac{1}{n}\sum_{i=1}^n\log\left(1+\sum_{k=1}^Ke^{f^{\ON}(u(x_i,\theta))}\right)-\sum_{k=1}^K\ind{z_i=k}f^{\ON}(u(x_i,\theta)).
\end{equation*}

Also recall the functions $p:\bbR^K\rightarrow \bbR^K$ and $\alpha:\bbR^K\rightarrow \bbR^{K+1}$ first defined in Section~\ref{sec:ordinal_rsc_proof_lemmas}:
\begin{equation*}
    \alpha_j(u)=\begin{cases}e^{\sum_{l=1}^ju_{l}}\left(1+e^{\sum_{l=1}^{j}u_{l}}\right)^{-2},&j\leq K,\\
    0,&j=K+1,
    \end{cases}
\end{equation*}
and
\begin{equation*}
    p_j(u)=\begin{cases}
    (1+e^{\sum_{l=1}^{j+1}u_{l}})^{-1}-(1+e^{\sum_{l=1}^ju_{l}})^{-1},&1\leq j<K,\\
    1-(1+e^{\sum_{l=1}^ju_{l}})^{-1},&j=K.
    \end{cases}
\end{equation*}
These two functions determine the derivatives of function $f^{\ON}$ and thus influences both our deviation bound and lower bound on the curvature. Lemma~\ref{lem:ordinal_p_alpha_f_bnd} shows the ranges of $p(\cdot),\alpha(\cdot)$ and derivatives of $f^{\ON}(\cdot)$, and will be useful in our proofs of Lemma~\ref{lem:ordinal_dev_bnd} and Lemma~\ref{lem:ordinal_rsc}. 

\begin{proof}[Proof of Lemma \ref{lem:ordinal_dev_bnd}]
First note that,
\begin{equation}
\begin{split}
    \nabla \cL_n^{\ON}(\theta^*)=&-\frac{1}{n}\sum_{i=1}^n \sum_{k'=1}^K (\delta_{ik'}-\bbE(\delta_{ik'}|x_i))\frac{\partial f_{k'}(u(x_i,\theta^*))}{\partial \theta^*}\\
    =&-\frac{1}{n}\sum_{i=1}^n \begin{pmatrix}
    x_i & 0 & \dots & 0\\
    -1 & 0&\dots&0\\
    0&-1&\dots&0\\
    0&0&\ddots&0\\
    0&0&\dots&-1
    \end{pmatrix}(\nabla f^{\ON}(u_i^*))^\top \epsilon_i,
\end{split}
\end{equation}
where $\delta_{ik}=\ind{z_i=k}$, $\epsilon_{ik'}=\delta_{ik'}-\bbE(\delta_{ik'}|x_i)$, $\epsilon_i=(\epsilon_{i1},\dots,\epsilon_{iK})^\top$, $u_i^*=u(x_i,\theta^*)\in\bbR^K$.
As shown in Section~\ref{sec:ordinal_rsc_proof_lemmas}, for any $1\leq j,k\leq K$,
\begin{equation*}
    (\nabla f^{\ON}(u_i^*))_{jk}=p_j(u_i^*)^{-1}\left[\alpha_j(u_i^*)\ind{k\leq j}-\alpha_{j+1}(u_i^*)\ind{k\leq j+1}\right],
\end{equation*}
where function $p(\cdot)$ and $\alpha(\cdot)$ are defined in \eqref{eq:ordinal_p_def} and \eqref{eq:ordinal_alpha_def}, and $p_j(\cdot),\alpha_j(\cdot)$ refer to the $j$th coordinate of functions $p(\cdot)$ and $\alpha(\cdot)$.
Let $\widetilde{\epsilon}\in \bbR^{K}$, where $\widetilde{\epsilon}_{ik}=(\nabla f^{\ON}(u_i^*)_{:,k})^{\top}\epsilon_i$, $1\leq k\leq K$. Then one can show that
\begin{equation}
\begin{split}
    \left|\widetilde{\epsilon}_{i1}\right|\leq& \left|\sum_{j=1}^Kp_j(u_i^*)^{-1}[
    \alpha_j(u_i^*)-\alpha_{j+1}(u_i^*)]\delta_{ij}\right|\\
    &+\left|\sum_{j=1}^Kp_j(u_i^*)^{-1}[
    \alpha_j(u_i^*)-\alpha_{j+1}(u_i^*)]\bbE(\delta_{ij}|x_i)\right|\\
    \leq &2\max_j\left|p_j(u_i^*)^{-1}[
    \alpha_j(u_i^*)-\alpha_{j+1}(u_i^*)]\right|\\
    \leq&\frac{1}{2}\max_jp_j(u_i^*)^{-1},
\end{split}
\end{equation}
where the last line is due to that $0<\alpha_j(u_i^*)<\frac{1}{4}$ as shown in Lemma~\ref{lem:ordinal_p_alpha_f_bnd}.
While for $2\leq k\leq K$,
\begin{equation}
\begin{split}
    \left|\widetilde{\epsilon}_{ik}\right|\leq& \left|\sum_{j=k}^Kp_{j}(u_i^*)^{-1}(\alpha_j(u_i^*)-\alpha_{j+1}(u_i^*))\delta_{ij}-p_{k-1}(u_i^*)^{-1}\alpha_k(u_i^*)\delta_{i,k-1}\right|\\
    &+\left|\sum_{j=k}^Kp_{j}(u_i^*)^{-1}(\alpha_j(u_i^*)-\alpha_{j+1}(u_i^*))\bbE\delta_{ij}-p_{k-1}(u_i^*)^{-1}\alpha_k(u_i^*)\bbE\delta_{i,k-1}\right|\\
    \leq &\frac{1}{2}\max_jp_j(u_i^*)^{-1}.
\end{split}
\end{equation}
Hence by Lemma~\ref{lem:ordinal_p_alpha_f_bnd}, $\|\widetilde{\epsilon}_i\|_{\infty}\leq \frac{1}{2}h^{\ON}(0,0)$. 
Therefore, 
\begin{equation*}
    \|(\nabla\cL_n^{\ON}(\theta^*))_{(p+1):(p+K)}\|_{\infty}=\max_k\left|\frac{1}{n}\sum_{i=1}^n\widetilde{\epsilon}_{ik}\right|\leq h^{\ON}(0,0)\sqrt{\frac{\log J+m+\log K}{2n}},
\end{equation*} 
with probability at least $1-K\exp\{-(\log J+m+\log K)\}\geq1-\exp\{-(\log J+m)\}$. While for bounding $\|(\nabla\cL_n^{\ON}(\theta^*))_{1:p}\|_{\omega,2,1}$, one can show that for any $1\leq j\leq J$, 
\begin{equation*}
    \left\|\left(\nabla \cL_n^{\ON}(\theta^*)\right)_{\cG_j}\right\|_2=\left\|\frac{1}{n}X_{\cG_j}^\top \widetilde{\epsilon}_{:,1}\right\|_2.
\end{equation*}
Here we use similar arguments from the proof of Lemma \ref{lem:multinomial_dev_bnd} for bounding $\left\|X_{\cG_j}^\top \widetilde{\epsilon}_{:,1}\right\|_2$. First note that 
\begin{equation}
    \begin{split}
        \bbE(\|X^\top_{\cG_j}\widetilde{\epsilon}_{:,1}\|_2|X)\leq&\left(\bbE(\|X^\top_{\cG_j}\widetilde{\epsilon}_{:,1}\|_2^2|X)\right)^{\frac{1}{2}}\\
        =&\left(\sum_{i=1}^n\bbE(\widetilde{\epsilon}^2_{i,1}|X)\left(X_{\cG_j}X^\top_{\cG_j}\right)_{ii}\right)^{\frac{1}{2}}\\
        \leq &\frac{1}{2}h^{\ON}(0,0)\sqrt{\tr(X_{\cG_j}X^\top_{\cG_j})}\\
        \leq&\frac{1}{2}h^{\ON}(0,0)\sqrt{m}\|X_{\cG_j}\|.
    \end{split}
\end{equation}
Let $g(u)=\|X^\top_{\cG_j}u\|_2$ be a function from $\bbR^n$ to $\bbR$, then $g$ is convex and $\|X_{\cG_j}\|$-Lipschitz. Applying Talagrand's contraction inequality (Theorem 5.2.16 in \cite{vershynin2018high}) shows us that
\begin{equation}\label{eq:dev-bnd-term-expt-ordinal}
    \left\|\|X^\top_{\cG_j} \widetilde{\epsilon}_{:,1}\|_2-\bbE(\|X^\top_{\cG_j}\widetilde{\epsilon}_{:,1}\|_2|X)\right\|_{\psi_2}\leq Ch^{\ON}(0,0)\|X_{\cG_j}\|,
\end{equation}
which further implies that
\begin{equation}\label{eq:dev-bnd-term-ordinal}
    \bbP(\|X^\top_{\cG_j}\widetilde{\epsilon}_{:,1}\|_2>2h^{\ON}(0,0)\|X_{\cG_j}\|(\sqrt{m}+C\sqrt{n}t)\leq \exp\{-nt^2\},
\end{equation}
for any $t>0$. Following from similar argument in the proof of Lemma \ref{lem:multinomial_dev_bnd}, we have the following:
\begin{equation}\label{eq:X_norm_bnd-ordinal}
    \|X_{\cG_j}\|\leq C\lambda^{\frac{1}{2}}_{\max}(\Sigma)\sqrt{n},
\end{equation}
holds for $1\leq j\leq J$ with probability at least $1-2J\exp\{-cn\}\geq 1-\exp\{-cn\}$, if $C>0$ is chosen appropriately in \eqref{eq:X_norm_bnd-ordinal}. 
Now let $t=\sqrt{\frac{2\log J+m}{n}}$ in \eqref{eq:dev-bnd-term-ordinal}, then we have
\begin{equation}
\begin{split}
    \|(\nabla\cL_n^{\ON}(\theta))_{1:p}\|_{\omega^{-1},2,\infty}\leq& \max_{j}\omega_j^{-1}\left\|\frac{1}{n}X_{\cG_j}^\top \widetilde{\epsilon}_{:,1}\right\|_2\\
    \leq &Ch^{\ON}(0,0)\lambda_{\max}^{\frac{1}{2}}(\Sigma)\sqrt{\frac{\log J+m}{n}},
\end{split}
\end{equation}
with probability at least $1-\exp\{-(\log J+m)\}-\exp\{-cn\}$.
\end{proof}
\begin{proof}[Proof of Lemma~\ref{lem:ordinal_rsc}]
By the definition of $\cL_n^{\ON}(\theta)$, one can show that
\begin{equation}
\begin{split}
    &\nabla \cL_n^{\ON}(\theta)\\
    =&\frac{1}{n}\sum_{i=1}^n\begin{pmatrix}x_i&0&\cdots&0\\
    -1&0&\cdots&0\\
    0&-1&\cdots&0\\
    0&0&\ddots&0\\
    0&0&\cdots&-1
    \end{pmatrix}\left(\nabla f^{\ON}(u(x_i,\theta))\right)^\top\left(\nabla A(f^{\ON}(u(x_i,\theta)))-\delta_i\right),
\end{split}
\end{equation}
where $\delta_i\in\bbR^K$ satisfies $\delta_{ik}=\ind{z_i=k}$.
Let $\Delta=\theta-\theta^*$, then we have
\begin{equation}
\begin{split}
    &\langle\nabla \cL_n^{\ON}(\theta)-\nabla \cL_n^{\ON}(\theta^*),\theta-\theta^*\rangle\\
    =&\frac{1}{n}\sum_{i=1}^n\left(\nabla A(f^{\ON}(u(x_i,\theta)))-\nabla A(f^{\ON}(u(x_i,\theta^*)))\right)^\top \nabla f^{\ON}(u(x_i,\theta))\begin{pmatrix}x_i^\top \Delta_{1:p}-\Delta_{p+1}\\
    -\Delta_{p+2}\\
    \vdots\\
    -\Delta_{p+K}\end{pmatrix}\\
    &-\frac{1}{n}\sum_{i=1}^n\epsilon_i^\top\left(\nabla f^{\ON}(u(x_i,\theta))-\nabla f^{\ON}(u(x_i,\theta^*))\right)\begin{pmatrix}x_i^\top \Delta_{1:p}-\Delta_{p+1}\\
    -\Delta_{p+2}\\
    \vdots\\
    -\Delta_{p+K}\end{pmatrix}\\
    =:&\mathrm{\RNum{1}}(\Delta)+\mathrm{\RNum{2}}(\Delta),
\end{split}
\end{equation}
where $\epsilon_i=\delta_i-\nabla A(f^{\ON}(u(x_i,\theta^*)))$. In the following we will show a lower bound for $\mathrm{\RNum{1}}(\Delta)$ and concentrate $\mathrm{\RNum{2}}(\Delta)$ uniformly over the feasible set.
\begin{enumerate}
    \item Lower bounding $\mathrm{\RNum{1}}(\Delta)$\\
    \begin{lemma}\label{lem:ordinal_rsc_term1}
For any $\Delta\in S(R_0,r_0)$,
\begin{equation}
\begin{split}
        \mathrm{\RNum{1}}(\Delta)\geq \frac{\gamma}{n}\sum_{i=1}^n(x_i^\top \Delta_{1:p}-\Delta_{p+1})^2+\gamma\|\Delta_{(p+2):(p+K)}\|_2^2,
\end{split}
\end{equation}
where $\gamma=\min_j\frac{n_j}{\pi_jn_u}\frac{e^{2(C_x+1)(R^*+R_0)}}{32h^{\ON}(R_0,r_0)(1+e^{(C_x+1)(R^*+R_0)})^{4}K}$.
\end{lemma}
By Lemma \ref{lem:ordinal_rsc_term1}, 
\begin{equation*}
\begin{split}
    \mathrm{\RNum{1}}(\Delta)\geq&\frac{\gamma}{n}\sum_{i=1}^n\left((x_i^\top\Delta_{1:p}-\Delta_{p+1})^2+\sum_{j=2}^{K}\Delta_{p+j}^2\right).
\end{split}
\end{equation*}
Let $\widetilde{\Delta}=\Delta_{1:(p+1)}$, $\widetilde{X}=\begin{pmatrix}x_1^\top&-1\\\vdots&\vdots\\x_n^\top&-1\end{pmatrix}$.
In the following we will show that
\begin{equation}\label{eq:ordinal_REC}
    \frac{\gamma}{n}\|\widetilde{X}\widetilde{\Delta}\|_2^2\geq \alpha\|\widetilde{\Delta}\|_2^2-C\alpha\frac{\log J+m}{n}\|\Delta_{1:p}\|_{\omega,2,1}^2,
\end{equation}
where $\alpha=\frac{1}{2}\gamma\min\{\lambda_{\min}(\Sigma),1\}$. 
Since the L.H.S. of \eqref{eq:ordinal_REC} is non-negative, \eqref{eq:ordinal_REC} trivially holds if $\|\widetilde{\Delta}_{1:p}\|_{\omega,2,1}\geq \sqrt{\rho}\|\widetilde{\Delta}\|_2$, where $\rho=\frac{n}{C(\log J+m)}$. Meanwhile, since 
\begin{equation*}
    \frac{1}{n}\bbE\|\widetilde{X}\widetilde{\Delta}\|_2^2=\widetilde{\Delta}^\top \begin{pmatrix}\Sigma &0\\
    0&1\end{pmatrix}\widetilde{\Delta}\geq \min\{\lambda_{\min}(\Sigma),1\}\|\widetilde{\Delta}\|_2^2,
\end{equation*}
we only have to prove that
\begin{equation}\label{eq:ordinal_REC_uniform_dev}
\begin{split}
    \sup_{\|\widetilde{\Delta}_{1:p}\|_{\omega,2,1}\leq \sqrt{\rho},\|\widetilde{\Delta}\|_2\leq 1}\frac{1}{n}\left|\|\widetilde{X}\widetilde{\Delta}\|_2^2-\bbE\|\widetilde{X}\widetilde{\Delta}\|_2^2\right|\leq\frac{\min\{\lambda_{\min}(\Sigma),1\}}{2}.
\end{split}
\end{equation}
Define function 
\begin{equation*}
    g(\widetilde{\Delta},\widetilde{\Delta})=\frac{1}{n}\|\widetilde{X}\widetilde{\Delta}\|_2^2-\frac{1}{n}\bbE\|\widetilde{X}\widetilde{\Delta}\|_2^2,
\end{equation*}
then following similar arguments from the proof of Lemma \ref{lem:RSC}, one can show that
\begin{equation*}
\begin{split}
    &\sup_{\|\widetilde{\Delta}_{1:p}\|_{\omega,2,1}\leq \sqrt{\rho},\|\widetilde{\Delta}\|_2\leq 1}\left|g(\widetilde{\Delta},\widetilde{\Delta})\right|\\
    \leq&9(1+2\max_j\omega_j^{-1})\sup_{S\subset [J],|S|=2\lfloor\rho\rfloor}\sup_{\widetilde{\Delta}\in \mathcal{N}_S}\left|g(\widetilde{\Delta},\widetilde{\Delta})\right|,
\end{split}
\end{equation*}
where $\mathcal{N}_S$ is a $\frac{1}{10}$-net of $\mathcal{K}(\cG_S)=\{\widetilde{\Delta}:\|\widetilde{\delta}\|_2=1,\widetilde{\Delta}_{\mathcal{G}_j}= 0\text{ if } j\notin S\}$, and $|\cN_S|\leq 21^{2\rho m+1}$. The following lemma concentrates $g(\widetilde{\Delta},\widetilde{\Delta})$ for each fixed $\widetilde{\Delta}$ w.h.p.
\begin{lemma}\label{lem:ordinal_rec_dev}
 For any $\widetilde{\Delta}\in \bbR^{p+1}$ such that $\|\widetilde{\Delta}\|_2=1$,
    \begin{equation}
            \bbP\left(\left|g(\widetilde{\Delta},\widetilde{\Delta})\right|\geq t\right)\leq 2\exp\left\{-cn\min\left\{\frac{t^2}{\sigma^4},\frac{t^2}{\sigma^2},\frac{t}{\sigma^2}\right\}\right\}.
        \end{equation}
    for any $t>0$.
\end{lemma}
We apply Lemma \ref{lem:ordinal_rec_dev} with $t=\frac{\min\{\lambda_{\min}(\Sigma),1\}}{18(1+2\max_j\omega_j^{-1})}$ and take a union bound over $\cN_S$ and $S\subset[J]$, then one can show that \eqref{eq:ordinal_REC_uniform_dev} holds with probability at least
\begin{equation*}
    \begin{split}
        &1-2(21)^{2\rho m+1}\binom{J}{2\lfloor \rho\rfloor}\exp\{-cn\}\\
        \geq &1-2\exp\{(2\rho m+1)\log(21)+3\rho\log(eJ/2\lfloor \rho\rfloor)]-cn\}\\
        \geq&1-2\exp\{(2\log (21)+3)\rho(m+\log J)+\log(21)-3\rho\log(2\lfloor \rho\rfloor)-cn\}\\
        \geq&1-2\exp\{-\frac{cn}{2}\}
    \end{split}
\end{equation*}
where the last line holds since $\rho=\frac{n}{C(\log J+m)}$ and thus as long as this constant $C$ is chosen to be larger than $\frac{4\log(21)+6}{c}$ and $\rho\geq 2$, it is guaranteed that $(2\log(21)+3)\rho(m+\log J)\leq \frac{c}{2}n$ and $\log(21)-3\rho\log(2\lfloor\rho\rfloor)$. Therefore, with probability at least $1-2\exp\{-cn\}$,
\begin{equation*}
\begin{split}
    \mathrm{\RNum{1}}(\Delta)\geq &\alpha\|\Delta\|_2^2-C\alpha\frac{\log J+m_1}{n}\|\Delta_{1:p}\|^2_{\omega,2,1}.
\end{split}
\end{equation*}
\item Concentrating $\mathrm{\RNum{2}}(\Delta)$\\
We follow similar arguments from the second part of the proof of Lemma~\ref{lem:RSC}. First we define \begin{equation*}
    \phi(\Delta,\{x_i\}_{i=1}^{n},\{\epsilon_i\}_{i=1}^n)=\frac{\mathrm{\RNum{2}}(\Delta)}{\tau(\Delta)},
\end{equation*} and 
\begin{equation*}
\begin{split}
    g_t^{(1)}((x_1,\epsilon_1),\dots,(x_n,\epsilon_n))=&\sup_{\substack{\frac{t}{2}\leq\tau(\Delta)\leq t\\ \Delta\in\cS(R_0,r_0)}}\phi(\Delta,\{x_i\}_{i=1}^n,\{\epsilon_i\}_{i=1}^n),\\
    g_t^{(2)}((x_1,\epsilon_1),\dots,(x_n,\epsilon_n))=&\sup_{\substack{\frac{t}{2}\leq\tau(\Delta)\leq t\\ \Delta\in\cS(R_0,r_0)}}-\phi(\Delta,\{x_i\}_{i=1}^n,\{\epsilon_i\}_{i=1}^n),
\end{split}
\end{equation*}
for any $t>0$. We will bound $g_t^{(j)}$ and then apply a peeling argument.
By Lemma~\ref{lem:symmetrization}, 
\begin{equation*}
    \bbE g_t^{(j)}((x_1,\epsilon_1),\dots,(x_n,\epsilon_n))\leq\frac{4}{nt}\bbE\sup_{\substack{\frac{t}{2}\leq \tau(\Delta)\leq t\\ \Delta\in\cS(R_0,r_0)}}\sum_{i=1}^nq_i(u_i(\Delta))\varepsilon_i,
\end{equation*}
where $q_i(v)=\epsilon_i^\top \left[\nabla f^{\ON}(u_i(\theta^*)+v)-\nabla f^{\ON}(u_i(\theta^*))\right]v$, and we used $u_i(\theta)$ to denote $u(x_i,\theta)$ for simplicity. The following lemma suggests $q_i(\cdot)$ to be $L$-Lipschitz within the region of our interest. 
\begin{lemma}\label{lem:ordinal_rsc_term2_lip}
 For any vector $\Delta\in\bbR^{p+K}$ such that $\Delta\in \cS(R_0,r_0)$, 
 $1\leq i\leq n$, 
 \begin{equation*}
 \begin{split}
     \left\|\nabla q_i(u_i(\Delta))\right\|_2\leq L,
 \end{split}
 \end{equation*}
 where $L=(\frac{4}{3}(C_x+1)R_0+1)\sqrt{K}h^{\ON}(R_0,r_0)+\frac{(C_x+1)R_0}{4}\sqrt{K}(h^{\ON}(R_0,r_0))^2$.
\end{lemma}
Then by lemma~\ref{lem:contraction_vector_domain} and similar arguments from the proof of Lemma~\ref{lem:RSC}, one can show that
\begin{equation}\label{eq:expt_g}
    \bbE g_t^{(j)}((x_1,\epsilon_1),\dots,(x_n,\epsilon_n))\leq \frac{4\sqrt{2}L}{nt}\bbE\sup_{\substack{\frac{t}{2}\leq \tau(\Delta)\leq t\\ \Delta\in\cS(R_0,r_0)}}\sum_{i=1}^n u_i(\Delta)^\top \widetilde{\varepsilon}_i,
\end{equation}
where $\widetilde{\varepsilon}_{i}\in\bbR^K$ and $\widetilde{\varepsilon}_{ij}, 1\leq i\leq n, 1\leq j\leq K$ are independent Rademacher random variables. One can show that
\begin{equation}\label{eq:sup_expt}
    \begin{split}
        &\frac{4\sqrt{2}L}{nt}\bbE\sup_{\substack{\frac{t}{2}\leq \tau(\Delta)\leq t\\ \Delta\in\cS(R_0,r_0)}}\sum_{i=1}^n u_i(\Delta)^\top \widetilde{\varepsilon}_i\\
        =&\frac{4\sqrt{2}L}{nt}\bbE\sup_{\substack{\frac{t}{2}\leq \tau(\Delta)\leq t\\ \Delta\in\cS(R_0,r_0)}}\Delta^\top\sum_{i=1}^n \begin{pmatrix}x_i\widetilde{\varepsilon}_{i1}\\
        -\widetilde{\varepsilon}_i\end{pmatrix}\\
        \leq&\frac{4\sqrt{2}L}{n}\bbE\Bigg[\sqrt{\frac{n}{\lambda_{\max}(\Sigma)(\log J+m)}}\left\|\sum_{i=1}^nx_i^\top\widetilde{\varepsilon}_{i1}\right\|_{\omega^{-1},2,\infty}\\
        &+\sqrt{\frac{n}{K(\log J+\log(2K))}}\left\|\sum_{i=1}^n\widetilde{\varepsilon}_i\right\|_2\Bigg]
    \end{split}
\end{equation}
Since we have the same sub-Gaussian assumptions on the covariates $x_i$ as the multinomial case, following the same arguments as the proof of Lemma~\ref{lem:mult_dev_expt} would lead us to 
\begin{equation}\label{eq:sup_expt_1}
    \bbE \max_{1\leq j\leq J}\omega_j^{-1}\|X_{\cG_j}^\top \widetilde{\epsilon}_{:,1}\|_2\leq C\lambda^{\frac{1}{2}}_{\max}(\Sigma)\sqrt{(m+\log J)n}.
\end{equation}
Now we provide an upper bound for $\bbE \max_{k}|\sum_{i=1}^n\widetilde{\epsilon}_{ik}|$. Since $\{\widetilde{\epsilon}_{ik}\}_{i=1}^n$ are independent sub-Gaussian random variables with constant parameter, applying Hoeffding type inequality and taking a union bound over $1\leq k\leq K$ would show us that
\begin{equation*}
    \bbP\left(\max_{1\leq k\leq K}\left|\sum_{i=1}^n \widetilde{\epsilon}_{ik}\right|>t\right)\leq 2Ke^{-\frac{Ct^2}{n}},
\end{equation*}
which further implies 
\begin{equation}\label{eq:sup_expt_2}
\begin{split}
    \bbE \max_{1\leq k\leq K}\left|\sum_{i=1}^n \widetilde{\epsilon}_{ik}\right|=&\int_{0}^{\infty} \bbP\left(\max_{1\leq k\leq K}\left|\sum_{i=1}^n \widetilde{\epsilon}_{ik}\right|>t\right)\mathrm{d}t\\
    \leq&\sqrt{\frac{\log 2K}{C}n}+\int_{\sqrt{\frac{\log 2K}{C}n}}^{\infty}e^{-\frac{C}{n}(t-\sqrt{\frac{\log 2K}{C}n})^2}\mathrm{d}t\\
    \leq&C\sqrt{\log(2K)n}.
    \end{split}
\end{equation}
Combining \eqref{eq:expt_g}, \eqref{eq:sup_expt}, \eqref{eq:sup_expt_1} and \eqref{eq:sup_expt_2}, we obtain that for $j=1,2$, 
\begin{equation*}
\begin{split}
    \bbE g_t^{(j)}((x_1,\epsilon_1),\dots,(x_n,\epsilon_n))\leq CL.
\end{split}
\end{equation*}
Furthermore, one can show that for $1\leq i\leq n$ and any $\Delta\in S(R_0,r_0)$,
\begin{equation*}
    \begin{split}
        &\left|\epsilon_i^\top (\nabla f^{\ON}(u_i(\theta^*+\Delta))-\nabla f^{\ON}(u_i(\theta^*)))\begin{pmatrix}x_i,0\\-I_K\end{pmatrix}^\top\Delta\right|\\
        \leq&\|x_i\epsilon_i^\top (\nabla f^{\ON}(u_i(\theta^*+\Delta))-\nabla f^{\ON}(u_i(\theta^*)))_{:,1}\|_{\omega^{-1},2,\infty}\|\Delta_{1:p}\|_{\omega,2,1}\\
        &+\|(\nabla f^{\ON}(u_i(\theta^*+\Delta))-\nabla f^{\ON}(u_i(\theta^*)))^\top\epsilon_i\|_2\|\Delta\|_2\\
        \leq&\|(\nabla f^{\ON}(u_i(\theta^*+\Delta))-\nabla f^{\ON}(u_i(\theta^*)))^\top\epsilon_i\|_{\infty}\\
        &\cdot(\max_j\omega_j^{-1}\|(x_i)_{\cG_j}\|_2\|\Delta_{1:p}\|_{\omega,2,1}+\sqrt{K}\|\Delta\|_2)\\
        \leq &Ch^{\ON}(R_0,r_0)(C_x\sqrt{m}\|\Delta_{1:p}\|_{\omega,2,1}+\sqrt{K}\|\Delta\|_2),
    \end{split}
\end{equation*}
where the last line is due to Assumption~\ref{assump:subgauss_row}, ~\ref{assump:scaling}, and the fact that 
\begin{equation}\label{eq:f_epsilon_bnd}
\begin{split}
    \|\nabla(f^{\ON}(u_i(\theta^*))\epsilon_i\|_{\infty}\leq \frac{1}{2}h^{\ON}(0,0)\leq \frac{1}{2}h^{\ON}(R_0,r_0),\\
    \sup_{\Delta\in \cS(R_0,r_0)}\|\nabla(f^{\ON}(u_i(\theta^*+\Delta)))\epsilon_i\|_{\infty}\leq \frac{1}{2}h^{\ON}(R_0,r_0),
\end{split}
\end{equation} 
which is implied by Lemma~\ref{lem:ordinal_p_alpha_f_bnd} and $\|\epsilon_i\|_1\leq 2$.
Therefore,
\begin{equation}\label{eq:expt_g_bnd}
    \begin{split}
        &\sup_{\Delta\in \cS(R_0,r_0)}\frac{\left|\epsilon_i^\top (\nabla f^{\ON}(u_i(\theta^*+\Delta))-\nabla f^{\ON}(u_i(\theta^*)))\begin{pmatrix}x_i,0\\-I_K\end{pmatrix}\Delta\right|}{n\tau(\Delta)}\\
        \leq&Ch^{\ON}(R_0,r_0)\bigg\{C_x\sqrt{\frac{m}{\lambda_{\max}(\Sigma)(\log J+m)n}}+\sqrt{\frac{1}{(\log J+\log(2K))n}}\bigg\}\\
        \leq&Ch^{\ON}(R_0,r_0)\bigg(C_x\sqrt{\frac{m}{\lambda_{\max}(\Sigma)}}+1\bigg)\frac{1}{\sqrt{n\log J}}.
    \end{split}
\end{equation}
The following bounded difference results are then directly implied by \eqref{eq:expt_g_bnd}:
\begin{equation*}
\begin{split}
    &\bigg|g_t^{(j)}((x_1,\epsilon_1),\dots,(x_i,\epsilon_i),\dots,(x_n,\epsilon_n))\\
    &-g_t^{(j)}((x_1,\epsilon_1),\dots,(x_i',\epsilon_i'),\dots,(x_n,\epsilon_n))\bigg|\\
    \leq&Ch^{\ON}(R_0,r_0)\bigg(C_x\sqrt{\frac{m}{\lambda_{\max}(\Sigma)}}+1\bigg)\frac{1}{\sqrt{n\log J}},
    \end{split}
\end{equation*}
and thus applying the bounded difference inequality~\citep{mcdiarmid1989method} upon\\ $g_t^{(j)}((x_1,\epsilon_1),\dots,(x_n,\epsilon_n))$ would lead us to
\begin{equation}\label{eq:ordinal_rec_term2_slice_bnd}
\begin{split}
    \sup_{\substack{\frac{t}{2}\leq\|\Delta\|_{\omega,2,1}\leq t\\ \Delta\in\cS(R_0,r_0)}}|\phi(\Delta,\{x_i\}_{i=1}^n, \{\epsilon_i\}_{i=1}^n)|\leq CL,
\end{split}
\end{equation}
with probability at least 
\begin{equation*}
  \begin{split}
      &1-2\exp\left\{-\frac{cL^2\log J}{(h^{\ON}(R_0,r_0))^2(C_x^2\lambda_{\max}^{-1}(\Sigma)m+1)}\right\}\\
      \geq&1-2\exp\{-\frac{cK\log J}{C_x^2m+1}\},
  \end{split}  
\end{equation*}
where we have applied the fact that $L^2\geq K(h^{\ON}(R_0,r_0))^2$ and $\lambda_{\max}(\Sigma)\geq c$.
Now we apply a peeling argument to extend the above bound to $\Delta\in \cS(R_0,r_0)$. First note that for any $\Delta\in \cS(R_0,r_0)$, 
\begin{equation*}
\begin{split}
    \tau(\Delta)\leq& \left(\max_j\omega_j\sqrt{\frac{\lambda_{\max}(\Sigma)(\log J+m)}{n}}+2\sqrt{\frac{K(\log J+\log(2K))}{n}}\right)R_0\\
    =:&C^{(1)}_n
\end{split}
\end{equation*} 
Define
$$C_n^{(2)}=\frac{CL\log J\min\{C_x^{-2},1\}}{(\frac{4}{3}h^{\ON}(R_0,r_0)+\frac{1}{4}(h^{\ON}(R_0,r_0))^2)mn}$$
$N=\log_2(\frac{C_n^{(1)}}{C_n^{(2)}})$, then one can show that
\begin{equation*}
\begin{split}
    \cS(R_0,r_0)\subset &\cup_{k=1}^N\{\Delta: 2^{k-1}C_n^{(2)}\leq \|\Delta\|_{\omega,2,1}\leq 2^kC_n^{(2)}\}\\
    &\cup\{\Delta: 0\leq \|\Delta\|_{\omega,2,1}\leq C_n^{(2)}\}.
\end{split}
\end{equation*}
Similarly from the proof of Lemma~\ref{lem:RSC}, we consider function $$\widetilde{\phi}(r,d,\{x_i\}_{i=1}^n,\{\epsilon_i\}_{i=1}^n)=\phi(rd,\{x_i\}_{i=1}^n, \{\epsilon_i\}_{i=1}^n)$$ 
for any $r\geq 0$ and $d\in \bbR^{p+K}$ such that $\tau(d)=1$. Some calculation shows that
\begin{equation*}
\begin{split}
    &\left|\frac{\partial \widetilde{\phi}(r,d,\{x_i\}_{i=1}^n,\{\epsilon_i\}_{i=1}^n)}{\partial r}\right|\\
    =&\left|\frac{1}{n}\sum_{i=1}^n \left\langle \nabla^2f^{\ON}(u_i(\theta^*+rd)),\epsilon_i\otimes\left(\begin{pmatrix}x_i,0\\
    -I_K\end{pmatrix}^\top d\right)\otimes\left(\begin{pmatrix}x_i,0\\
    -I_K\end{pmatrix}^\top d\right)\right\rangle\right|\\
    \leq&\max_i\|\nabla^2f^{\ON}(u_i(\theta^*+rd))\|_{\infty}\|\epsilon_i\|_1\left\|\begin{pmatrix}x_i,0\\
    -I_K\end{pmatrix}^\top d\right\|_1^2\\
    \leq&2\max\{C_x^2,1\}\|d\|_1^2\max_i\|\nabla^2f^{\ON}(u_i(\theta^*+rd))\|_{\infty},
\end{split}
\end{equation*}
where the last line is due to that $\|\epsilon_i\|_1\leq \|\delta_i\|_1+\|\bbE(\delta_i|x_i)\|_1\leq 2$, and 
\begin{equation*}
    \left\|\begin{pmatrix}x_i,0\\
    -I_K\end{pmatrix}^\top d\right\|_1\leq |x_i^\top d_{1:p}|+\|d_{(p+1):(p+K)}\|_1\leq \max\{C_x,1\}\|d\|_1.
\end{equation*}
Note that 
\begin{equation*}
    \begin{split}
        \|d\|_1\leq& \max_j\omega_j^{-1}\sqrt{m}\|d_{1:p}\|_{\omega,2,1}+\sqrt{K}\|d_{(p+1):(p+K)}\|_2\\
        \leq&\max_j\omega_j^{-1}\sqrt{\frac{mn}{\lambda_{\max}(\Sigma)(\log J+m)}}+\sqrt{\frac{n}{\log J+\log(2K)}}\\
        \leq&C\sqrt{\frac{mn}{\log J}},
    \end{split}
\end{equation*}
and by Lemma~\ref{lem:ordinal_p_alpha_f_bnd}, we have that
\begin{equation*}
    \begin{split}
    &\left|\frac{\partial \widetilde{\phi}(r,d,\{x_i\}_{i=1}^n,\{\epsilon_i\}_{i=1}^n)}{\partial r}\right|\\
    \leq&C\max\{C_x^2,1\}(\frac{4}{3}h^{\ON}(R_0,r_0)+\frac{1}{4}(h^{\ON}(R_0,r_0))^2\frac{mn}{\log J},
    \end{split}
\end{equation*}
for any $rd\in S(R_0,r_0)$.
Therefore, 
\begin{equation*}
    \begin{split}
        &\sup_{\substack{0\leq \|\Delta\|_{\omega,2,1}\leq C_n^{(2)}\\\Delta\in \cS(R_0,r_0)}}\left|\phi(\Delta,\{x_i\}_{i=1}^n, \{\epsilon_i\}_{i=1}^n)\right|\\
        \leq&\sup_{\substack{\|\Delta\|_{\omega,2,1}=C_n^{(2)}\\ \Delta\in \cS(R_0,r_0)}}\left|\phi(\Delta,\{x_i\}_{i=1}^n, \{\epsilon_i\}_{i=1}^n)\right|\\
        &+C\max\{C_x^2,1\}(\frac{4}{3}h^{\ON}(R_0,r_0)+\frac{1}{4}(h^{\ON}(R_0,r_0))^2\frac{mn}{\log J}C_n^{(2)}\\
        \leq&\sup_{\substack{\|\Delta\|_{\omega,2,1}=C_n^{(2)}\\\Delta\in \cS(R_0,r_0)}}\left|\phi(\Delta,\{x_i\}_{i=1}^n, \{\epsilon_i\}_{i=1}^n)\right|+CL.
    \end{split}
\end{equation*}
Now we apply the probabilistic bound \eqref{eq:ordinal_rec_term2_slice_bnd} for $t=C_n^{(2)},2C_n^{(2)},\dots,2^{N-1}C_n^{(2)}$ and take a union bound, which leads to the following:
\begin{equation*}
\begin{split}
    \sup_{\Delta\in \cS(R_0,r_0)}\left|\phi(\Delta,\{x_i\}_{i=1}^n,\{\epsilon_i\}_{i=1}^n\right|\leq CL,
    \end{split}
\end{equation*}
with probability at least $1-2N\exp\{-\frac{cK\log J}{C_x^2m+1}\}$.
Now note that
\begin{equation*}
\begin{split}
     N=&\log_2\Bigg\{\left[\max_j\omega_j m\sqrt{\frac{\lambda_{\max}(\Sigma)(\log J+m)n}{(\log J)^2}}+m\sqrt{\frac{K(\log J+\log(2K))n}{(\log J)^2}}\right]\\
    &\cdot \frac{R_0(\frac{4}{3}h^{\ON}(R_0,r_0)+\frac{1}{4}(h^{\ON}(R_0,r_0))^2)}{C\min\{C_x^{-1},1\}L}\Bigg\}\\
    \leq&\log_2\left\{Cm\sqrt{\frac{n}{\log J}}(\max_j\omega_j\lambda^{\frac{1}{2}}_{\max}(\Sigma)+\sqrt{K})\frac{R_0(\frac{4}{3}h^{\ON}(R_0,r_0)+\frac{1}{4}(h^{\ON}(R_0,r_0))^2)}{C\min\{C_x^{-1},1\}L}\right\}.
\end{split}
\end{equation*}
The following lemma shows another property of $R_0$ that can be implied by Assumption \ref{assump:ordinal_region_bound}.
\begin{lemma}\label{lem:ordinal_region_consequence}
    Suppose $R_0>0$ satisfies Assumption \ref{assump:ordinal_region_bound}, then 
    \begin{equation*}
        \begin{split}
            R_0&\leq \min_j\frac{n_j}{\pi_jn_u}\frac{\left(1+e^{(C_x+1)(R^*+R_0)}\right)^{-6}}{512(C_x+1)K^3(1+\frac{2}{r^* - r_0})^3(2+\frac{2}{r^*-r_0})}\\
            &\leq \min_j\frac{n_j}{\pi_jn_u}\frac{\left(1+e^{(C_x+1)(R^*+R_0)}\right)^{-2}}{512(C_x+1)K^3(h^{\ON}(R_0,r_0))^3(1+h^{\ON}(R_0,r_0))}.
        \end{split}
    \end{equation*}
\end{lemma}
Lemma \ref{lem:ordinal_region_consequence} implies that
\begin{equation*}
    \begin{split}
      &R_0(\frac{4}{3}h^{\ON}(R_0,r_0)+\frac{1}{4}(h^{\ON}(R_0,r_0))^2)\\  \leq&\min_j\frac{n_j}{\pi_jn_u}\frac{\left(1+e^{(C_x+1)(R^*+R_0)}\right)^{-2}}{372(C_x+1)K^3(h^{\ON}(R_0,r_0))^2}\\
        \leq&\min_j\frac{n_j}{\pi_jn_u}\frac{3e^{2(C_x+1)(R^*+C)}(1+e^{(C_x+1)(R^*+C)})^{-4}}{512(C_x+1)(h^{\ON}(C,r_0))^2}\\
        \leq&\frac{C}{C_x+1}.
    \end{split}
\end{equation*}
The inequality above combined with the fact that $L\geq 1$ leads to $\frac{R_0(\frac{4}{3}h^{\ON}(R_0,r_0)+\frac{1}{4}(h^{\ON}(R_0,r_0))^2)}{C\min\{C_x^{-1},1\}L}\leq C$. Meanwhile, since $m\leq Cn$, $\max_j\omega_j\leq C\sqrt{\frac{n}{\log J}}$, $K\leq Cn$, and $\lambda_{\max}(\Sigma)\leq C\sigma^2\leq C$, clearly we have $N\leq C\log_2n$.
Therefore,
\begin{equation*}
    \left|\mathrm{\RNum{2}}(\Delta)\right|\leq CL\tau(\Delta)
\end{equation*}
holds for all $\Delta\in S(R_0,r_0)$ with probability at least 
\begin{equation*}
\begin{split}
    1-C\log_2n\exp\{-\frac{cK\log J}{C_x^2m+1}\}\geq &1-\exp\{\log(C\log n)-\frac{cK\log J}{C_x^2m+1}\}\\
    \geq& 1-\exp\{-\frac{cK\log J}{C_x^2m+1}\}.
\end{split}
\end{equation*}
\end{enumerate}
\end{proof}
\section{Proof of Supporting Lemmas for Proving Theorem~\ref{thm:mult_upp} and Theorem~\ref{thm:ordinal_upp}}
\subsection{Supporting Lemmas for the Multinomial-PU Model}\label{sec:supp_lemma_MN}
\begin{proof}[Proof of Lemma \ref{lem:sketch_MN_I_lrbnd}]
    Recall the definition of $A(\cdot)$, one can show that for any $\eta\in \bbR^K$,
    \begin{equation}
        (\nabla A(\eta))_i=\frac{e^{\eta_i}}{1+\sum_{k=1}^K e^{\eta_k}}, 1\leq i\leq K,
    \end{equation}
    and
    \begin{equation}
        \begin{split}
            (\nabla^2 A(\eta))_{jk}=&\frac{e^{\eta_j}}{1+\sum_{l=1}^K e^{\eta_l}}\left(\ind{j=k}-\frac{e^{\eta_k}}{1+\sum_{l=1}^K e^{\eta_l}}\right)\\
            =&(\nabla A(\eta))_j\left(\ind{j=k}-(\nabla A(\eta))_k\right).
        \end{split}
    \end{equation}
    Thus the smallest eigenvalue of $\nabla^2 A(\eta)$ is lower bounded as follows:
    \begin{equation}
    \begin{split}
        \lambda_{\min}(\nabla^2 A(\eta))=&\inf_{\|v\|_2=1}v^\top \nabla^2 A(\eta)v\\
        =&\inf_{\|v\|_2=1}\sum_{j=1}^K(\nabla A(\eta))_jv_j^2-\left(\sum_{j=1}^K (\nabla A(\eta))_jv_j\right)^2\\
        \geq&\inf_{\|v\|_2=1}\sum_{j=1}^K(\nabla A(\eta))_jv_j^2-\sum_{j=1}^K(\nabla A(\eta))_jv_j^2 \sum_{j=1}^K (\nabla A(\eta))_j\\
        \geq &\min_j (\nabla A(\eta))_j\left(1-\sum_{j=1}^K (\nabla A(\eta))_j\right)\\
        =&\frac{\min_j e^{\eta_j}}{\left(1+\sum_{l=1}^K e^{\eta_l}\right)^2}.
    \end{split}
    \end{equation}
    Recall that $R^*=\max_{1\leq i\leq K}\|(\Theta^*)_{:,i}\|_1$ and $\Theta\in \Theta^* + \mathbb{B}_{1,\infty}(R)$. Then we have $$\exp\{[f(\Theta^\top x_i)]_k\} = \frac{\frac{n_k}{\pi_k n_u}\exp\{\Theta_{:,k}^\top x_i\}}{1+\sum_{j=1}^K\exp\{\Theta_{:,j}^\top x_i\}}\geq \frac{\min_j\frac{n_j}{\pi_j n_u}e^{-C_x(R+R^*)}}{1+Ke^{C_x(R+R^*)}}, $$
    $$
    \sum_{j=1}^K\exp\{[f(\Theta^\top x_i)]_j\}\leq \max_j\frac{n_j}{\pi_j n_u}.
    $$
    Thus we have
    \begin{equation}\label{eq:MN_mineigen_A}
    \begin{split}
        \lambda_{\min}(\nabla^2 A(f(\Theta^\top x_i))\geq&\frac{\exp\{\min_j (f(\Theta^\top x_i))_j\}}{\left[1+\sum_{j=1}^K \exp\{(f(\Theta^\top x_i))_j\}\right]^2}\\
        \geq&\frac{e^{-C_x(R+R^*)}\min_j\frac{n_j}{\pi_j n_u}}{(1+\max_j\frac{n_j}{\pi_j n_u})^2(1+Ke^{C_x(R+R^*)})}
    \end{split}
    \end{equation}
Now we lower bound the smallest eigenvalue of $$\nabla f(v)^\top\nabla f(v)=(I_K-1_K \nabla A(v)^\top)^\top(I_K-1_K \nabla A(v)^\top).$$ For any vector $z\in \bbR^K$, let $y=\nabla f(v) z=z-(\nabla A(v)^\top z)1_K$. Then one can show that $$\nabla A(v)^\top y=\nabla A(v)^\top z(1-\nabla A(v)^\top 1_K)=(1+\sum_{j=1}^Ke^{v_j})^{-1}\nabla A(v)^\top z,$$ which implies $y=z-(1+\sum_{j=1}^Ke^{v_j})\nabla A(v)^\top y 1_K$, and thus
\begin{equation}
    \|z\|_2\leq \|y\|_2+\sqrt{K}\left(1+\sum_{j=1}^Ke^{v_j}\right)\|\nabla A(v)\|_2\|y\|_2.
\end{equation}
Therefore, 
\begin{equation}
\begin{split}
    \lambda_{\min}(\nabla f(v)^\top \nabla f(v))=&\inf_{z\in \bbR^K}\frac{\|\nabla f(v) z\|_2^2}{\|z\|_2^2}\\
    \geq &\left[1+\sqrt{K}\left(1+\sum_{j=1}^Ke^{v_j}\right)\|\nabla A(v)\|_2\right]^{-2}\\
    =&\left[1+\sqrt{K}\left(\sum_{j=1}^Ke^{2v_j}\right)^{\frac{1}{2}}\right]^{-2},
\end{split}
\end{equation}
and 
\begin{equation}\label{eq:MN_mineigen_f}
    \lambda_{\min}(\nabla f(\Theta^\top x_i)^\top \nabla f(\Theta^\top x_i))\geq  \left[1+Ke^{C_x(R+R^*)}\right]^{-2}.
\end{equation}
Combining \eqref{eq:MN_mineigen_A} and \eqref{eq:MN_mineigen_f} completes the proof.
\end{proof}
\begin{proof}[Proof of Lemma \ref{lem:mult_rsc_term1}]
By Taylor's theorem, there exists $t\in [0,1]$ such that
    \begin{equation}
        \begin{split}
            \mathrm{\RNum{1}}(\Delta)=&\frac{1}{n}\sum_{i=1}^n (x_i^\top \Delta)\nabla f(\Theta^\top x_i)^\top \left[\nabla A(f(\Theta^\top x_i))-\nabla A(f(\Theta^{*\top}x_i))\right]\\
            =&\frac{1}{n}\sum_{i=1}^n (\Delta^\top x_i)^\top\nabla f(\Theta^\top x_i)^\top \nabla^2 A(f(\Theta_t^\top x_i))\nabla f(\Theta_t^\top x_i)(\Delta^\top x_i)\\
            \geq&\frac{1}{n}\sum_{i=1}^n (\Delta^\top x_i)^\top\nabla f(\Theta_t^\top x_i)^\top \nabla^2 A(f(\Theta_t^\top x_i))\nabla f(\Theta_t^\top x_i)(\Delta^\top x_i)\\
            &-\frac{1}{n}\sum_{i=1}^n \|\Delta^\top x_i\|_2^2\left\|\nabla f(\Theta_t^\top x_i)^\top\nabla^2 A(f(\Theta_t^\top x_i))\left[\nabla f(\Theta_t^\top x_i)-\nabla f(\Theta^\top x_i)\right]\right\|,
        \end{split}
    \end{equation}
    where $\Theta_t=t\Theta+(1-t)\Theta^*$ lies between $\Theta$ and $\Theta^*$. 
By Lemma \ref{lem:sketch_MN_I_lrbnd}, we have 
\begin{equation}
\begin{split}
    &(\Delta^\top x_i)^\top\nabla f(\Theta_t^\top x_i)^\top \nabla^2 A(f(\Theta_t^\top x_i))\nabla f(\Theta_t^\top x_i)(\Delta^\top x_i)\\
    \geq &(h^{\MN}(R)+4C_xR)\|\Delta^\top x_i\|_2^2.
\end{split}
\end{equation}
On the other hand, since $\nabla f(v)=I_K-1_K\nabla A(v)^\top$, one can show that
\begin{equation}
    \begin{split}
        &\left\|\nabla f(\Theta_t^\top x_i)^\top\nabla^2 A(f(\Theta_t^\top x_i))\left[\nabla f(\Theta_t^\top x_i)-\nabla f(\Theta^\top x_i)\right]\right\|\\
        \leq &\left\|\nabla f(\Theta_t^\top x_i)^\top\nabla^2 A(f(\Theta_t^\top x_i))1_K\right\|_2\left\|\nabla A(\Theta_t^\top x_i)-\nabla A(\Theta^\top x_i)\right\|\\
        \leq&2\left\|\nabla A(\Theta_t^\top x_i)-\nabla A(\Theta^\top x_i)\right\|\\
        \leq&4C_xR,
    \end{split}
\end{equation}
where the third line is due to that
\begin{equation}
    \begin{split}
        &\left\|\nabla f(\Theta_t^\top x_i)^\top\nabla^2 A(f(\Theta_t^\top x_i))1_K\right\|_2\\
        \leq&\|\nabla^2 A(f(\Theta_t^\top x_i))1_K\|_2+1_K^\top \nabla^2 A(f(\Theta_t^\top x_i))1_K\|\nabla A(\Theta_t^\top x_i)\|_2\\
        \leq&1+\|\nabla A(\Theta_t^\top x_i)\|_2\\
        \leq&2,
    \end{split}
\end{equation}
and the last line is due to that
\begin{equation}
    \begin{split}
        &\left\|\nabla A(\Theta_t^\top x_i)-\nabla A(\Theta^\top x_i)\right\|_2\\
        \leq&\sup_{\eta\in \bbR^K}\left[\sum_{i=1}^K\|(\nabla^2 A(\eta))_{i,:}\|^2_1\right]^{\frac{1}{2}}\|(\Theta_t-\Theta)^\top x_i\|_{\infty}\\
        \leq &2\|(\Theta_t-\Theta)^\top x_i\|_{\infty}\\
        \leq &2C_xR.
    \end{split}
\end{equation}
Therefore,
\begin{equation}
   \mathrm{\RNum{1}}(\Delta) \geq \frac{h^{\MN}(R_0)}{n}\sum_{i=1}^n \|\Delta^\top x_i\|_2^2,
\end{equation}
where we have applied the fact that $R\leq R_0$ and $h^{\MN}(\cdot)$ is a decreasing function.
\end{proof}
\begin{proof}[Proof of Lemma \ref{lem:l1_l0_balls}]
    Let \begin{equation}
        A_1=\bbB_{\omega,2,1}(\sqrt{\rho})\cap \bbB_{2}(1)
    \end{equation} 
    and 
    \begin{equation}
        A_2=\left(1+\frac{2}{\min_j \omega_j}\right)\mathrm{cl}\{\mathrm{conv}\{\bbB_{\cG,0}(\lfloor \rho\rfloor)\cap \bbB_2(1)\}\}.
    \end{equation} 
    Since both $A_1$ and $A_2$ are convex sets, it suffices to show that support function $\phi_{A_1}(U)\leq \phi_{A_2}(U)$ for any $U\in \bbR_{p\times K}$, where $\phi_{A_1}(U)=\sup_{V\in A_1}\langle V,U\rangle$ and $\phi_{A_2}(U)=\sup_{V\in A_2}\langle V,U\rangle$.
    Let $\cG_{\lfloor \rho\rfloor}$ be the union of $\lfloor \rho\rfloor$ $\cG_j$'s such that $\|U_{\cR_j,\cC_j}\|_2$ are the largest. Then 
    \begin{equation}
    \begin{split}
        \phi_{A_1}(U)\leq &\|U_{\cG_{\lfloor \rho\rfloor}}\|_2+\sqrt{\rho}\|U_{\cG^c_{\lfloor \rho\rfloor}}\|_{\omega^{-1},2,\infty}\\
        \leq&\|U_{\cG_{\lfloor \rho\rfloor}}\|_2+\frac{\sqrt{\rho}}{\sqrt{\lfloor \rho\rfloor}\min_j\omega_j}\|U_{\cG_{\lfloor \rho\rfloor}}\|_2\\
        \leq&\left(1+\frac{2}{\min_j\omega_j}\right)\|U_{\cG_{\lfloor \rho\rfloor}}\|_2.
    \end{split}
    \end{equation}
    On the other hand, one can also show that 
    \begin{equation}
        \begin{split}
            \phi_{A_2}(U)=&\left(1+\frac{2}{\min_j\omega_j}\right)\sup_{V\in \bbB_{\cG,0}(\lfloor \rho\rfloor)\cap\bbB_2(1)}\langle V,U\rangle\\
            =&\left(1+\frac{2}{\min_j\omega_j}\right)\|U_{\cG_{\lfloor \rho\rfloor}}\|_2.
        \end{split}
    \end{equation}
    Therefore, we have $A_1\subset A_2$.
    \end{proof}
    \begin{proof}[Proof of Lemma \ref{lem:rec_dev}]
        By assumption \ref{assump:subgauss_row}, $\{(x_i\Delta_{:,j}\}_{i=1}^n$ are independent sub-Gaussian random variables with parameter $\sigma\|\Delta_{:,j}\|_2$, which implies that $\|(x_i^\top \Delta_{:,j})^2\|_{\psi_1}\leq C\sigma^2\|\Delta_{:,j}\|_2^2$. Thus 
        \begin{equation}\label{eq:rec_dev_col}
            \bbP\left(\left|\Delta_{:,j}^\top \left(\frac{1}{n}X^\top X-\Sigma\right)\Delta_{:,j}\right|\geq t\|\Delta_{:,j}\|_2^2\right)\leq 2\exp\left\{-cn\min\left\{\frac{t^2}{\sigma^4},\frac{t}{\sigma^2}\right\}\right\}.
        \end{equation}
        Since $g(\Delta,\Delta)=\sum_{j=1}^K\Delta_{:,j}^\top \left(\frac{1}{n}X^\top X-\Sigma\right)\Delta_{:,j}$, \eqref{eq:rec_dev_col} implies that
        \begin{equation}
            \bbP\left(\left|g(\Delta,\Delta)\right|\geq t\right)\leq 2K\exp\left\{-cn\min\left\{\frac{t^2}{\sigma^4},\frac{t}{\sigma^2}\right\}\right\}.
        \end{equation}
    \end{proof}
    \begin{proof}[Proof of Lemma~\ref{lem:symmetrization}]
    Let $\{\widetilde{U}_i\}_{i=1}^n$ be independent copies of $\{U_i\}_{i=1}^n$, and $\{\varepsilon_i\}_{i=1}^n$ be independent Rademacher random variables. Then one can show that
    \begin{equation*}
        \begin{split}
            &\bbE_{U}\left(\sup_{\gamma\in\Gamma}\sum_{i=1}^n(\gamma(U_i)-\bbE(\gamma(U_i)))\right)\\
            =&\bbE_U\left(\sup_{\gamma\in\Gamma}\sum_{i=1}^n(\gamma(U_i)-\bbE_{\widetilde{U}}(\gamma(\widetilde{U}_i)))\right)\\
            =&\bbE_U\left(\sup_{\gamma\in\Gamma}\bbE_{\widetilde{U}}\left(\sum_{i=1}^n(\gamma(U_i)-\gamma(\widetilde{U}_i))|U\right)\right)\\
            \leq&\bbE_{U,\widetilde{U}}\left(\sup_{\gamma\in\Gamma}\sum_{i=1}^n(\gamma(U_i)-\gamma(\widetilde{U}_i))\right)\\
            =&\bbE_{\varepsilon}\bbE_{U,\widetilde{U}}\left(\sup_{\gamma\in\Gamma}\sum_{i=1}^n\varepsilon_i(\gamma(U_i)-\gamma(\widetilde{U}_i))\right),
        \end{split}
    \end{equation*}
    where the 4th line is due to that 
    \begin{equation*}
        \sum_{i=1}^n(\gamma(U_i)-\gamma(\widetilde{U}_i))\leq \sup_{\gamma\in\Gamma}\sum_{i=1}^n(\gamma(U_i)-\gamma(\widetilde{U}_i))
    \end{equation*} for any $\gamma\in \Gamma$, and the last line is due to that $U_i,\widetilde{U}_i, i=1,\dots,n$ are all independent, $U_i$ and $\widetilde{U}_i$ are identically distributed. Furthermore, we have
    \begin{equation*}
        \begin{split}
            &\bbE_{\varepsilon}\bbE_{U,\widetilde{U}}\left(\sup_{\gamma\in\Gamma}\sum_{i=1}^n\varepsilon_i(\gamma(U_i)-\gamma(\widetilde{U}_i))\right)\\
            \leq&\bbE\left(\sup_{\gamma\in\Gamma}\sum_{i=1}^n\varepsilon_i\gamma(U_i)\right)+\bbE\left(\sup_{\gamma\in\Gamma}\sum_{i=1}^n(-\varepsilon_i)\gamma(\widetilde{U}_i)\right)\\
            =&2\bbE\left(\sup_{\gamma\in\Gamma}\sum_{i=1}^n\varepsilon_i\gamma(U_i)\right),
        \end{split}
    \end{equation*}
    where we have utilized the symmetricity of $\varepsilon_i$.
    \end{proof}
     \begin{proof}[Proof of Lemma \ref{lem:liptchitz-h}]
    First note that
    \begin{equation}
    \begin{split}
        \nabla h_i(u)=&\left[\nabla A(\Theta^{*\top}x_i)-\nabla A(\Theta^{*\top}x_i+u)\right]1_K^\top \epsilon_i\\
        &-\nabla^2 A(\Theta^{*\top}x_i+u)^\top u1_K^\top \epsilon_i,
    \end{split}
    \end{equation}
    and $|1_K^\top \epsilon_i|=|1_K^\top \delta_i-1_K^\top\bbE\delta_i|\leq 1$. Recall the definition of $A(u)=\log(1+\sum_{k=1}^Ke^{u_i})$, one can show that $\|\nabla A(u)\|_2\leq 1$. As shown in the proof of Lemma \ref{lem:mult_rsc_term1}, $\nabla^2 A(u)$ is positive definite with each entry 
    \begin{equation*}
        (\nabla^2 A(u))_{jk}=\frac{e^{u_j}}{1+\sum_{l=1}^Ke^{u_l}}\left(\ind{j=k}-\frac{e^{u_k}}{1+\sum_{l=1}^Ke^{u_l}}\right).
    \end{equation*}
    Thus 
    \begin{equation*}
        \begin{split}
            \lambda_{\max}(\nabla^2 A(u))=&\sup_{\|v\|_2=1}v^\top \nabla^2A(u)v\\
            \leq&\sup_{\|v\|_2=1}\sum_{j=1}^K\frac{v_j^2e^{u_j}}{1+\sum_{l=1}^Ke^{u_l}}\\
            \leq&1,
        \end{split}
    \end{equation*}
    which implies $\|\nabla^2A(u)v\|_2\leq \|v\|_2$ for any $u,v\in \bbR^K$. Meanwhile, for any $u\in \bbR^K$,
    \begin{equation*}
        \begin{split}
            \|\nabla^2A(u)u\|_2^2=&\sum_{i=1}^K\left(\frac{u_ie^{u_i}}{1+\sum_{k=1}^Ke^{u_k}}-\frac{\sum_{j=1}^Ku_j e^{u_i+u_j}}{(1+\sum_{k=1}^Ke^{u_k})^2}\right)^2\\
            =&\frac{\sum_{i=1}^K \left(u_ie^{u_i}+\sum_{j=1}^K(u_i-u_j)e^{u_i+u_j}\right)^2}{(1+\sum_{k=1}^Ke^{u_k})^4}\\
            \leq&\frac{K^2\sum_{i=1}^K e^{4u_i}}{(1+\sum_{k=1}^Ke^{u_k})^4}\\
            \leq &K^2,
        \end{split}
    \end{equation*}
    where we have applied the fact that $x\leq e^x$ for any $x\in \bbR$ on the third line. Therefore,
    \begin{equation}
        \begin{split}
            \|\nabla h_i(u)\|_2\leq& 2+\|\nabla^2A(\Theta^{*\top} x_i+u)^\top u\|_2\\
            \leq& 2+\|\nabla^2A(\Theta^{*\top} x_i+u)^\top \Theta^*x_i\|_2\\
            &+\|\nabla^2A(\Theta^{*\top} x_i+u)^\top (\Theta^*x_i+u)\|_2\\
            \leq& 2+\|\Theta^*x_i\|_2+K\\
            \leq &K+\sqrt{K}R^*C_x+2,
        \end{split}
    \end{equation}
    where we used the fact that $\|\Theta^*\|_{1,\infty}\leq R^*$ and $\|x_i\|_{\infty}\leq C_x$ on the last line.
    \end{proof}
    \begin{proof}[Proof of Lemma~\ref{lem:mult_dev_expt}]
     Let $E\in \bbR^{n\times K}$ such that $E_{ij}=\widetilde{\varepsilon}_{ij}$, then 
    \begin{equation*}
        \bbE\left\|\sum_{i=1}^nx_i\widetilde{\varepsilon}_i^\top\right\|_{\omega^{-1},2,\infty}=\bbE\max_{1\leq j\leq J}\frac{1}{\omega_j}\|X_{\cR_j}^\top E_{\cC_j}\|_2.
    \end{equation*}
    First note that by similar arguments from the proof of Lemma \ref{lem:multinomial_dev_bnd}, one can show that
    \begin{equation}
            \bbE\left[\|X_{\cR_j}^\top E_{\cC_j}\|_2|X\right]\leq \left(c_j\tr(X_{\cR_j}X_{\cR_j}^\top)\right)^{\frac{1}{2}}\leq \sqrt{r_jc_j}\|X_{\cR_j}\|.
    \end{equation}
    Meanwhile, applying Talagrand's contraction inequality (see e.g., Theorem 5.2.16 in \cite{vershynin2018high}) shows us that
    \begin{equation}\label{eq:RSC_term2_tail1}
        \bbP\left(\|X_{\cR_j}^\top E_{\cC_j}\|_2>C\|X_{\cR_j}\|(\sqrt{r_jc_j}+t_j)|X\right)\leq \exp\{-ct_j^2\}.
    \end{equation}
    To obtain a tail probability bound without conditioning on $X$, we apply bounds on spectral norm of matrices with independent isotropic sub-Gaussian rows to $X_{\cR_j}\Sigma_{\cR_j,\cR_j}^{-\frac{1}{2}}$, then for any $t_j>0$,
    \begin{equation}\label{eq:RSC_term2_tail2}
        \bbP\left(\|X_{\cR_j}\|> \lambda^{\frac{1}{2}}_{\max}(\Sigma)(C\sqrt{n}+t_j)\right)\leq 2\exp\{-ct_j^2\}.
    \end{equation}
    Combining \eqref{eq:RSC_term2_tail1} and \eqref{eq:RSC_term2_tail2}, one can show that
    \begin{equation}
    \begin{split}
        &\bbP\left(\|X_{\cR_j}^\top E_{\cC_j}\|_2>C\lambda^{\frac{1}{2}}_{\max}(\Sigma)(\sqrt{r_jc_jn}+(\sqrt{n}+\sqrt{r_jc_j})t_j+t_j^2)\right)\\
        \leq &3\exp\{-nt_j^2\},
    \end{split}
    \end{equation}
    which is equivalent to
    \begin{equation}
    \begin{split}
        &\bbP\left(\|X_{\cR_j}^\top E_{\cC_j}\|_2>C\lambda^{\frac{1}{2}}_{\max}(\Sigma)\sqrt{r_jc_jn}+t_j\right)\\
        \leq &3\exp\left\{-\min\left\{\frac{t_j^2}{\lambda_{\max}(\Sigma)(n+r_jc_j)},\frac{t_j}{\lambda^{\frac{1}{2}}_{\max}(\Sigma)}\right\}\right\}.
    \end{split}
    \end{equation}
    Therefore, we can bound $\bbE\max_{1\leq j\leq J}\frac{1}{\omega_j}\|X_{\cR_j}^\top E_{\cC_j}\|_2$ as follows:
    \begin{equation}
        \begin{split}
            &\bbE\max_{1\leq j\leq J}\frac{1}{\omega_j}\|X_{\cR_j}^\top E_{\cC_j}\|_2\\
            \leq&\int_{0}^{\infty}\bbP\left(\max_{1\leq j\leq J}\frac{1}{\omega_j}\|X_{\cR_j}^\top E_{\cC_j}\|_2>t\right)\mathrm{d}t\\
            \leq&C\lambda^{\frac{1}{2}}_{\max}(\Sigma)\sqrt{mn}+C\lambda^{\frac{1}{2}}_{\max}(\Sigma)\sqrt{\log J(m+n)}\\
            &+\int_{C\lambda^{\frac{1}{2}}_{\max}(\Sigma)\sqrt{\log J(m+n)}}^{\infty}3J\exp\left\{-c\min\left\{\frac{t}{\lambda^{\frac{1}{2}}_{\max}(\Sigma)},\frac{t^2}{\lambda_{\max}(\Sigma)(m+n)}\right\}\right\}\mathrm{d}t\\
            \leq&C\lambda^{\frac{1}{2}}_{\max}(\Sigma)\sqrt{mn}+C\lambda^{\frac{1}{2}}_{\max}(\Sigma)\sqrt{\log J(m+n)}\\
            &+\int_{0}^{\infty}3\exp\left\{-c\min\left\{\frac{t}{\lambda^{\frac{1}{2}}_{\max}(\Sigma)},\frac{t^2}{\lambda_{\max}(\Sigma)(m+n)}\right\}\right\}\mathrm{d}t\\
            \leq&C\lambda^{\frac{1}{2}}_{\max}(\Sigma)(\sqrt{mn}+\sqrt{(\log J(m+n)}+\sqrt{m+n})\\
            \leq&C\lambda^{\frac{1}{2}}_{\max}(\Sigma)\sqrt{(m+\log J)n},
        \end{split}
    \end{equation}
    where we have applied Assumption \ref{assump:scaling} on the last line. 
    \end{proof}
\subsection{Supporting Lemmas for the Ordinal-PU Model}\label{sec:support_lemma_ON}
\begin{proof}[Proof of Lemma \ref{lem:ordinal_rsc_term1}]
Let $u_i=u(x_i,\theta)$, $u_i^*=u(x_i,\theta^*)$. First note that $\exists t\in [0,1]$ such that 
\begin{equation*}
\begin{split}
    &\mathrm{\RNum{1}}(\Delta)\\
    =&\frac{1}{n}\sum_{i=1}^n \begin{pmatrix}x_i^\top \Delta_{1:p}-\Delta_{p+1}\\
    -\Delta_{p+2}\\
    \vdots\\
    -\Delta_{p+K}\end{pmatrix}^\top\nabla f^{\ON}(u_i^t)^\top\nabla^2 A(f^{\ON}(u_i^t))\nabla f^{\ON}(u_i) \begin{pmatrix}x_i^\top \Delta_{1:p}-\Delta_{p+1}\\
    -\Delta_{p+2}\\
    \vdots\\
    -\Delta_{p+K}\end{pmatrix},
\end{split}
\end{equation*}
where $u_i^t=u_i^*+t(u_i-u_i^*)$. Then some calculations show that
\begin{equation*}
    \begin{split}
        \mathrm{\RNum{1}}(\Delta)\geq&\frac{1}{n}\sum_{i=1}^n\Bigg[\lambda_{\min}\left(\nabla f^{\ON}(u_i^t)^\top \nabla^2 A(f^{\ON}(u_i^t))\nabla f^{\ON}(u_i^t)\right)\\
        &-\left\|\nabla f^{\ON}(u_i^t)^\top\nabla^2 A(f^{\ON}(u_i^t)) (\nabla f^{\ON}(u_i^t)-\nabla f^{\ON}(u_i))\right\|\Bigg]\\
        &\left((x_i^\top\Delta_{1:p}-\Delta_{p+1})^2+\sum_{j=2}^{K}\Delta_{p+j}^2\right).
    \end{split}
\end{equation*}
To further lower bound the term above, two key steps here are to (i) lower bound the minimum eigenvalue of $\nabla f^{\ON}(u_i^t)^\top \nabla f^{\ON}(u_i^t)$ and to (ii) upper bound $\|\nabla f^{\ON}(u_i^t)-\nabla f^{\ON}(u_i)\|$. 

The second key step is why we need to constrain the distance between $\theta$ and $\theta^*$ in Lemma~\ref{lem:ordinal_rsc}: by focusing on $\theta\in\theta^*+S(R_0,r_0)$ for some $R_0, r_0>0$, we can upper bound $\|\nabla f^{\ON}(u_i^t)-\nabla f^{\ON}(u_i)\|$ appropriately so that we can show that $\mathrm{\RNum{1}}(\Delta)\geq\frac{\gamma}{n}\sum_{i=1}^n\left((x_i^\top\Delta_{1:p}-\Delta_{p+1})^2+\sum_{j=2}^{K}\Delta_{p+j}^2\right)$ for some $\gamma>0$. 
The proof of Lemma \ref{lem:ordinal_rsc_term1_1} and Lemma \ref{lem:ordinal_rsc_term1_2} are included in Section~\ref{sec:ordinal_rsc_proof_lemmas}.
As shown in the proof of Lemma \ref{lem:mult_rsc_term1}, for any vector $\eta\in \bbR^K$,
\begin{equation*}
    \lambda_{\min}(\nabla^2 A(\eta))\geq \frac{\min_j e^{\eta_j}}{\left(1+\sum_{l=1}^Ke^{\eta_l}\right)^2}.
\end{equation*}
Recall that $\Delta\in S(R_0,r_0)$, which implies that $\max\{\|\Delta_{1:p}\|_1,\|\Delta_{(p+1):(p+K)}\|_1\}\leq R_0$, and $-\min_{2\leq j\leq K}\Delta_{p+j}\leq r_0<r^*$. Since $\exp\{f_j(u_i^t)\}=\frac{n_j}{\pi_jn_u}p_j(u_i^t)$, where $p_j$ is as defined in \eqref{eq:ordinal_p_def}, then by Lemma \ref{lem:ordinal_p_alpha_f_bnd}, one can show that
\begin{equation*}
    \min_j\exp\{f_j(u_i^t)\}=\min_j\frac{n_j}{\pi_jn_u}p_j(u_i^t)\geq \min_j\frac{n_j}{\pi_jn_u}(h^{\ON}(R_0,r_0))^{-1}.
\end{equation*}
Thus we have
\begin{equation*}
    \lambda_{\min}(\nabla^2 A(f^{\ON}(u_i^t)))\geq \min_j\frac{n_j}{4\pi_jn_uh^{\ON}(R_0,r_0)}.
\end{equation*}

Now we recall Lemma \ref{lem:ordinal_rsc_term1_1} and Lemma \ref{lem:ordinal_rsc_term1_2_brief} in the main paper.
Here we give a more slight modification of Lemma \ref{lem:ordinal_rsc_term1_2_brief} that better suits our purposes:
\begin{lemma}\label{lem:ordinal_rsc_term1_2}
 For any $\Delta\in S(R_0,r_0)$,
 $$
 \|\nabla f^{\ON}(u_i^t)-\nabla f^{\ON}(u_i)\|\leq \sqrt{5}[(h^{\ON}(R_0,r_0))^2+h^{\ON}(R_0,r_0)](C_x+1)R_0K,
 $$
 where $h^{\ON}(\cdot)$ is as defined in \eqref{eq:ordinal_h_def}.
\end{lemma}

By Lemma~\ref{lem:ordinal_rsc_term1_1},
\begin{equation}
\begin{split}
    &\lambda_{\min}(\nabla f^{\ON}(u_i^t)^\top \nabla^2 A(f^{\ON}(u_i^t))\nabla f^{\ON}(u_i^t))\\
    \geq&\lambda_{\min}(\nabla f^{\ON}(u_i^t)^\top\nabla f^{\ON}(u_i^t))\lambda_{\min}(\nabla^2 A(f^{\ON}(u_i^t)))\\
    \geq &\min_j\frac{n_j}{\pi_jn_u}\frac{e^{2(C_x+1)(R^*+R_0)}}{16Kh^{\ON}(R_0,r_0)(1+e^{(C_x+1)(R^*+R_0)})^4}.
\end{split}
\end{equation}

Now we upper bound $\left\|\nabla f^{\ON}(u_i^t)\nabla^2 A(f^{\ON}(u_i^t)) (\nabla f^{\ON}(u_i^t)-\nabla f^{\ON}(u_i))\right\|$. By \eqref{eq:ordinal_RSC_f_grad_lwrbnd},
\begin{equation*}
    \begin{split}
        \|\nabla f^{\ON}(u_i^t)v\|_2^2=&\sum_{j=1}^{K-1}\frac{(\alpha_j(u_i^t)\widetilde{v}_j-\alpha_{j+1}(u_i^t)\widetilde{v}_{j+1})^2}{p_j^2(u_i^t)}+\frac{\alpha_K^2(u_i^t)\widetilde{v}_K^2}{p_K^2(u_i^t)}\\
        \leq&\sum_{j=1}^{K-1}\frac{2\alpha_j^2(u_i^t)\widetilde{v}^2_j+2\alpha_{j+1}^2(u_i^t)\widetilde{v}^2_{j+1}}{p_j^2(u_i^t)}+\frac{\alpha_K^2(u_i^t)\widetilde{v}_K^2}{p_K^2(u_i^t)}\\
        \leq&\frac{2(1+e^{(C_x+1)(R^*+R_0)})^4}{(r^*-r_0)^2e^{2(C_x+1)(R^*+R_0)}}\sum_{j=1}^{K-1}(\widetilde{v}^2_j+\widetilde{v}^2_{j+1})\\
        &+(1+e^{(C_x+1)(R^*+R_0)})^2\widetilde{v}_K^2\\
        \leq&(h^{\ON}(R_0,r_0))^2\left(2\sum_{j=1}^{K-1}(\widetilde{v}_j^2+\widetilde{v}_{j+1}^2)+\widetilde{v}_K^2\right)\\
        \leq&(h^{\ON}(R_0,r_0))^2\left(2\sum_{j=1}^{K-1}(j\sum_{l=1}^jv_l^2+(j+1)\sum_{l=1}^{j+1}v_{l}^2)+K\sum_{l=1}^Kv_l^2\right)\\
        \leq&(h^{\ON}(R_0,r_0))^2\sum_{l=1}^K\left(2l+2\sum_{j=l}^{K-1}(2j+1)+K\right)v_l^2\\
        \leq&3(h^{\ON}(R_0,r_0))^2K^2\|v\|_2^2,
    \end{split}
\end{equation*}
which implies that
\begin{equation*}
    \|\nabla f^{\ON}(u_i^t)\|\leq \sqrt{3}h^{\ON}(R_0,r_0)K.
\end{equation*}
Furthermore, by Lemma~\ref{lem:ordinal_rsc_term1_2},
\begin{equation*}
\begin{split}
    \|\nabla f^{\ON}(u_i^t)-\nabla f^{\ON}(u_i)\|\leq \sqrt{5}((h^{\ON}(R_0,r_0))^2+h^{\ON}(R_0,r_0))(C_x+1)R_0K.
    \end{split}
\end{equation*}
While for $\|\nabla^2 A(f^{\ON}(u_i^t))\|$, note that for any $v\in \bbR^K$,
\begin{equation*}
    \begin{split}
        v^\top \nabla^2 A(f^{\ON}(u_i^t))v=&\sum_{j=1}^K(\nabla A(f^{\ON}(u_i^t)))_jv_j^2-\left(\sum_{j=1}^K(\nabla A(f^{\ON}(u_i^t)))_jv_j\right)^2\\
        \leq&\sum_{j=1}^K(\nabla A(f^{\ON}(u_i^t)))_jv_j^2\\
        \leq&\max_j\frac{e^{f^{\ON}(u_{ij}^t)}}{1+\sum_{k=1}^Ke^{f^{\ON}(u_{ik}^t)}}\|v\|_2^2\\
        \leq&\|v\|_2^2,
    \end{split}
\end{equation*}
and hence $\|\nabla^2 A(f^{\ON}(u_i^t))\|=\lambda_{\max}(\nabla^2 A(f^{\ON}(u_i^t)))\leq 1$.
Therefore, 
\begin{equation*}
\begin{split}
    &\|\nabla f^{\ON}(u_i^t)^\top\nabla^2A(f^{\ON}(u_i^t))(\nabla f^{\ON}(u_i^t)-\nabla f^{\ON}(u_i))\|\\
    \leq&4((h^{\ON}(R_0,r_0))^3+(h^{\ON}(R_0,r_0))^2)(C_x+1)R_0K^2.
\end{split}
\end{equation*}
By Lemma \ref{lem:ordinal_region_consequence}, we know that
\begin{equation*}
    R_0\leq \min_j\frac{n_j}{\pi_jn_u}\frac{\left(1+e^{(C_x+1)(R^*+R_0)}\right)^{-2}}{512(C_x+1)K^3(h^{\ON}(R_0,r_0))^3(1+h^{\ON}(R_0,r_0))},
\end{equation*}
which then implies
\begin{equation*}
    \begin{split}
        &\|\nabla f^{\ON}(u_i^t)^\top\nabla^2A(f^{\ON}(u_i^t))(\nabla f^{\ON}(u_i^t)-\nabla f^{\ON}(u_i))\|\\
    \leq&\min_j\frac{n_j}{\pi_jn_u}\frac{\left(1+e^{(C_x+1)(R^*+R_0)}\right)^{-2}}{128Kh^{\ON}(R_0,r_0)}\\
    \leq & \min_j\frac{n_j}{\pi_jn_u}\frac{e^{2(C_x+1)(R^*+R_0)}}{32Kh^{\ON}(R_0,r_0)\left(1+e^{(C_x+1)(R^*+R_0)}\right)^{4}},
    \end{split}
\end{equation*}
where the last line is due to the fact that $e^{(C_x+1)(R^*+R_0)}\geq 1$.
Therefore,
\begin{equation*}
\begin{split}
    &\lambda_{\min}(\nabla f^{\ON}(u_i^t)^\top \nabla^2A(f^{\ON}(u_i^t))\nabla f^{\ON}(u_i^t))\\
    &-\|\nabla f^{\ON}(u_i^t)^\top \nabla^2A(f^{\ON}(u_i^t))(\nabla f^{\ON}(u_i^t)-\nabla f^{\ON}(u_i))\|\\
    \geq&\min_j\frac{n_j}{\pi_jn_u}\frac{e^{2(C_x+1)(R^*+R_0)}}{32Kh^{\ON}(R_0,r_0)(1+e^{(C_x+1)(R^*+R_0)})^{4}},
\end{split}
\end{equation*}
which finishes the proof.
\end{proof}
\begin{proof}[Proof of Lemma \ref{lem:ordinal_rec_dev}]
 By assumption \ref{assump:subgauss_row}, $\{(x_i^\top \widetilde{\Delta}_{1:p}\}_{i=1}^n$ are independent sub-Gaussian random variables with parameter $\sigma\|\widetilde{\Delta}_{1:p}\|_2$, which implies that 
 \begin{equation*}
     \|(x_i^\top \widetilde{\Delta}_{1:p})^2\|_{\psi_1}\leq C\sigma^2\|\widetilde{\Delta}_{1:p}\|_2^2\leq C\sigma^2,
 \end{equation*} 
 and 
 \begin{equation*}
     \|x_i^\top \widetilde{\Delta}_{1:p}|\widetilde{\Delta}_{p+1}|\|_{\psi_2}\leq \sigma|\widetilde{\Delta}_{p+1}|\|\widetilde{\Delta}_{1:p}\|_2\leq \frac{\sigma}{2}.
 \end{equation*}
 Meanwhile, note that
 \begin{equation*}
     g(\widetilde{\Delta},\widetilde{\Delta})=\frac{1}{n}\sum_{i=1}^n(x_i^\top\widetilde{\Delta}_{1:p})^2-\widetilde{\Delta}_{1:p}^\top \Sigma\widetilde{\Delta}_{1:p}-2\widetilde{\Delta}_{p+1}\widetilde{\Delta}_{1:p}^\top x_i
 \end{equation*}
 Thus 
 \begin{equation}
     \begin{split}
         \bbP\left(|g(\widetilde{\Delta},\widetilde{\Delta})|\geq t\right)\leq&\bbP\left(\left|\frac{1}{n}\sum_{i=1}^n(x_i^\top\widetilde{\Delta}_{1:p})^2-\widetilde{\Delta}_{1:p}^\top \Sigma\widetilde{\Delta}_{1:p}\right|\geq \frac{t}{2}\right)\\
         &+\bbP\left(\left|\widetilde{\Delta}_{p+1}\widetilde{\Delta}_{1:p}^\top x_i\right|\geq \frac{t}{4}\right)\\
         \leq&\exp\left\{-cn\min\left\{\frac{t^2}{\sigma^4},\frac{t}{\sigma^2}\right\}\right\}+\exp\left\{-cn\frac{t^2}{\sigma^2}\right\}\\
         \leq&2\exp\left\{-cn\min\left\{\frac{t^2}{\sigma^4},\frac{t^2}{\sigma^2},\frac{t}{\sigma^2}\right\}\right\}.
    \end{split}
 \end{equation}  
\end{proof}
\begin{proof}[Proof of Lemma~\ref{lem:ordinal_rsc_term2_lip}]
 For any $1\leq l\leq K$, some calculation shows that
 \begin{equation*}
 \begin{split}
     \left|(\nabla q_i(u_i(\Delta)))_l\right|=&\bigg|\epsilon_i^\top\left(\nabla f^{\ON}(u_i(\theta^*)+u_i(\Delta))-\nabla f^{\ON}(u_i(\theta^*))\right)_{:,l}\\
     &+\epsilon_i^\top (\nabla^2f^{\ON}(u_i(\theta^*)+u_i(\Delta)))_{:,:,l}u_i(\Delta)\bigg|\\
     \leq&2\left\|\nabla f^{\ON}(u_i(\theta^*)+u_i(\Delta))-\nabla f^{\ON}(u_i(\theta^*))\right\|_{\infty}\\
     &+2\|u_i(\Delta)\|_1\left\|\nabla^2f^{\ON}(u_i(\theta^*)+u_i(\Delta))\right\|_{\infty}.
 \end{split}
 \end{equation*}
 By Lemma~\ref{lem:ordinal_p_alpha_f_bnd}, we have
 \begin{equation*}
     \left|(\nabla q_i(u_i(\Delta)))_l\right|\leq h^{\ON}(R_0,r_0)+2\|u_i(\Delta)\|_1\left[\frac{2}{3}h^{\ON}(R_0,r_0)+\frac{1}{8}(h^{\ON}(R_0,r_0))^2\right].
 \end{equation*}
 Meanwhile, since
 $\|u_i(\Delta)\|_1\leq C_x\|\Delta_{1:p}\|_1+\|\Delta_{(p+1):(p+K)}\|_1\leq (C_x+1)R_0$,
 we obtain that
 \begin{equation*}
 \begin{split}
      \|\nabla q_i(u_i(\Delta))\|_2\leq &\sqrt{K}h^{\ON}(R_0,r_0)+2\sqrt{K}(C_x+1)R_0\left[\frac{2}{3}h^{\ON}(R_0,r_0)+\frac{1}{8}(h^{\ON}(R_0,r_0))^{2}\right]\\
     =&(\frac{4}{3}(C_x+1)R_0+1)\sqrt{K}h^{\ON}(R_0,r_0)+\frac{(C_x+1)R_0}{4}\sqrt{K}(h^{\ON}(R_0,r_0))^{2}.
 \end{split}
 \end{equation*}
\end{proof}
\begin{proof}[Proof of Lemma \ref{lem:ordinal_region_consequence}]
    Firstly, we note that Assumption \ref{assump:ordinal_region_bound} suggests that
    \begin{equation*}
    \begin{split}
        R_0&\leq 80\frac{\left(1+e^{(C_x+1)(R^*+0.01)}\right)^{-6}}{512(C_x+1)K^3(1+\frac{2}{r^* - r_0})^3(2+\frac{2}{r^*-r_0})}\\
        &\leq \frac{80}{1024K^3}\\
        &< 0.01,
    \end{split}
    \end{equation*}
    where we have applied $K\geq 2$ in the last line. Hence Assumption \ref{assump:ordinal_region_bound} also implies
    \begin{equation*}
        \begin{split}
            R_0\leq \min_j\frac{n_j}{\pi_jn_u}\frac{\left(1+e^{(C_x+1)(R^*+R_0)}\right)^{-6}}{512(C_x+1)K^3(1+\frac{2}{r^* - r_0})^3(2+\frac{2}{r^*-r_0})}.
        \end{split}
    \end{equation*}
    On the other hand, we know that $h^{\ON}(R_0,r_0) \leq (1+e^{(C_x+1)(R^*+R_0)})(1+\frac{2}{r^*-r_0})$. This bound then implies the second inequality in Lemma \ref{lem:ordinal_region_consequence}.
\end{proof}
\subsection{Proof of Properties of $f^{\ON}$ for the Ordinal-PU model}\label{sec:ordinal_rsc_proof_lemmas}
In this section, we prove Lemma~\ref{lem:ordinal_rsc_term1_1} and Lemma~\ref{lem:ordinal_rsc_term1_2}, which tackles some of the major challenges brought by the complicated form of ordinal log-likelihood losses.
Some calculations show that, for any vector $u\in\bbR^K$, $\nabla f^{\ON}(u)\in \bbR^{K\times K}$ satisfies
 \begin{equation*}
     \nabla f^{\ON}(u)_{jk}=\frac{\alpha_j(u)\ind{k\leq j}-\alpha_{j+1}(u)\ind{k\leq j+1}}{p_j(u)}.
 \end{equation*}
where functions $p:\bbR^K\rightarrow \bbR^K$ and $\alpha:\bbR^K\rightarrow \bbR^{K+1}$ are defined as follows:
\begin{equation}\label{eq:ordinal_alpha_def}
    \alpha_j(u)=\begin{cases}e^{\sum_{l=1}^ju_{l}}\left(1+e^{\sum_{l=1}^{j}u_{l}}\right)^{-2},&j\leq K,\\
    0,&j=K+1,
    \end{cases}
\end{equation}
and
\begin{equation}\label{eq:ordinal_p_def}
    p_j(u)=\begin{cases}
    (1+e^{\sum_{l=1}^{j+1}u_{l}})^{-1}-(1+e^{\sum_{l=1}^ju_{l}})^{-1},&1\leq j<K,\\
    1-(1+e^{\sum_{l=1}^ju_{l}})^{-1},&j=K,
    \end{cases}
\end{equation}
where $\alpha_j(u),p_j(u)$ refers to the $j$th coordinate of $\alpha(u)$ and $p(u)$. 
\begin{proof}[Proof of Lemma~\ref{lem:ordinal_rsc_term1_1}]
By the definition of function $f^{\ON}$ in \eqref{eq:ordinal_f_def}, one can show that for any $1\leq j,k\leq K$,
\begin{equation}\label{eq:ordinal_f_grad}
    (\nabla f^{\ON}(u_i^*))_{jk}=p_j(u_i^*)^{-1}\left[\alpha_j(u_i^*)\ind{k\leq j}-\alpha_{j+1}(u_i^*)\ind{k\leq j+1}\right].
\end{equation}
For any vector $v\in \bbR^K$, we have
\begin{equation*}
    \begin{split}
        \nabla f^{\ON}(u_i^t)v=&\begin{pmatrix}
        \frac{\alpha_1(u_i^t)}{p_1(u_i^t)}&-\frac{\alpha_2(u_i^t)}{p_1(u_i^t)}&0&\hdots&0\\
        0&\frac{\alpha_2(u_i^t)}{p_2(u_i^t)}&-\frac{\alpha_3(u_i^t)}{p_2(u_i^t)}&\hdots&0\\
        \vdots&\vdots&\ddots&\ddots&0\\
        0&\hdots&\hdots&\frac{\alpha_{K-1}(u_i^t)}{p_{K-1}(u_i^t)}&-\frac{\alpha_K(u_i^t)}{p_{K-1}(u_i^t)}\\
        0&\hdots&\hdots&\hdots&\frac{\alpha_K(u_i^t)}{p_{K}(u_i^t)}
        \end{pmatrix}\begin{pmatrix}v_1\\v_1+v_2\\\sum_{j=1}^3v_j\\\vdots\\\sum_{j=1}^Kv_j\end{pmatrix}\\
        =&\left(\frac{\alpha_1(u_i^t)\widetilde{v}_1-\alpha_2(u_i^t)\widetilde{v}_2}{p_1(u_i^t)},\dots,\frac{\alpha_{K-1}(u_i^t)\widetilde{v}_{K-1}-\alpha_K(u_i^t)\widetilde{v}_K}{p_{K-1}(u_i^t)},\frac{\alpha_K(u_i^t)\widetilde{v}_K}{p_{K}(u_i^t)}\right)^\top,
    \end{split}
\end{equation*}
where $\widetilde{v}_k=\sum_{j=1}^kv_j$. Then one can show that 
\begin{equation}\label{eq:ordinal_RSC_f_grad_lwrbnd}
    \begin{split}
        \left\|\nabla f^{\ON}(u_i^t)v\right\|_2^2=&\sum_{j=1}^{K-1}\frac{(\alpha_j(u_i^t)\widetilde{v}_j-\alpha_{j+1}(u_i^t)\widetilde{v}_{j+1})^2}{p_j^2(u_i^t)}+\frac{\alpha_K^2(u_i^t)\widetilde{v}_K^2}{p_K^2(u_i^t)}\\
        \geq&\frac{1}{K}\sum_{k=1}^{K-1}\sum_{j=k}^{K-1}\frac{(\alpha_j(u_i^t)\widetilde{v}_j-\alpha_{j+1}(u_i^t)\widetilde{v}_{j+1})^2}{p_j^2(u_i^t)}+\frac{\alpha_K^2(u_i^t)\widetilde{v}_K^2}{p_K^2(u_i^t)},
    \end{split}
\end{equation}
since $$\sum_{j=1}^{K-1}\frac{(\alpha_j(u_i^t)\widetilde{v}_j-\alpha_{j+1}(u_i^t)\widetilde{v}_{j+1})^2}{p_j^2(u_i^t)}=\max_{k}\sum_{j=k}^{K-1}\frac{(\alpha_j(u_i^t)\widetilde{v}_j-\alpha_{j+1}(u_i^t)\widetilde{v}_{j+1})^2}{p_j^2(u_i^t)}.$$ 
By the Cauchey-Schwarz inequality, 
\begin{equation*}
    \begin{split}
        &\left(\sum_{j=k}^{K-1}\frac{(\alpha_j(u_i^t)\widetilde{v}_j-\alpha_{j+1}(u_i^t)\widetilde{v}_{j+1})^2}{p_j^2(u_i^t)}+\frac{\alpha_K^2(u^t_i)\widetilde{v}_K^2}{p_K^2(u_i^t)}\right)\left(\sum_{j=k}^{K}p_j^2(u_i^t)\right)\\
        \geq &\left(\sum_{j=k}^{K-1}(\alpha_j(u_i^t)\widetilde{v}_j-\alpha_{j+1}(u_i^t)\widetilde{v}_{j+1})+\alpha_K(u_i^t)\widetilde{v}_K\right)^2 =  \alpha^2_k(u^t_i)\widetilde{v}^2_k,
    \end{split}
\end{equation*}
which implies 
\begin{equation*}
\begin{split}
        \left\|\nabla f^{\ON}(u_i^t)v\right\|_2^2\geq&\frac{1}{K}\Bigg[\frac{\alpha_K^2(u_i^t)\widetilde{v}_K^2}{p_K^2(u_i^t)}+\sum_{k=1}^{K-1}\left(\sum_{j=k}^{K}p_j^2(u_i^t)\right)^{-1}\alpha^2_k(u^t_i)\widetilde{v}^2_k\Bigg]\\
        \geq&\frac{1}{K}\sum_{k=1}^K\alpha^2_k(u^t_i)\widetilde{v}^2_k.
    \end{split}
\end{equation*}
In addition, since
\begin{equation*}
    \begin{split}
        \|v\|_2^2=\sum_{j=1}^{K-1}(\widetilde{v}_{j+1}-\widetilde{v}_j)^2+\widetilde{v}_1^2\leq 4\sum_{j=1}^K\widetilde{v}_j^2\leq \frac{4}{\min_{1\leq k\leq K}\alpha_k^2(u_i^t)}\sum_{j=1}^K\alpha_j^2(u_i^t)\widetilde{v}_j^2,
    \end{split}
\end{equation*}
we have 
$$
\left\|\nabla f^{\ON}(u_i^t)v\right\|_2^2\geq \frac{\min_k\alpha_k^2(u_i^t)}{4K}\|v\|_2^2.
$$
Recall the definition of $\alpha_k(u)$ in \eqref{eq:ordinal_alpha_def} and the fact that
\begin{equation*}
    \left|\sum_{j=1}^ku_{ij}^t\right|\leq (C_x+1)(R^*+R_0),
\end{equation*}
one can show that for $1\leq k\leq K$,
\begin{equation*}
    \alpha_k(u_i^t)\geq \frac{e^{(C_x+1)(R^*+R_0)}}{(1+e^{(C_x+1)(R^*+R_0)})^2}.
\end{equation*}
Hence
\begin{equation*}
    \lambda_{\min}(\nabla f^{\ON}(u_i^t)^\top \nabla f^{\ON}(u_i^t))\geq \frac{e^{2(C_x+1)(R^*+R_0)}}{4K(1+e^{(C_x+1)(R^*+R_0)})^4}.
\end{equation*}
\end{proof}
\begin{proof}[Proof of Lemma~\ref{lem:ordinal_rsc_term1_2}]
While for bounding $\|\nabla f^{\ON}(u_i^t)-\nabla f^{\ON}(u_i)\|$, we first recall each entry of $\nabla f^{\ON}$ in \eqref{eq:ordinal_f_grad}, and note that
\begin{equation}\label{eq:ordinal_rsc_f_grad_diff}
    \begin{split}
        &\|\nabla f^{\ON}(u_i^t)-\nabla f^{\ON}(u_i)\|\\
        \leq &\|\nabla f^{\ON}(u_i^t)-\nabla f^{\ON}(u_i)\|_F\\
        =&\Bigg[\sum_{j=1}^Kj\left(\frac{\alpha_j(u_i^t)}{p_j(u_i^t)}-\frac{\alpha_j(u_i)}{p_j(u_i)}-\frac{\alpha_{j+1}(u_i^t)}{p_j(u_i^t)}+\frac{\alpha_{j+1}(u_i)}{p_j(u_i)}\right)^2\\
        &+\sum_{j=1}^{K-1}\left(\frac{\alpha_{j+1}(u_i^t)}{p_j(u_i^t)}-\frac{\alpha_{j+1}(u_i)}{p_j(u_i)}\right)^2\Bigg]^{\frac{1}{2}}\\
        \leq&\sqrt{2K^2+3K-1}\max_{j,k}\left|\frac{\alpha_k(u_i^t)}{p_j(u_i^t)}-\frac{\alpha_k(u_i)}{p_j(u_i)}\right|.
    \end{split}
\end{equation}
One can show that
\begin{equation}\label{eq:ordinal_rsc_f_grad_diff0}
\begin{split}
   &\max_{j,k}\left|p_j^{-1}(u_i^t)\alpha_k(u_i^t)-p_j^{-1}(u_i)\alpha_k(u_i)\right|\\
   \leq&\max_j\left|p_j^{-1}(u_i^t)-p_j^{-1}(u_i)\right|+\max_{j,k}p_j^{-1}(u_i)\left|\alpha_k(u_i^t)-\alpha_k(u_i)\right|,
\end{split}
\end{equation}
where we have applied the fact that $0<\alpha_k(w)< 1$ for any $w\in \bbR^K$ on the second line. Meanwhile, for any $w\in \bbR^K$, $1\leq j<K$,
\begin{equation*}
    \nabla_{w} (p_j^{-1}(w))=\frac{e^{\sum_{l=1}^{j+1}w_l}}{p_j^{2}(w)(1+e^{\sum_{l=1}^{j+1}w_l})^{2}}(1_{j+1}^\top,0_{K-j-1}^\top)^\top-\frac{e^{\sum_{l=1}^{j}w_l}}{p_j^{2}(w)(1+e^{\sum_{l=1}^{j}w_l})^{2}}(1_{j}^\top,0_{K-j}^\top)^\top,
\end{equation*}

and 
\begin{equation*}
    \nabla_w (p_K^{-1}(w))=-p_K^{-2}(w)\frac{e^{\sum_{l=1}^Kw_l}}{(1+e^{\sum_{l=1}^K w_l})^{2}}1_K
\end{equation*}
Since $0<\frac{e^x}{(1+e^x)^2}<1$ for $x\in \bbR$, we have
\begin{equation*}
    \left\|\nabla_w (p_j^{-1}(w))\right\|_{\infty}\leq p_j^{-2}(w)\leq (h^{\ON}(R,r))^2,
\end{equation*}
for any $w$ lying between $u_i^t$ and $u_i$, which implies that
\begin{equation}\label{eq:ordinal_rsc_f_grad_diff1}
    \max_j|p_j^{-1}(u_i^t)-p_j^{-1}(u_i)|\leq (h^{\ON}(R,r))^2\|u_i^t-u_i\|_1.
\end{equation}
While for bounding $p_j^{-1}(u_i)|\alpha_k(u_i^t)-\alpha_k(u_i)|$, we apply the following lemma which gives several bounds for functions $p(\cdot)$, $\alpha(\cdot)$, $f^{\ON}(\cdot)$:
\begin{lemma}\label{lem:ordinal_p_alpha_f_bnd}
 The functions $p(\cdot),\alpha(\cdot)$ and $f^{\ON}(\cdot)$ defined in \eqref{eq:ordinal_p_def}, \eqref{eq:ordinal_alpha_def} and \eqref{eq:ordinal_f_def} satisfy the following bounds:
 \begin{equation}
     \begin{split}
        0<\alpha_j(u)\leq \frac{1}{4},\quad p_j(u)\geq\min\{\frac{-\max_{k>1}u_ke^{\|u\|_1}}{(1+e^{\|u\|_1})^2},\frac{1}{1+e^{\|u\|_1}}\},
     \end{split}
 \end{equation}
 if $u_j<0$ for $j>1$.
 In particular, if $u=u(x,\theta)$ for any $\|x\|_{\infty}\leq C_x$, and $\theta\in \theta^*+S(R,r)$ for $$S(R,r)=\{\Delta:\|\Delta_{1:p}\|_1,\|\Delta_{(p+1):(p+K)}\|_1\leq R,\min_{2\leq j\leq K}\Delta_{p+j}\geq -r\},$$ then
 \begin{equation}
     \begin{split}
     p_j(u)\geq (h^{\ON}(R,r))^{-1},\quad\|\nabla f^{\ON}(u)\|_{\infty}\leq \frac{1}{4}h^{\ON}(R,r)\\
       \|\nabla^2f^{\ON}(u)\|_{\infty}\leq\frac{2}{3}h^{\ON}(R,r)+\frac{1}{8}(h^{\ON}(R,r))^2.
     \end{split}
 \end{equation}
\end{lemma}
Lemma~\ref{lem:ordinal_p_alpha_f_bnd} is proved subsequently.
Hence by Lemma~\ref{lem:ordinal_p_alpha_f_bnd}, $p_j^{-1}(u_i)\leq h^{\ON}(R_0,r_0)$. In addition, let $g(x)=\frac{e^x}{(1+e^x)^2}$, then $\alpha_k(u_i^t)-\alpha_k(u_i)=g(\sum_{j=1}^ku^t_{ij})-g(\sum_{j=1}^ku_{ij})$. Since
\begin{equation*}
   |g'(x)|=\frac{|e^x-e^{2x}|}{(1+e^x)^3}\leq 1,
\end{equation*}
\begin{equation}\label{eq:ordinal_rsc_f_grad_diff2}
\begin{split}
     &p_j^{-1}(u_i)|\alpha_k(u_i^t)-\alpha_k(u_i)|\\
     \leq &h^{\ON}(R,r)\|u_i^t-u_i\|_1.
\end{split}
\end{equation}
Since $\|u_i^t-u_i\|_1\leq \|u_i^*-u_i\|_1\leq (C_x+1)R$, combining \eqref{eq:ordinal_rsc_f_grad_diff}, \eqref{eq:ordinal_rsc_f_grad_diff0}, \eqref{eq:ordinal_rsc_f_grad_diff1} and \eqref{eq:ordinal_rsc_f_grad_diff2} leads us to 
\begin{equation*}
\begin{split}
    \|\nabla f^{\ON}(u_i^t)-\nabla f^{\ON}(u_i)\|\leq \sqrt{5}((h^{\ON}(R,r))^2+h^{\ON}(R,r))(C_x+1)RK,
    \end{split}
\end{equation*}
\end{proof}
\begin{proof}[Proof of Lemma~\ref{lem:ordinal_p_alpha_f_bnd}]
Let $g(x)=\frac{x}{(1+x)^2}$ then one can show that $g(x)\leq g(1)=\frac{1}{4}$ for any $x\geq 0$, which directly implies that $0<\alpha_j(u)\leq\frac{1}{4}$. 
While for $p_j(u)$, by \eqref{eq:ordinal_p_def} one can show that
\begin{equation*}
    p_K(u)=1-(1+e^{\sum_{l=1}^Ku_l})^{-1}=(1+e^{-\sum_{l=1}^Ku_l})^{-1}\geq\frac{1}{1+e^{\|u\|_1}}.
\end{equation*}
For $1\leq j<K$, $p_j(u)=-u_{j+1}g(\xi)$ for some $\xi\in [e^{\sum_{l=1}^{j+1}u_l},e^{\sum_{l=1}^ju_l}]$. Noting the fact that $g(x)$ is increasing when $x<1$ and decreasing when $x>1$, and $g(\frac{1}{x})=g(x)$, we know that $$g(\xi)\geq \min\{g(e^{\sum_{l=1}^{j+1}u_l}), g(e^{\sum_{l=1}^ju_l})\}\geq \frac{e^{\|u\|_1}}{(1+e^{\|u\|_1})^2},$$
which implies $p_j(u)\geq \frac{-\max_{k>1}u_ke^{\|u\|_1}}{(1+e^{\|u\|_1})^2}$. In addition, recall that we have defined $u$ as $u=u(x,\theta)$, and we assumed $\|x\|_{\infty}\leq C_x$, $\theta\in S(R,r)$. Thus one has
\begin{equation*}
    \|u(x,\theta)\|_1\leq C_x\|\theta_{1:p}\|_1+\|\theta_{(p+1):(p+K)}\|_1\leq (C_x+1)(R^*+R),
\end{equation*}
and for $k>1$,
\begin{equation*}
    u_k(x,\theta)\leq -\min_{k>1}\theta_{p+k}\leq r-r^*.
\end{equation*}
Thus by the definition of $h^{\ON}(\cdot,\cdot)$ in \eqref{eq:ordinal_h_def},
\begin{equation*}
    p_j(u)\geq (h^{\ON}(R,r))^{-1} 
\end{equation*}
for $1\leq j\leq K$. 

Now we start bounding $\|\nabla f^{\ON}(u)\|_{\infty}$ and $\|\nabla^2f^{\ON}(u)\|_{\infty}$. Some calculation shows that, for any vector $u\in\bbR^K$, $\nabla f^{\ON}(u)\in \bbR^{K\times K}$ satisfies
 \begin{equation*}
     \nabla f^{\ON}(u)_{jk}=\frac{\alpha_j(u)\ind{k\leq j}-\alpha_{j+1}(u)\ind{k\leq j+1}}{p_j(u)}.
 \end{equation*}
 Since $0<\alpha_j(u)\leq \frac{1}{4}$, $p_j(u)\geq (h^{\ON}(R,r))^{-1}$, we have $\|\nabla f^{\ON}(u)\|_{\infty}\leq \frac{1}{4}h^{\ON}(R,r)$. Meanwhile, one can show that
 \begin{equation*}
 \begin{split}
     (\nabla^2 f^{\ON}(u))_{jkl}=&\left(\frac{(\nabla \alpha_j(u))_l}{p_j(u)}-\frac{\alpha_j(u)(\nabla p_j(u))_l}{p_j^2(u)}\right)\ind{k\leq j}\\
     &-\left(\frac{(\nabla \alpha_{j+1}(u))_l}{p_j(u)}-\frac{\nabla p_j(u))_l\alpha_{j+1}(u)}{p_j^2(u)}\right)\ind{k\leq j+1},
 \end{split}
 \end{equation*}
 and $|(\nabla p_j(u))_l|=|\alpha_j(u)\ind{l\leq j}-\alpha_{j+1}(u)\ind{l\leq j+1}|\leq \frac{1}{4}$, 
 \begin{equation*}
 \begin{split}
     |(\nabla \alpha_j(u))_l|=\frac{\left|e^{\sum_{m=1}^ju_j}-e^{2\sum_{m=1}^ju_j}\right|}{(1+e^{\sum_{m=1}^ju_j})^3}\ind{l\leq j}\leq\frac{1}{3}.
 \end{split}
 \end{equation*}
 Thus we have
 \begin{equation*}
     \|\nabla^2f^{\ON}(u)\|_{\infty}\leq \max_j\left(\frac{2}{3}h^{\ON}(R,r)+\frac{1}{8}(h^{\ON}(R,r))^2\right).
 \end{equation*}
\end{proof}

\section{Regularized EM algorithms, Additional Convergence Guarantees, and Proofs of Convergence}\label{append:converg}
\paragraph{EM algorithm for Multinomial-PU model:} An alternative estimation method is the EM algorithm, which has been widely adopted in the literature on missing/hidden variable models. Here we present a regularized EM algorithm for our model. 
Based on the full log-likelihood function $$\log L_f^{\MN}(\Theta,b;\{(x_i,y_i,z_i,s_i=1)\}_{i=1}^n)$$ presented in Lemma~\ref{alg:MN-EM}, we summarize our regularized EM algorithm for the multinomial-PU model in Algorithm~\ref{alg:MN-EM}. In the E-step of Algorithm~\ref{alg:MN-EM}, we estimate the unobserved $\ind{y_i=j}$ by $$\bbP_{\Theta^m,b^m}(y_i=j|x_i,z_i,s_i=1)$$ for $j=0,\dots,K$ at the $m$th iteration. Since $z_i=k>0$ implies $y_i=k$, one can show that $$\bbP_{\Theta^m,b^m}(y_i=j|x_i,z_i=k,s_i=1)=\ind{j=k}$$ if $k>0$. On the other hand, when $z_i=0$,
\begin{equation*}
\begin{split}
    &\bbP_{\Theta^m,b^m}(y_i=j|x_i,z_i=0,s_i=1)\\
    =&\frac{\bbP_{\Theta^m,b^m}(y_i=j,z_i=0|x_i,s_i=1)}{\sum_{k=1}^K\bbP_{\Theta^m,b^m}(y_i=k,z_i=0|x_i,s_i=1)+\bbP_{\Theta^m,b^m}(y_i=0,z_i=0|x_i,s_i=1)}\\
    =&\frac{e^{x_i^{\top}\Theta_j^m+b_j^m}}{1+\sum\limits_{k=1}^Ke^{x_i^{\top}\Theta_k^m+b_k^m}},
\end{split}
\end{equation*}
where the last line can be derived from the log-likelihood function for the full data presented $\log L_f^{\MN}$ in Lemma~\ref{lem:MN-log-likelihoods}. In the M-step, we minimize the regularized full log-likelihood loss $$-\frac{1}{n}\log L_f^{\MN}(\Theta,b;\{(x_i,\widehat{y}_i(\Theta^m,b^m),z_i,s_i=1)\}_{i=1}^n)+P^{\MN}_{\lambda}(\Theta),$$
where we use $\widehat{y}_{ik}(\Theta^m,b^m)$ as a surrogate for $\ind{y_i=k}, k=1,\dots, K$.

To solve the M-step, one can apply the proximal gradient descent algorithm~\citep{wright2009sparse}. 
For the initialization, we still consider any $\Theta^0$, $b^0$ such that the regularized observed log-likelihood loss $\mathcal{F}_n^{\MN}$ is no larger than any intercept-only model. Without any prior knowledge, one can simply choose the $(0_{p\times K},\arg\min_b\mathcal{F}_n^{\MN}(0_{p\times K},b))$ as the initializer.

\begin{algorithm}[H]\label{alg:MN-EM}
\SetAlgoLined
Input: $\Theta^0$, $b^0$ such that $\mathcal{F}^{\MN}_n(\Theta^0,b^0)\leq\min_b\mathcal{F}_n^{\MN}(0_{p\times K},b)$\\
 \For{$m=0,1,\dots$}{
   E-step: calculate $\widehat{y}(\Theta^m,b^m)\in \bbR^{n\times K}$ as follows, where each entry $\widehat{y}_{ij}$ is an estimate for $\ind{y_i=j}$:
   \[\hat{y}_{ij}(\Theta^m,b^m)= \begin{cases}
   \frac{e^{x_i^{\top}\Theta_j^m+b_j^m}}{1+\sum\limits_{k=1}^Ke^{x_i^{\top}\Theta_k^m+b_k^m}} & \mathrm{if}\ z_i=0\\
   1 & \mathrm{if}\ z_i=j\\
   0 & \mathrm{else}
   \end{cases}\]
   M-step: obtain $\Theta^{m+1}, b^{m+1}$ by
   \begin{equation*}
   \begin{split}
       (\Theta^{m+1},b^{m+1})\in\argmin\limits_{\Theta,b}&\frac{1}{n}\sum_{i=1}^n\Bigg[\log\left(1+\sum\limits_{k=1}^{K}(1+\frac{n_k}{\pi_kn_u})e^{x_i^\top\Theta_k+b_k}\right)\\
   &-\sum_{j=1}^K\hat{y}_{ij}(\Theta^m,b^m)(x_i^\top \Theta_j+b_j)\Bigg]+P_{\lambda}^{\MN}(\Theta)
   \end{split}
   \end{equation*}
 }
 \caption{EM for the Multinomial-PU Model}
\end{algorithm}
\paragraph{EM algorithm for Ordinal-PU model: }We also propose the following regularized EM algorithm for the ordinal-PU model.

\begin{algorithm}[H]\label{alg:ON-EM}
\SetAlgoLined
Input: $\theta^0$ such that $\mathcal{F}^{\ON}_n(\theta^0)\leq\min_{\theta_{1:p}=0_{p\times 1}}\mathcal{F}^{\ON}_n(\theta)$\\
 \For{$m=0,1,\dots$}{
   E-step: calculate $\widehat{y}(\theta^m)\in\bbR^{n\times K}$ as follows, where each entry $[\widehat{y}(\theta^m)]_{ij}$ is an estimate for $\ind{y_i=j}$:\\
   \[\hat{y}_{ij}(\theta^m)= \begin{cases}
   \frac{1}{1+e^{x^\top \theta^m_{1:p}-\sum_{l=1}^{j+1}\theta^m_{p+l}}}-\frac{1}{1+e^{x^\top \theta^m_{1:p}-\sum_{l=1}^{j}\theta^m_{p+l}}},  &\mathrm{if}\ z_i=0,\\
   1,&\mathrm{if}\ z_i=j,\\
   0,& \mathrm{else.}
   \end{cases}\]\\
   M step: obtain $\theta^{m+1}$ by\\
   \begin{equation*}
       \begin{split}
           \theta_j^{m+1}\in\argmin\limits_{\theta}&\frac{1}{n}\sum_{i=1}^n\Bigg[-\sum_{j=1}^K\widehat{y}_{ij}(\theta^m)\log r_j(x_i, \theta)\\
           &+\log\left(1+\sum\limits_{k=1}^{K}\frac{n_k+\pi_kn_u}{\pi_kn_u}r_k(x_i,\theta)\right)\Bigg]+P^{\ON}_{\lambda}(\theta_{1:p}).
       \end{split}
   \end{equation*}
 }
 \caption{EM for the Ordinal-PU Model}
\end{algorithm}
Similarly to Algorithm~\ref{alg:MN-EM}, in the E-step, we still estimate $\ind{y_i=j}$ by $\bbP_{\theta^m}(y_i=j|x_i,z_i,s_i=1)$ for $j=0,\dots,K$ at the $m$th iteration, which satisfies $\bbP_{\theta^m}(y_i=j|x_i,z_i,s_i=1)=\ind{j=z_i}$ if $z_i>0$. While if $z_i=0$, by the full log-likelihood function $\log L_f^{\MN}$ presented in Lemma~\ref{lem:MN-log-likelihoods},
\begin{equation*}
\begin{split}
    &\bbP_{\theta^m}(y_i=j|x_i,z_i=0,s_i=1)\\
    =&\frac{r_j(x_i,\theta^m)}{\sum_{k=1}^Kr_k(x_i,\theta^m)+1}\\
    =&(1+e^{x^\top \theta^m_{1:p}-\sum_{l=1}^{j+1}\theta^m_{p+l}})^{-1}-(1+e^{x^\top \theta^m_{1:p}-\sum_{l=1}^{j}\theta^m_{p+l}})^{-1}.
\end{split}
\end{equation*}
The M-step minimizes $$-\frac{1}{n}\log L_f^{\ON}(\theta;\{(x_i,\widehat{y}_i(\theta^m),z_i,s_i=1)\}_{i=1}^n)+P^{\ON}_{\lambda}(\theta),$$ and can also be solved by the proximal gradient descent algorithm~\citep{wright2009sparse}
\paragraph{Convergence Properties:} In fact, we can show the same convergence properties for both the PGD algorithms and the regularized EM algorithms. They all converge to stationary points of the corresponding penalized log-likelihood losses. 
\paragraph{PGD for solving penalized MLE v.s. EM algorithm: }From the theoretical perspective, the PGD algorithm applied on the penalized log-likelihood losses and the EM algorithms are similar since they have the same convergence guarantees. However, as suggested by our numerical experiments, the PGD algorithms are much faster while enjoying similar statistical errors to the EM algorithms under both models. Hence we recommend the practitioners to use the PGD algorithms.
\begin{proposition}[Convergence of algorithms for the ordinal-PU model]\label{lem:converg-ON}
 If the parameter iterates $\{\theta^m\}_m)$ are generated by the proximal gradient descent algorithm~\eqref{eq:pgd-ON} with proper choices of step sizes or Algorithm~\ref{alg:ON-EM}, they would satisfy the following:
 \begin{itemize}
    \item[(i)] The sequence $\{\theta^m\}_m$ has at least one limit point.
     \item[(ii)]There exists $R^{\ON}>0$ such that all limit points of $\{\theta^m\}_m$ belong to $\Gamma^{\ON}$, the set of first order stationary points of the optimization problem $\min_{\|\theta\|_2\leq R^{\ON}}\mathcal{F}_n^{\ON}(\theta)$.
     \item[(iii)]The sequence of function values $\{\mathcal{F}_n^{\ON}(\theta^m)\}_m$ is non-increasing, and $\mathcal{F}_n^{\ON}(\theta^{m+1})<\mathcal{F}_n^{\ON}(\theta^{m})$ holds if $\theta^m\notin \Gamma^{\ON}$. There exists $\widetilde{\theta}\in\Gamma^{\ON}$ such that $\{\mathcal{F}_n^{\ON}(\theta^m)\}_m$ converges monotonically to $\mathcal{F}_n^{\ON}(\widetilde{\theta})$.
 \end{itemize}
\end{proposition}
\begin{proof}[Proof of Lemma~\ref{lem:converg-MN} and Lemma~\ref{lem:converg-ON}]
Here we first prove that the sequence of function values is non-increasing for both models and algorithms. Specifically, consider the sequence $\{\cF_n^{\MN}(\Theta^m,b^m)\}_m$ where the parameter iterates $\{\Theta^m,b^m\}_m$ are generated by Algorithm~\ref{alg:MN-EM}. For any $m\geq 0$, one can show that
\begin{align*}
    &\cL^{\MN}(\Theta,b;\{(x_i,z_i)\}_{i=1}^n)\\
    =&-\frac{1}{n}\sum_{i=1}^n\log \bbP(z_i|x_i,\Theta,b)\\
    =&-\frac{1}{n}\sum_{i=1}^n\bbE_{y_i|x_i,z_i,\Theta^m,b^m}\left[\log\frac{\bbP(y_i,z_i|x_i,\Theta,b)}{\bbP(z_i|x_i,\Theta,b)}\right]\\
    =&-\frac{1}{n}\sum_{i=1}^n\bbE_{y_i|x_i,z_i,\Theta^m,b^m}\left[\log\frac{\bbP(y_i,z_i|x_i,\Theta,b)}{\bbP(y_i|x_i,z_i,\Theta^m,b^m)}-\log \frac{\bbP(y_i|x_i,z_i,\Theta,b)}{\bbP(y_i|x_i,z_i,\Theta^m,b^m)}\right]\\
    \leq&-\frac{1}{n}\sum_{i=1}^n\bbE_{y_i|x_i,z_i,\Theta^m,b^m}\left[\log\frac{\bbP(y_i,z_i|x_i,\Theta,b)}{\bbP(y_i|x_i,z_i,\Theta^m,b^m)}\right]\\
    =&-\frac{1}{n}\log L_f^{\MN}(\Theta,b;\{(x_i,\widehat{y}_i(\Theta^m,b^m),z_i)\}_{i=1}^n)\\
    &+\frac{1}{n}\sum_{i=1}^n\bbE_{y_i|x_i,z_i,\Theta^m,b^m}\log\bbP(y_i|x_i,z_i,\Theta^m,b^m),
\end{align*}
where the fifth line is due to that
\begin{align*}
    &\bbE_{y_i|x_i,z_i,\Theta^m,b^m}\left[\log \frac{\bbP(y_i|x_i,z_i,\Theta,b)}{\bbP(y_i|x_i,z_i,\Theta^m,b^m)}\right]\\
    \leq&\log \bbE_{y_i|x_i,z_i,\Theta^m,b^m}\left[\frac{\bbP(y_i|x_i,z_i,\Theta,b)}{\bbP(y_i|x_i,z_i,\Theta^m,b^m)}\right]\\
    =&0,
\end{align*}
and the equality holds if and only if $\Theta=\Theta^m$, $b=b^m$.
Here we have omitted data point $(x_i,z_i,s_i=1)$ to $(x_i,z_i)$ for simplicity. Therefore, we have
\begin{align*}
    &\cF_n^{\MN}(\Theta^{m+1},b^{m+1})\\
    \leq &-\frac{1}{n}\log L_f^{\MN}(\Theta^{m+1},b^{m+1};\{(x_i,\widehat{y}_i(\Theta^m,b^m),z_i)\}_{i=1}^n)+P_{\lambda}^{\MN}(\Theta^{m+1})\\
    &+\frac{1}{n}\sum_{i=1}^n\bbE_{y_i|x_i,z_i,\Theta^m,b^m}\left[\log P(y_i|x_i,z_i,\Theta^m,b^m)\right]\\
    \leq &-\frac{1}{n}\log L_f^{\MN}(\Theta^m,b^m;\{(x_i,\widehat{y}_i(\Theta^m,b^m),z_i)\}_{i=1}^n)+P_{\lambda}^{\MN}(\Theta^m)\\
    &+\frac{1}{n}\sum_{i=1}^n\bbE_{y_i|x_i,z_i,\Theta^m,b^m}\left[\log P(y_i|x_i,z_i,\Theta^m,b^m)\right]\\
    =&\cF_n^{\MN}(\Theta^m,b^m),
\end{align*}
where the second inequality is due to the M-step in Algorithm~\ref{alg:MN-EM}. If $\cF_n^{\MN}(\Theta^{m+1},b^{m+1})=\cF_n^{\MN}(\Theta^m,b^m)$, then $(\Theta^{m},b^{m})$ must be a minimizer of $$-\frac{1}{n}\log L_f^{\MN}(\Theta,b;\{(x_i,\widehat{y}_i(\Theta^m,b^m),z_i)\}_{i=1}^n)+P_{\lambda}^{\MN}(\Theta)$$ over $\bbR^{p\times K}\times \bbR^K$. Since 
\begin{align*}
    &\left(\nabla \cL^{\MN}(\Theta,b;\{(x_i,z_i)\}_{i=1}^n)\right)|_{\Theta^m,b^m}\\
    =&\left(-\frac{1}{n}\nabla \log L_f^{\MN}(\Theta,b;\{(x_i,\widehat{y}_i(\Theta^m,b^m),z_i)\}_{i=1}^n)\right)|_{\Theta^m,b^m},
\end{align*} 
$\cF_n^{\MN}(\Theta^{m+1},b^{m+1})=\cF_n^{\MN}(\Theta^m,b^m)$ implies that $(\Theta^m,b^m)\in \Gamma^{\MN}$. Using the same arguments we can show that $\{\cF_n^{\ON}(\theta^m)\}_{m}$ is also non-increasing if $\{\theta^m\}_m$ are generated by Algorithm~\ref{alg:ON-EM}, and $\cF_n^{\ON}(\theta^{m+1})<\cF_n^{\ON}(\theta^{m})$ if $\theta^m\in \Gamma^{\ON}$. On the other hand, if $\{\Theta^m,b^m\}_m$ (resp. $\{\theta^m\}_m$) are generated by the proximal gradient descent algorithm with each update~\eqref{eq:pgd-MN} (resp.~\eqref{eq:pgd-ON}), then existed results~\citep[see Lemma 10.4]{beck2017first} suggest that with appropriately chosen step sizes $\eta_m$, $\{\cF_n^{\MN}(\Theta^m,b^m)\}$ (resp. $\{\cF_n^{\ON}(\theta^m)\}$) is non-increasing. If $\cF_n^{\MN}(\Theta^{m+1},b^{m+1})=\cF_n^{\MN}(\Theta^m,b^m)$, $(\Theta^m,b^m)$ must be a minimizer of 
$$
\eta_mP^{\MN}_{\lambda}(\Theta)+\frac{1}{2}\|[\Theta;b]-([\Theta^m;b^m]-\eta_m\nabla \cL_n^{\MN}(\Theta^m,b^m))\|_2^2.
$$
Since any subgradient of the loss above evaluated at $(\Theta^m,b^m)$ satisfies
\begin{align*}
    &\left(\nabla \eta_mP^{\MN}_{\lambda}(\Theta)+\frac{1}{2}\|[\Theta;b]-([\Theta^m;b^m]-\eta_m\nabla \cL_n^{\MN}(\Theta^m,b^m))\|_2^2\right)|_{\Theta^m,b^m}\\
    =&\eta_m(\nabla P^{\MN}_{\lambda}(\Theta^m)+\nabla \cL_n^{\MN}(\Theta^m,b^m)))\\
    =&\eta_m(\nabla \cF_n^{\MN}(\Theta,b))|_{\Theta^m,b^m},
\end{align*}
we also have $\cF_n^{\MN}(\Theta^{m+1},b^{m+1})<\cF_n^{\MN}(\Theta^m,b^m)$ if $(\Theta^m,b^m)\notin \Gamma^{\MN}$. Similarly we have $\cF_n^{\ON}(\theta^{m+1})<\cF_n^{\ON}(\theta^m)$ if $\theta^m\notin \Gamma^{\ON}$ when $\{\theta^m\}_m$ are generated by the proximal gradient descent algorithm~\eqref{eq:pgd-ON}.

Now we show that $\lim_{\|\Theta\|_2+\|b\|_2\rightarrow +\infty}\cF_n^{\MN}(\Theta,b)=+\infty$ and $\lim_{\|\theta\|_2\rightarrow +\infty}\cF_n^{\ON}(\theta)=+\infty$, and hence $$\cF_n^{\MN}(\Theta,b)\leq \argmin_{b}\cF_n^{\MN}(0_{p\times K},b),\quad \cF_n^{\ON}(\theta)\leq \argmin_{\theta_{1:p}=0_{p\times 1}}\cF_n^{\ON}(\theta)$$ implies $\|\Theta\|_2+\|b\|_2\leq R^{\MN}$ and $\|\theta\|_2\leq R^{\ON}$ for some $R^{\MN}>0$, $R^{\ON}>0$. By the initialization conditions for our proximal gradient descent algorithms and EM algorithms (Algorithm~\ref{alg:MN-EM},~\ref{alg:ON-EM}) and the non-increasing property of $\{\cF_n^{\MN}(\Theta^m,b^m)\}$ and $\{\cF_n^{\ON}(\theta^m)\}$, it is guaranteed that all parameter iterates $\{(\Theta^m,b^m)\}_m$ belong to set $\cC^{\MN}=\{[\Theta;b]:\|\Theta\|_2+\|b\|_2\leq R^{\MN}\}$, and $\{\theta^m\}_m$ belong to $\cC^{\ON}=\{\theta:\|\theta\|_2\leq R^{\ON}\}$. Both $\cC^{\MN}$ and $\cC^{\ON}$ are compact sets, which implies (i) in Lemma~\ref{lem:converg-MN} and Lemma~\ref{lem:converg-ON}. 

Due to the continuity of our loss functions $\cF_n^{\MN}$, $\cF^{\ON}$ and the fact that all parameter iterates belong to compact sets $\cC^{\MN}$, $\cC^{\ON}$, we can apply the global convergence theorem in~\cite{zangwill1969nonlinear} and obtain (ii) and (iii) in Proposition~\ref{lem:converg-MN} and Proposition~\ref{lem:converg-ON}.

\end{proof}
\section{Derivation of log-likelihood functions in the case-control setting}\label{sec:proof-likelihoods}
First we derive the conditional distribution of $z$ given $x, y$ and $s=1$ in the case-control setting, which was presented earlier in Section \ref{sec:problem-form}. Then we will present the proofs of
Lemma~\ref{lem:MN-log-likelihoods} and Lemma~\ref{lem:ON-log-likelihoods}. 

First note that in the case-control setting, $n_u$ samples are randomly drawn from the whole population and $n_k$ samples are drawn from the population with label $k$ for $1\leq k\leq K$. Hence we have
\begin{equation*}
    \frac{\bbP(z=k|s=1)}{\bbP(z=0|s=1)}=\frac{n_k}{n_u},
\end{equation*}
which implies that $$\frac{\bbP(z=k,y=k|s=1)}{\bbP(z=0,y=k|s=1)}=\frac{\bbP(z=k|s=1)}{\bbP(z=0|s=1)\bbP(y=k)}=\frac{n_k}{\pi_kn_u}.$$ Therefore,
\begin{equation*}
    \bbP(z=k|y=k,s=1)=\frac{\bbP(z=k,s=1|y=k)}{\bbP(s=1|y=k)}=\frac{n_k}{n_k+\pi_kn_u}, \quad \bbP(z=0|y=k,s=1) = \frac{\pi_k n_u}{n_k + \pi_k n_u}.
\end{equation*}
\begin{proof}[Proof of Lemma~\ref{lem:MN-log-likelihoods}]
First we derive the conditional probability mass function of any data point $(x,y,z,s=1)$ given $x$ and $s=1$. Note that there are $2K+1$ possible combinations of $(y,z)$: $y=z=0$, $y=z=j$, and $y=j, z=0$ for $1\leq j\leq K$. Then one can show that for any $1\leq j\leq K$,
\begin{equation}
 \begin{split}
        \bbP(y=j,z=j|x,s=1)=&\frac{\bbP(y=j|x)\bbP(z=j,s=1|y=j)}{\sum_{k=1}^K\bbP(y=k|x)\bbP(z=k,s=1|y=k)+\bbP(z=0,s=1|x)}\\
        =&\frac{n_j/\pi_je^{x^\top\Theta_j+b_j}}{n_u+\sum\limits_{k=1}^{K}(n_k/\pi_k+n_u)e^{x^\top\Theta_k+b_k}},
    \end{split}
\end{equation}
where we have applied the fact that $\bbP(y=j|x)=\frac{e^{x^\top\Theta_j+b_j}}{1+\sum\limits_{k=1}^{K}e^{x^\top\Theta_k+b_k}}$ and $$\frac{\bbP(z=j,s=1|y=j)}{\bbP(z=k,s=1|y=k)}=\frac{n_j/\pi_j}{n_k/\pi_k},\quad \frac{\bbP(z=j,s=1|y=j)}{\bbP(z=0,s=1|x)}=\frac{n_j/\pi_j}{n_u}.$$ Applying the above facts also leads us to
\begin{equation*}
    \begin{split}
        \bbP(y=j,z=0|x,s=1)=&\frac{\bbP(y=j|x)\bbP(z=0,s=1|y=j)}{\sum_{k=1}^K\bbP(y=k|x)\bbP(z=k,s=1|y=k)+\bbP(z=0,s=1|x)}\\
        =&\frac{n_ue^{x^\top\Theta_j+b_j}}{n_u+\sum\limits_{k=1}^{K}(n_k/\pi_k+n_u)e^{x^\top\Theta_k+b_k}},\\
        \bbP(y=0,z=0|x,s=1)=&\frac{\bbP(y=0|x)\bbP(z=0,s=1|y=0))}{\sum_{k=1}^K\bbP(y=k|x)\bbP(z=k,s=1|y=k)+\bbP(z=0,s=1|x)}\\
        =&\frac{n_u}{n_u+\sum\limits_{k=1}^{K}(n_k/\pi_k+n_u)e^{x^\top\Theta_k+b_k}}.
    \end{split}
\end{equation*}
Thus the full log-likelihood function of $\{x_i,y_i,z_i,s_i=1\}_{i=1}^n$ is as follows:
\begin{equation}
\begin{split}
    &\log L^{\MN}_f(\Theta,b;\{x_i,y_i,z_i,s_i=1\}_{i=1}^n)\\
    =&\sum_{i=1}^n\left[\sum_{k=1}^K\ind{y_i=k}(x_i^\top \Theta_k+b_k)-\log\left(1+\sum\limits_{k=1}^{K}(1+\frac{n_k}{\pi_kn_u})e^{x_i^\top\Theta_k+b_k}\right)\right]\\
    &+\sum_{i=1}^n\sum_{k=1}^K\ind{y_i=z_i=k}\log \frac{n_k}{\pi_kn_u}.
\end{split}
\end{equation}

To show the log-likelihood function for the observed data $\{x_i,z_i,s_i=1\}_{i=1}^n$, note that the conditional probability mass function for $z$ given $x,s=1$ is as follows:
\begin{equation}
    \begin{split}
        \bbP(z=j|x,s=1)=&\bbP(y=j,z=j|x,s=1)\\
        =&\frac{\frac{n_j}{\pi_jn_u}e^{x^\top\Theta_j+b_j}}{1+\sum\limits_{k=1}^{K}(1+\frac{n_k}{\pi_kn_u})e^{x^\top\Theta_k+b_k}},\quad j>0,\\
        \bbP(z=0|x,s=1)=&\sum_{j=1}^K\bbP(y=j,z=0|x,s=1)+\bbP(y=0,z=0|x,s=1)\\
        =&\frac{1+\sum_{j=1}^Ke^{x^\top\Theta_j+b_j}}{1+\sum\limits_{k=1}^{K}(1+\frac{n_k}{\pi_kn_u})e^{x^\top\Theta_k+b_k}}.
    \end{split}
\end{equation}
Therefore, the log-likelihood function of the observed data $\{x_i,z_i,s_i=1\}_{i=1}^n$ is
\begin{equation}
    \begin{split}
        &\log L^{\MN}(\Theta,b;\{x_i,z_i,s_i=1\}_{i=1}^n)\\
    =&\sum_{i=1}^n\Bigg[\sum_{j=1}^K\ind{z_i=j}(x_i^\top \Theta_j+b_j+\log (\frac{n_j}{\pi_jn_u}))+\ind{z_i=0}\log(1+\sum_{j=1}^Ke^{x_i^\top\theta_j+b_j})\\
    &-\log\left(1+\sum\limits_{k=1}^{K}(1+\frac{n_k}{\pi_kn_u})e^{x_i^\top\Theta_k+b_k}\right)\Bigg]\\
    =&\sum_{i=1}^n\Bigg[\sum_{j=1}^K\ind{z_i=j}(x_i^\top \Theta_j+b_j+\log \frac{n_j}{\pi_j n_u}-\log(1+\sum_{j=1}^Ke^{x_i^\top\Theta_j+b_j}))\\
    &-\log\left(1+\frac{\sum\limits_{k=1}^{K}\frac{n_k}{\pi_kn_u}e^{x_i^\top\Theta_k+b_k}}{1+\sum_{k=1}^K e^{x_i^\top \Theta_k+b_k}}\right)\Bigg].
    \end{split}
\end{equation}

\end{proof}
\begin{proof}[Proof of Lemma~\ref{lem:ON-log-likelihoods}]
Similarly from the previous set-up, we can derive the distribution of $y,z$ given $x$ as
\begin{equation}
    \begin{split}
        \bbP(y=j,z=j|x,s=1)=&\frac{\bbP(y=j,z=j,s=1|x)}{\bbP(s=1|x)}\\
        =&\frac{\frac{n_j}{\pi_jn_u}r_j(x,\theta)}{1+\sum\limits_{k=1}^{K}(\frac{n_k+\pi_kn_u}{\pi_kn_u})r_k(x,\Theta)},\quad j>0,\\
        \bbP(y=j,z=0|x,s=1)=&\frac{r_j(x,\theta)}{1+\sum\limits_{k=1}^{K}(\frac{n_k+\pi_kn_u}{\pi_kn_u})r_k(x,\theta)},\\
        \bbP(y=0,z=0|x,s=1)=&\frac{1}{1+\sum\limits_{k=1}^{K}(\frac{n_k+\pi_kn_u}{\pi_kn_u})r_k(x,\theta))},
    \end{split}
\end{equation}
where
\begin{equation}
\begin{split}
    r_j(x,\theta)=&\frac{\bbP_{\theta}(y=j|x)}{\bbP_{\theta}(y=0|x)}\\
    =&(1+e^{x^\top \theta_{1:p}-\theta_{p+1}})\left[(1+e^{x^\top \theta_{1:p}-\sum_{l=1}^{j+1}\theta_l})^{-1}-(1+e^{x^\top \theta_{1:p}-\sum_{l=1}^j\theta_l})^{-1}\right].
\end{split}
\end{equation}
On the other hand, the distribution of $z$ given $x$ is
\begin{equation}
    \begin{split}
        \bbP(z=j|x,s=1)=&\frac{\frac{n_j}{\pi_jn_u}r_j(x,\theta)}{1+\sum\limits_{k=1}^{K}(\frac{n_k+\pi_kn_u}{\pi_kn_u})r_k(x,\theta)},\quad j>0,\\
        \bbP(z=0|x,s=1)=&\frac{1+\sum_{k=1}^Kr_k(x,\theta)}{1+\sum\limits_{k=1}^{K}(\frac{n_k+\pi_kn_u}{\pi_kn_u})r_k(x,\theta)}.
    \end{split}
\end{equation}
Thus the full log-likelihood function of $\{(x_i,y_i,z_i,s_i=1)\}_{i=1}^n$ is as follows:
\begin{equation}
\begin{split}
    &\log L_f^{\ON}(\theta;\{(x_i,y_i,z_i,s_i=1)\}_{i=1}^n)\\
    =&\sum_{i=1}^n\left[\sum_{j=1}^K\ind{y_i=j}\log r_j(x_i, \theta)-\log\left(1+\sum\limits_{k=1}^{K}\frac{n_k+\pi_kn_u}{\pi_kn_u}r_k(x_i,\theta)\right)\right]\\
    &+\sum_{i=1}^n\sum_{j=1}^K\ind{y_i=z_i=j}\log \frac{n_j}{\pi_jn_u},
\end{split}
\end{equation}
and the log-likelihood of the observed data $\{(x_i,z_i,s_i=1)\}_{i=1}^n$ is
\begin{equation*}
    \begin{split}
        &\log L^{\ON}(\theta;\{(x_i,z_i,s_i=1)\}_{i=1}^n)\\
    =&\sum_{i=1}^n\Bigg[\sum_{j=1}^K\ind{z_i=j}(\log r_j(x_i,\theta)-\log (\frac{n_j}{\pi_jn_u}))+\ind{z_i=0}\log(1+\sum_{j=1}^Kr_j(x_i,\theta))\\
    &-\log\left(1+\sum\limits_{k=1}^{K}\frac{n_k+\pi_kn_u}{\pi_kn_u}r_k(x_i,\theta)\right)\Bigg]\\
    =&\sum_{i=1}^n \left[\sum_{k=1}^K \ind{z_i=k}[f(\log r(x_i,\theta))]_k-\log\left(1+\sum_{j=1}^Ke^{[f(\log r(x_i,\theta))]_k}\right)\right],
 \end{split}
 \end{equation*}
 where $f:\bbR^K\rightarrow \bbR^K$ satisfies $(f^{\ON}(u))_k=u_k+\log \frac{n_k}{\pi_kn_u}-\log(1+\sum_{j=1}^Ke^{u_j})$.%
\end{proof}
\section{Algorithms for the PU models under the single-training-set scenario}\label{sec:alg-st}
Here we consider the single-training-set scenario~\citep{elkan2008learning}, a different setting from the case-control scenario for PU models. In particular, we still assume that the multinomial and ordinal distributions for the true labeled data $(x_i,y_i)$, while letting the observation $z_i$ be unlabeled with some constant probability. Formally, for any $j,k\geq 0$,
\begin{align}\label{eq:random-missing}
    \bbP(z_i=k|x_i,y_i=j)=\begin{cases}\pi^{\st}_j,&\text{if $j=k>0$,}\\
    1-\pi^{\st}_j,& \text{if $j>0$, $k=0$,}\\
    1,&\text{if $j=k=0$,}\\
    0,&\text{otherwise},
    \end{cases}
\end{align}
where $\pi^{\st}\in(0,1)^K$ indicates the probabilities of positive samples being labeled, and the superscript $\st$ refers to the single-training-set scenario.
\subsection{Multinomial-PU Model}
The following lemma presents the log-likelihood functions for the observed data $\{(x_i,z_i)\}$ and for the full data $\{(x_i,y_i,z_i)\}$ under the multinomial-PU model and the single-training-set scenario.
\begin{lemma}\label{lem:MN-log-likelihoods-st}
 When the observed responses $z_i$ is generated according to~\eqref{eq:random-missing}, the log-likelihood function for the observed presence-only data $\{x_i,z_i\}_{i=1}^n$ under the multinomial model is
 \begin{equation}\label{eq:MN-log-likelihood-st}
 \begin{split}
     &\log L^{\MN,\st}(\Theta,b;\{(x_i,z_i)\}_{i=1}^n)\\
     =&\sum_{i=1}^n \left[\sum_{k=1}^K \ind{z_i=k}[\widetilde{f}(\Theta^\top x_i+b+\log (1-\pi^{st}))]_k-\log\left(1+\sum_{k=1}^Ke^{[\widetilde{f}(\Theta^\top x_i+b+\log (1-\pi^{st}))]_k}\right)\right],
 \end{split}
 \end{equation}
 where $\widetilde{f}:\bbR^K\rightarrow \bbR^K$ satisfies $(\widetilde{f}(u))_k=u_k+\log \frac{\pi^{\st}_k}{1-\pi^{\st}_k}-\log(1+\sum_{j=1}^Ke^{u_j})$. Here $\log (1-\pi^{\st})=(\log(1-\pi^{\st}_1),\dots,\log(1-\pi^{\st}_K))^\top\in\bbR^K$.
 The log-likelihood function for the full data $\{x_i,y_i,z_i\}_{i=1}^n$ is
 \begin{equation}\label{eq:MN-full-log-likelihood-st}
 \begin{split}
     &\log L_f^{\MN,\st}(\Theta,b;\{(x_i,y_i,z_i)\}_{i=1}^n)\\
     =&\sum_{i=1}^n
    \bigg[\sum_{k=1}^K \ind{y_i=k}(x_i^\top \Theta_k+b_k+\log(1-\pi^{\st}_k))-\log\left(1+\sum_{k=1}^Ke^{x_i^\top\Theta_k+b_k}\right)\\
    &+\sum_{k=1}^K\ind{y_i=z_i=k}\log\frac{\pi^{\st}_k}{1-\pi^{\st}_k}\bigg],
 \end{split}
 \end{equation}
 where $\Theta_k$ is the $k$th column of $\Theta$.
\end{lemma}
\begin{proof}
Under the multinomial-PU model with single-training-set scenario, one can show that for any data point $(x,y,z)$ and $1\leq j\leq K$,
\begin{align*}
    \bbP(y=j,z=j|x)=&\bbP(y=j|x)\bbP(z=j|y=j)=\frac{\pi^{\st}_je^{x^\top\Theta_j+b_j}}{1+\sum_{k=1}^Ke^{x^\top\Theta_k+b_k}},\\
    \bbP(y=j,z=0|x)=&\bbP(y=j|x)\bbP(z=0|y=j)=\frac{(1-\pi^{\st}_j)e^{x^\top\Theta_j+b_j}}{1+\sum_{k=1}^Ke^{x^\top\Theta_k+b_k}},\\
    \bbP(y=0,z=0|x)=&\bbP(y=0|x)=\frac{1}{1+\sum_{k=1}^Ke^{x^\top\Theta_k+b_k}}.
\end{align*}
While for the observed PU data $(x,z)$, we have
\begin{align*}
    \bbP(z=j|x)=&\bbP(y=j,z=j|y=j)=\frac{\pi^{\st}_je^{x^\top\Theta_j+b_j}}{1+\sum_{k=1}^Ke^{x^\top\Theta_k+b_k}},\\
    \bbP(z=0|x)=&\sum_{j=1}^K\bbP(y=j,z=0|y=j)+\bbP(y=0,z=0|x)=\frac{1+\sum_{j=1}^K(1-\pi^{\st}_j)e^{x^\top\Theta_j+b_j}}{1+\sum_{k=1}^Ke^{x^\top\Theta_k+b_k}}.
\end{align*}
Hence the full log-likelihood for samples $\{(x_i,y_i,z_i)\}_{i=1}^n$ can be written as follows:
\begin{align*}
    &\log L_f^{\MN,\st}(\Theta,b;\{(x_i,y_i,z_i)\}_{i=1}^n)\\
    =&\sum_{i=1}^n\left[\sum_{j=1}^K\ind{y_i=j}(x_i^\top \Theta_j+b_j+\log(1-\pi^{\st}_j))-\log(1+\sum_{k=1}^Ke^{x_i^\top\Theta_k+b_k})\right]\\
    &+\ind{z_i=j}\log\frac{\pi^{\st}_j}{1-\pi^{\st}_j}.
\end{align*}
While for the log-likelihood for the observed PU data $\{(x_i,z_i)\}_{i=1}^n$ is
\begin{align*}
    &\log L^{\MN,\st}(\Theta,b;\{(x_i,z_i)\}_{i=1}^n)\\
    =&\sum_{i=1}^n\Bigg[\sum_{j=1}^K\ind{z_i=j}(x_i^\top\Theta_j+b_j+\log \pi^{\st}_j-\log(1+\sum_{k=1}^K(1-\pi^{\st}_k)e^{x_i^\top\Theta_k+b_k})))\\
    &-\log(1+\frac{\sum_{k=1}^K\pi^{\st}_ke^{x_i^\top\Theta_k+b_k}}{1+\sum_{k=1}^K(1-\pi^{\st}_k)e^{x_i^\top\Theta_k+b_k}}) \Bigg]\\
    =&\sum_{i=1}^n \Bigg[\sum_{k=1}^K \ind{z_i=k}[\widetilde{f}(\Theta^\top x_i+b+\log (1-\pi^{st}))]_k-\log\left(1+\sum_{j=1}^Ke^{[\widetilde{f}(\Theta^\top x_i+b+\log (1-\pi^{st}))]_k}\right)\Bigg],
\end{align*}
\end{proof} 

Similarly to the estimation under the case-control case (discussed in Section~\ref{sec:algorithm}), one can apply the proximal gradient descent algorithm on $$\cF_n^{\MN,\st}(\Theta,b)=-\frac{1}{n}\log L^{\MN,\st}(\Theta,b;\{(x_i,z_i)\}_{i=1}^n)+P_{\lambda}^{\MN}(\Theta)$$ for the estimation of $\Theta$ and $b$. We can also derive the EM algorithm under this setting (see Algorithm~\ref{alg:MN-EM-st}).

\begin{algorithm}[H]\label{alg:MN-EM-st}
\SetAlgoLined
Input: $\Theta^0$, $b^0$ such that $\mathcal{F}^{\MN,\st}_n(\Theta^0,b^0)\leq\min_b\mathcal{F}_n^{\MN,\st}(0_{p\times K},b)$\\
 \For{$m=0,1,\dots$}{
   E-step: calculate $\widehat{y}(\Theta^m,b^m)\in \bbR^{n\times K}$ as follows, where each entry $\widehat{y}_{ij}$ is an estimate for $\ind{y_i=j}$:
   \[\hat{y}_{ij}(\Theta^m,b^m)= \begin{cases}
   \frac{e^{x_i^{\top}\Theta_j^m+b_j^m+\log(1-\pi^{\st}_j)}}{1+\sum\limits_{k=1}^Ke^{x_i^{\top}\Theta_k^m+b_k^m+\log(1-\pi^{\st}_k)}} & \mathrm{if}\ z_i=0\\
   1 & \mathrm{if}\ z_i=j\\
   0 & \mathrm{else}
   \end{cases}\]
   M-step: obtain $\Theta^{m+1}, b^{m+1}$ by
   \begin{equation*}
   \begin{split}
       (\Theta^{m+1},b^{m+1})\in\argmin\limits_{\Theta,b}&\frac{1}{n}\sum_{i=1}^n\Bigg[\log\left(1+\sum\limits_{k=1}^{K}e^{x_i^\top\Theta_k+b_k}\right)\\
   &-\sum_{j=1}^K\hat{y}_{ij}(\Theta^m,b^m)(x_i^\top \Theta_j+b_j)\Bigg]+P_{\lambda}^{\MN}(\Theta)
   \end{split}
   \end{equation*}
 }
 \caption{EM for the Multinomial-PU Model under the single-training-set scenario}
\end{algorithm}

\paragraph{Theoretical properties: }One key observation is that the multinomial-PU model under the single-training-set scenario can be viewed as a simple reparameterization of the case-control setting. More specifically, to obtain the log-likelihood functions under the single-training-set scenario, one can simply substitute $\frac{n_k}{\pi_kn_u}$ in $\log L_f^{\MN}$ and $\log L^{\MN}$ by $\frac{\pi^{\st}_k}{1-\pi^{\st}_k}$, and change $b_k$ to $b_k+\log(1-\pi^{\st}_k)$ for $1\leq k\leq K$. Hence all the optimization and statistical guarantees presented in Proposition~\ref{lem:converg-MN} and Theorem~\ref{thm:mult_upp} still hold for this setting, as long as we change the condition on $\frac{n_k}{\pi^{\st}_kn_u}$ to the condition on $\frac{\pi^{\st}_k}{1-\pi^{\st}_k}$.
\subsection{Ordinal-PU Model}
The following lemma presents the log-likelihood functions for the observed data $\{(x_i,z_i)\}$ and for the full data $\{(x_i,y_i,z_i)\}$ under the ordinal-PU model and the single-training-set scenario.
\begin{lemma}\label{lem:ON-log-likelihoods-st}
 When the observed responses $z_i$ is generated according to~\eqref{eq:random-missing}, the log-likelihood function for the observed presence-only data $\{x_i,z_i\}_{i=1}^n$ under the ordinal-PU model is
 \begin{equation}\label{eq:ON-log-likelihood-st}
 \begin{split}
     &\log L^{\ON,\st}(\theta;\{(x_i,z_i)\}_{i=1}^n)\\
     =&\sum_{i=1}^n \left[\sum_{k=1}^K \ind{z_i=k}[\widetilde{f}(\log r(x_i,\theta)+\log(1-\pi^{\st}))]_k-\log\left(1+\sum_{k=1}^Ke^{[\widetilde{f}(\log r(x_i,\theta)+\log(1-\pi^{\st}))]_k}\right)\right],
 \end{split}
 \end{equation}
 where $\widetilde{f}$ is defined in Lemma~\ref{lem:MN-log-likelihoods-st} and $r:\bbR^p\times \bbR^{p+K}\rightarrow \bbR^{K}$ is defined in Lemma~\ref{lem:ON-log-likelihoods}.
 The log-likelihood function for the full data $\{x_i,y_i,z_i\}_{i=1}^n$ is
 \begin{equation}\label{eq:ON-full-log-likelihood-st}
 \begin{split}
     &\log L_f^{\ON,\st}(\theta;\{(x_i,y_i,z_i)\}_{i=1}^n)\\
     =&\sum_{i=1}^n
    \bigg[\sum_{k=1}^K \ind{y_i=k}(\log r_k(x_i,\theta)+\log(1-\pi^{\st}_k)-\log\left(1+\sum_{k=1}^Kr_k(x_i,\theta)\right)\\
    &+\sum_{k=1}^K\ind{y_i=z_i=k}\log\frac{\pi^{\st}_k}{1-\pi^{\st}_k}\bigg].
 \end{split}
 \end{equation}
\end{lemma}
\begin{proof}
Similarly to the proof of Lemma~\ref{lem:MN-log-likelihoods-st}, we first write out the joint probability mass for $(x,y,z)$ as follows:
\begin{align*}
    \bbP(y=j,z=j|x)=&\frac{\pi^{\st}_j r_j(x,\theta)}{1+\sum_{k=1}^K r_k(x,\theta)},\\
    \bbP(y=j,z=0|x)=&\frac{(1-\pi^{\st}_j) r_j(x,\theta)}{1+\sum_{k=1}^K r_k(x,\theta)},\\
    \bbP(y=0,z=0|x)=&\frac{1}{1+\sum_{k=1}^K r_k(x,\theta)}.
\end{align*}
The joint probability mass for $(x,z)$ would then become:
\begin{align*}
    \bbP(z=j|x)=&\frac{\pi^{\st}_j r_j(x,\theta)}{1+\sum_{k=1}^K r_k(x,\theta)},\\
    \bbP(z=0|x)=&\frac{1+\sum_{k=1}^K(1-\pi^{\st}_k) r_k(x,\theta)}{1+\sum_{k=1}^K r_k(x,\theta)}.
\end{align*}
Therefore, one can directly write out the log-likelihood functions as shown in Lemma~\ref{lem:ON-log-likelihoods-st}.
\end{proof}
Let $$\cF^{\ON,\st}_n(\theta)=-\frac{1}{n}\log L^{\ON,\st}(\theta;\{(x_i,z_i)\}_{i=1}^n)+P^{\ON}_{\lambda}(\theta_{1:p}),$$ then we can estimate $\theta^*$ by applying the proximal gradient descent algorithm on $\cF^{\ON,\st}_n(\theta)$. The regularized EM algorithm under this setting is summarized in Algorithm~\ref{alg:ON-EM-st}.

\begin{algorithm}[H]\label{alg:ON-EM-st}
\SetAlgoLined
Input: $\theta^0$ such that $\mathcal{F}^{\ON,\st}_n(\theta^0)\leq\min_{\theta_{1:p}=0_{p\times 1}}\mathcal{F}^{\ON,\st}_n(\theta)$\\
 \For{$m=0,1,\dots$}{
   E-step: calculate $\widehat{y}(\theta^m)\in\bbR^{n\times K}$ as follows, where each entry $[\widehat{y}(\theta^m)]_{ij}$ is an estimate for $\ind{y_i=j}$:\\
   \[\hat{y}_{ij}(\theta^m)= \begin{cases}
   \frac{(1-\pi^{\st}_j)r_j(x_i,\theta^m)}{1+\sum_{k=1}^K(1-\pi^{\st}_k)r_k(x_i,\theta^m)},  &\mathrm{if}\ z_i=0,\\
   1,&\mathrm{if}\ z_i=j,\\
   0,& \mathrm{else.}
   \end{cases}\]\\
   M step: obtain $\theta^{m+1}$ by\\
   \begin{equation*}
       \begin{split}
           \theta_j^{m+1}\in\argmin\limits_{\theta}&\frac{1}{n}\sum_{i=1}^n\Bigg[-\sum_{j=1}^K\widehat{y}_{ij}(\theta^m)\log r_j(x_i, \theta)\\
           &+\log\left(1+\sum\limits_{k=1}^{K}r_k(x_i,\theta)\right)\Bigg]+P^{\ON}_{\lambda}(\theta_{1:p}).
       \end{split}
   \end{equation*}
 }
 \caption{EM for the Ordinal-PU Model under the single-training-set scenario}
\end{algorithm}
\paragraph{Theoretical properties: }Using similar arguments as the proof of Proposition~\ref{lem:converg-ON}, we can still show the convergence of these two algorithms to stationary points of the optimization problem $\min_{\|\theta\|_2\leq C} \cF_n^{\ON,\st}(\theta)$ for some constant $C>0$. Although we do not have a rigorous statistical error bound for the stationary points under this setting, we conjecture that they would satisfy similar statistical properties from Theorem~\ref{thm:ordinal_upp}. A rigorous proof for this conjecture is left as future work.
\section{Numeric Details}
\paragraph{Controlling the true prevalence of different categories:} In our comparative simulation studies, we controlled the prevalence of different categories in an approximate way by setting the intercepts appropriately. In particular, for the experiments investigating the effect of the prevalence, we choose the intercepts to ensure that the intercept-only model have prevalence $\pi_0 \in \{0.1,\,0.2,\dots,0.7\}$ and $\pi_1 = \pi_2 = \frac{1-\pi_0}{2}$. Since all features $X_j$ are generated from a mean zero Gaussian distribution, we would expect that the prevalence of the full models are not too different from the prevalence of the intercept-only models. The prevalence of the full models are then estimated from the generated data and presented in the figures. While for the experiments investigating the effect of sampling ratio $\frac{n_u}{n}$, we ensure the intercept-only models have $30\%$ unlabeled data, $35\%$ samples of label $1$ and $2$. The resulting prevalence of the full multinomial model is $\pi_0=0.2584,\,\pi_1 = 0.3677,\,\pi_2 = 0.3739$, and that of the full ordinal model is $\pi_0=0.3753,\,\pi_1 = 0.2153,\,\pi_2 = 0.4095$. Hence the whole population is reasonably balanced.

 \section*{Acknowledgments}
 We would like to thank David Neiman for some preliminary empirical studies and real data exploration. GR and LZ were partially supported by NSF-DMS 1811767, NIH R01 GM131381-01. GR was also partially supported by NSF-DMS 1839338.

\bibliographystyle{plainnat} 
\bibliography{reference} 
\end{document}